\documentclass[ejs]{imsart}

\RequirePackage{amsthm,amsmath,amsfonts,amssymb}
\RequirePackage[authoryear]{natbib}%% uncomment this for author-year citations
\RequirePackage[colorlinks,citecolor=blue,urlcolor=blue]{hyperref}
\RequirePackage{graphicx}
\startlocaldefs
\usepackage[utf8]{inputenc}
\usepackage{psfrag,epsf}
\usepackage{enumerate}
\usepackage{url} 
\usepackage{algorithm,algorithmic}
\usepackage{array}
\usepackage{multirow}
\usepackage{IEEEtrantools}
\usepackage{mathtools}
\usepackage{comment}
\usepackage{bbm}
\usepackage{multirow}
\usepackage{diagbox}
\usepackage{booktabs}
\usepackage{bigints}
\usepackage{soul}
\usepackage{siunitx}
\theoremstyle{plain}
\newtheorem{thm}{Theorem}
\numberwithin{thm}{section}
\newtheorem{lemma}{Lemma}
\numberwithin{lemma}{section}
\newtheorem{cor}{Corollary}
\numberwithin{cor}{section}
\newtheorem{prop}{Proposition}
\numberwithin{prop}{section}

\theoremstyle{definition}

\newtheorem{assp}{Assumption}
\newtheorem{remark}{Remark}
\numberwithin{remark}{section}

\def\hmu{\hat{\mu}}
\def\bH{\mathbb{H}}
\def\bc{\boldsymbol{c}}
\def\bZ{\boldsymbol{Z}}
\def\fU{\mathfrak{U}}

\newcommand{\E}{\mathbb{E}}

\newcommand{\R}{\mathbb{R}}
\newcommand{\N}{\mathbb{N}}

\newcommand{\U}{\mathcal{U}}
\newcommand{\D}{\mathcal{D}}
\newcommand{\Sb}{\mathbb{S}}
\newcommand{\Sc}{\mathcal{S}}
\newcommand{\T}{\mathcal{T}}
\newcommand{\G}{\mathcal{G}}
\newcommand{\Lp}{\mathcal{L}}
\newcommand{\p}{\mathbb{P}}
\newcommand{\err}{\varepsilon}
\newcommand{\HS}{\mathcal{H}}

\newcommand{\amin}{\text{argmin}}
\newcommand{\innerproduct}[2]{\left\langle #1, #2 \right\rangle}
\newcommand{\mbf}{\mathbf}
\newcommand{\eps}{\varepsilon}
\newcommand{\bbi}{\mathbbm{1}}
\newcommand{\Cov}{\text{Cov}}
\newcommand{\Var}{\text{Var}}
\newcommand{\brac}[1]{\left \{ #1 \right \}}
\newcommand{\df}[1]{#1^{(1)}}
\newcommand{\C}{\mathcal{C}}
\DeclarePairedDelimiterX{\sbrac}[1]{[}{]}{#1}
\DeclarePairedDelimiterX{\pbrac}[1]{(}{)}{#1}
\DeclarePairedDelimiterX{\norm}[1]{\lVert}{\rVert}{#1}
\DeclarePairedDelimiterX{\abs}[1]{\lvert}{\rvert}{#1}
\DeclarePairedDelimiter{\ceil}{\lceil}{\rceil}

\newcommand{\proofpart}[2]{%
  \par
  \addvspace{\medskipamount}%
  \noindent \textbf{Part #1:} #2\par\nobreak
  \addvspace{\smallskipamount}%
  \@afterheading
}
\newcommand\numberthis{\addtocounter{equation}{1}\tag{\theequation}}

\endlocaldefs

\begin{document}
\begin{frontmatter}
\title{Two-sample inference for sparse functional data}
%\title{A sample article title with some additional note\thanksref{t1}}
\runtitle{Two-sample inference}
%\thankstext{T1}{A sample additional note to the title.}

\begin{aug}
%%%%%%%%%%%%%%%%%%%%%%%%%%%%%%%%%%%%%%%%%%%%%%%
%% Only one address is permitted per author. %%
%% Only division, organization and e-mail is %%
%% included in the address.                  %%
%% Additional information can be included in %%
\author{\fnms{Chi}~\snm{Zhang}\ead[label=e1]{c378zhan@uwaterloo.ca}},
\author{\fnms{Peijun}~\snm{Sang}\ead[label=e2]{peijun.sang@uwaterloo.ca}}
\and
\author{\fnms{Yingli}~\snm{Qin}\ead[label=e3]{yingli.qin@uwaterloo.ca}}
%%%%%%%%%%%%%%%%%%%%%%%%%%%%%%%%%%%%%%%%%%%%%%
%% Addresses                                %%
%%%%%%%%%%%%%%%%%%%%%%%%%%%%%%%%%%%%%%%%%%%%%%
%\address[A]{Department,
%University or Company Name\printead[presep={,\ }]{e1}}

%\address[B]{Department,
%University or Company Name\printead[presep={,\ }]{e2,e3}}
%\runauthor{F. Author et al.}

\address{Department of Statistics and Actuarial Science\\
University of Waterloo, Ontario, Canada\\ \printead[presep={\ }]{e1,e2,e3}}
\runauthor{C. Zhang et al.}
\end{aug}

\begin{abstract}
We propose a novel test procedure for comparing mean functions across two groups within the reproducing kernel Hilbert space (RKHS) framework. Our proposed method is adept at handling sparsely and irregularly sampled functional data when observation times are random for each subject. Conventional approaches, which are built upon functional principal components analysis, usually assume a homogeneous covariance structure across groups. Nonetheless, justifying this assumption in real-world scenarios can be challenging. To eliminate the need for a homogeneous covariance structure, we first develop a linear approximation for the mean estimator under the RKHS framework; this approximation is a sum of i.i.d. random elements, which naturally leads to the desirable pointwise limiting distributions. Moreover, we establish weak convergence for the mean estimator, allowing us to construct a test statistic for the mean difference. Our method is easily implementable and outperforms some conventional tests in controlling type I errors across various settings. We demonstrate the finite sample performance of our approach through extensive simulations and two real-world applications.
\end{abstract}

%\begin{keyword}[class=MSC]
%\kwd[Primary ]{00X00}
%\kwd{00X00}
%\kwd[; secondary ]{00X00}
%\end{keyword}


\end{frontmatter}
%%%%%%%%%%%%%%%%%%%%%%%%%%%%%%%%%%%%%%%%%%%%%%
%% Please use \tableofcontents for articles %%
%% with 50 pages and more                   %%
%%%%%%%%%%%%%%%%%%%%%%%%%%%%%%%%%%%%%%%%%%%%%%
%\tableofcontents

\section{Introduction}
Due to advances in data collection techniques, functional data analysis (FDA) has gained increasing prominence in modern data analysis. FDA offers a non-parametric framework, especially when repeated measurements from each subject are considered as discrete observations from realizations of continuous random functions. A comprehensive introduction to FDA can be found in monographs such as \cite{ramsayFDA2005}, \cite{Hsing2015FDA} and \cite{KokoszkaFDA2017}.

In the context of FDA, there are primarily two types of functional data that have been extensively studied: dense functional data and sparse functional data. {While there is no universally accepted standard for differentiating between these two types, relevant discussions can be found in \cite{TCaiAOS11} and \cite{zhang2016sparse}}. Generally speaking, dense functional data involve a large number of densely spaced observations per subject, which are often obtained through automated high-frequency data collection instruments. Sparse functional data consist of only a few irregularly spaced observations per subject, which are commonly encountered in longitudinal studies. Estimating mean functions for both types has been studied extensively. For dense functional data, \cite{RiceAndSilverman1991JRSSB} proposed recovering trajectories for each subject via smoothing splines, then deriving the sample mean function based on the recovered trajectories. {For sparse functional data, \cite{FY2005} advocated aggregating observations across subjects and then using local linear smoothing to estimate the mean function.} Motivated by this idea, \cite{TCaiAOS11} employed smoothing splines for mean function estimation from aggregated data. 

This paper focuses on statistical inference for two-sample mean functions. Well-established methods are available in the context of dense functional data. Under the assumption that the number of observations grows faster than the number of subjects, the approximation error incurred by the pre-processing  becomes negligible in subsequent analysis \citep{JinTingZhangAOS2007}. This is the primary reason why various statistical inference approaches exist for two-sample mean functions in the context of dense functional data. {Specifically, the methods proposed by \cite{JinTingZhangJSTP2010} and \cite{JinTingZhangSJS2014}, which were initially intended for fully observed data, suggest that utilizing pre-smoothed individual trajectories in test statistics is asymptotically equivalent to using fully observed data.}

However, conducting inference for two-sample mean functions in the context of sparse functional data remains relatively unexplored. In cases with only a few irregularly spaced observations per subject, recovering trajectories from such sparse data could result in significant estimation errors, which cannot be overlooked in subsequent analysis. Consequently, developing a statistically sound procedure for mean function inference becomes more challenging. Existing methods, to our knowledge, either consider a polynomial-type null hypothesis \citep{LRT2014SJS} or assume a homogeneous covariance structure for two groups \citep{JRSSC2016, ejs2sample}.

Testing if a mean function is a polynomial can be overly restrictive when prior knowledge about the mean function is lacking. Additionally, determining the polynomial order can be challenging. On the other hand, the homogeneous covariance assumption allows us to project individual trajectories onto basis functions derived from pooled data. These basis functions are eigenfunctions of the common covariance function, known as functional principal components in the FDA literature. Without assuming a common covariance structure, trajectories from different groups are projected onto different directions, making comparisons of mean functions based on projections misleading. However, verifying this assumption can be quite challenging in itself. For example, \cite{pigoli_distances_2014} proposed an infinite-dimensional version of the Procrustes size-and-shape distance for covariance operators. However, the theoretical properties of this test statistic are not yet fully established. As a result, the calculation of p-values using this statistic depends on re-sampling techniques, such as permutations.

In this paper, we propose a novel procedure for testing the mean difference based on sparse functional data within the RKHS framework. The proposed method does not entail a homogeneous covariance for the two groups of functional data. We begin by aggregating data within each group, enabling us to borrow information from data across all subjects to obtain a group-wise mean function estimator via smoothing splines. We then develop a {\it functional Bahadur representation} for this estimator, which is crucial for establishing the asymptotic properties of the mean estimator. This representation was first developed by \cite{zfsannals13}, and we adopt the same term throughout the paper. Specifically, it allows us to represent the estimation error as a summation of i.i.d. random elements plus higher-order terms that are negligible in the asymptotic analysis.
With the aid of this representation, we establish pointwise asymptotic distributions and weak convergence for the mean estimator within each group. Under the independence assumption of these two groups, we subsequently develop pointwise and global tests for the mean difference between the two groups. To evaluate finite-sample performance of our proposed method, we conduct extensive simulation studies to investigate its performance in estimating the mean function and testing the mean difference between the two groups. 

The remainder of this article is organized as follows. In Section \ref{sec:model_and_est}, we outline the construction of a mean function estimator via smoothing splines for sparse functional data. In Section \ref{sec:theory}, we develop asymptotic results for the mean estimator, and design a bootstrap-based algorithm to implement the testing procedure with statistical guarantees. Detailed proofs are relegated to the Appendix. We showcase finite-sample performance of the proposed method under various settings in Section \ref{sec:Simulation_studies}.  We consider two real applications of the proposed method in Section \ref{sec:real_data_examples}: diffusion tensor imaging data and Beijing air quality data. Finally, we conclude the paper in Section \ref{sec:conclusion}.

\section{Model and Estimation}\label{sec:model_and_est}
In Section \ref{sec:setup}, we introduce our model, formalize the hypothesis of detecting mean differences in two groups of functional data, and review existing methods. Subsequently, in Section \ref{sec:mean_est}, we propose a smoothing-spline-based estimator for the mean function within each group under the RKHS framework, leading to the immediate availability of a mean difference estimator.

\subsection{Problem Setup} \label{sec:setup}
In this paper, we consider a two-sample problem involving data generated from the following model: 
\begin{equation}\label{eq:model}
    Y_{gij} = X_{gi}(T_{gij}) + \err_{gij}, \quad T_{gij} \in \T,
\end{equation}
where $g=1,2$, $i=1, 2, \ldots, n_{g}$ and $j=1, 2, \ldots, N_{gi}$. Here $g$ is the group index.  Without loss of generality, let $\T = [0, 1]$. 

\begin{enumerate}
    \item[(i)] Let $Y_{gij}$ denote the $j$th noisy observation of the $i$th random function $X_{gi}$;
    \item[(ii)] Let $T_{gij}$'s denote i.i.d. observation times with a common group-wise probability density defined over $\T$; 
    \item[(iii)] Let $\err_{gij}$'s denote i.i.d. random errors with mean zero and variance $\sigma_{g, \err}^2$;
    \item[(iv)] Let $N_{gi}$'s denote numbers of observations, assumed to be i.i.d. within group $g$.
\end{enumerate}
We further assume that every $\err_{gij}$ is independent of the random process $X_g$ and the observation time $T_{gij}$; $N_{gi}$'s are independent of all other random objects in the model. Notice that $N_{gi}$'s characterize sparseness of the observations in group $g$. This model structure has been widely adopted in the FDA literature, including \cite{FY2005}, \cite{li2010uniform}, \cite{zhang2016sparse}, \cite{ejs2sample} and references therein. 

For $g = 1, 2$, we assume $\E\{X_{g}(t)\} = \mu_g(t)$ for all $t \in \T$. In this paper, we focus on two types of hypothesis testing problems. The first is commonly referred to as pointwise testing, where the null hypothesis is defined for every specific time point $t^* \in \T$:
\begin{equation}\label{eq:pt-hypo}
    H_0: \mu_1(t^*) = \mu_2(t^*), ~\text{ VS }~ H_1: \mu_1(t^*) \neq \mu_2(t^*).
\end{equation}

The second type is global testing, which formulates the testing problem for the entire domain $\T$:
\begin{equation}\label{eq:hypo}
    H_0: \mu_1(t) = \mu_2(t), ~ \forall t \in \T ~\text{ VS }~ H_1: \mu_1(t) \neq \mu_2(t) \text{, for some } t \in \T.
\end{equation}
{The hypothesis testing problem described in \eqref{eq:pt-hypo} is focused on testing mean differences at specific time points. However, the global test may not align with the pointwise test even for all $t\in\T$, at the same significance level. Furthermore, in contrast to \eqref{eq:pt-hypo}, the testing problem in \eqref{eq:hypo} is more complex due to the infinite-dimensional nature of the parameter of interest.} Addressing this problem with a finite number of observations entails dimension reduction techniques. Functional principal components analysis (fPCA) plays a pivotal role in this regard and has been extensively explored in the FDA literature. In classical fPCA, the covariance function $\Cov(X_{gi}(t), X_{gi}(s)) \triangleq c_g(t,s)$ can be orthogonally expanded using a sequence of non-increasing and non-negative eigenvalues $\theta_{gk}$ and their corresponding eigenfunctions $\phi_{gk}$: $c_g(t,s) = \sum_{k} \theta_{gk}\phi_{gk}(t)\phi_{gk}(s)$. Then, the Karhunen–Lo\'{e}ve expansion yields
\begin{equation}\label{eq:KL_expansion}
    X_{gi}(t) = \mu_g(t) + \sum_{k=1}^{\infty}\xi_{gik}\phi_{gk}(t),
\end{equation}
where $\xi_{gik} = \int_{\mathcal{T}}\{X_{gi}(t) - \mu_g(t)\}\phi_{gk}(t) dt$ is the $k$th fPC score satisfying:
\begin{enumerate}
    \item[(i)] $\E\xi_{gik} = 0$;
    \item[(ii)] $\Cov(\xi_{gik_1}, \xi_{gik_2}) = \theta_{gk}\delta_{k_1k_2}$, which equals $\theta_{gk}$ if $k_1 = k_2$, and $0$ otherwise.
\end{enumerate}
Here $k, k_1, k_2 \geq 1$, $g=1,2$, and $i = 1, \ldots n_g$.

{Assuming $c_1(t,s) = c_2(t,s)$, the eigenfunctions for both groups are identical. Under the null hypothesis in \eqref{eq:hypo}, the difference between the groups of random functions arises from the randomness of fPC scores as described in \eqref{eq:KL_expansion}.
In practice, we usually consider a truncated version of \eqref{eq:KL_expansion}, by retaining only the first $p$ eigenfunctions estimated from $Y_{gij}$'s within each group. Interested readers can refer to \cite{ramsayFDA2005} and \cite{FY2005} for further details.  Here $p$ can be determined via cross-validation or criteria such as the Bayesian information \citep{FY2005}. Such a truncation converts the infinite-dimensional inference problem to a finite-dimensional one. Following this strategy, \cite{ejs2sample} constructed a $\chi^2_{p}$ test statistic by assuming joint normality of scores and random errors. To compare the distributions of these groups of functional data, \cite{JRSSC2016} proposed the Anderson–Darling test for comparing the distribution of the two groups of fPC scores. 
It should be noted that the homogeneous covariance assumption is vital for fPCA-based methods to perform dimension reduction, although verifying it in real-world scenarios can be challenging.}

\subsection{Estimation of Mean Functions} \label{sec:mean_est}
In the FDA literature, it is common to assume that discrete observations are taken from sample paths that are realizations of a random process, typically with certain smoothness properties \citep{Wang2016FDAReview}. In this paper, we assume that almost surely the sample path of $X_g$ for $g = 1, 2$ resides in an $m$th order Sobolev space, which is defined as
\begin{equation*}\label{def:SobolevSpace}
\begin{aligned}
    \HS := H^{m}(\T) = \{f: \T \to \R \mid f, f^{(1)}, f^{(2)}, &\ldots, f^{(m-1)} \text{ are absolutely} \\
    &\text{ continuous and }f^{(m)} \in \Lp_2(\T)\},
\end{aligned}
\end{equation*}
where $f^{(v)}$ denotes the $v$th order derivative of $f$ and $$\Lp_2(\T) = \left\{f: \T \to \R ~\Big | ~\int_\T f^2(t) dt < \infty \right\},$$ denotes the collection of square integrable functions on $\T$. This assumption also appears in \cite{TCaiAOS11}. Furthermore, as demonstrated in the same reference, the mean function $\mu_g$ belongs to $\HS$, as indicated by the following inequality:
\begin{equation*}
    \int_{\T} \{\mu_g^{(m)}(t)\}^2dt \leq \E\left[\int_{\T} \{X_{g}^{(m)}(t)\}^2dt\right].
\end{equation*}

Choosing an appropriate mean function estimator hinges on the subject sampling frequency within the same group. For fully observed data $X_{g1}, \ldots, X_{gn_g}$, the sample mean is a natural choice. However, in practice, observations are often discretely collected over fixed or random grids. {Estimating the mean function for dense functional data commonly involves a two-step method.   Firstly, a pre-smoothing procedure is employed to reconstruct individual trajectories denoted as $\tilde{X}_{gi}$ via non-parametric techniques like smoothing splines \citep{RiceAndSilverman1991JRSSB}, local polynomial regression \citep{JinTingZhangAOS2007}, and B-splines \citep{Cardot2000JNS}. The subsequent analysis then utilizes $\tilde{X}_{gi}$ in place of $X_{gi}$ in the sample mean estimator.}

The pre-smoothing approach becomes inappropriate with limited observations per subject due to insufficient information for accurate trajectory recovery. In sparse settings, \cite{FY2005} suggested aggregating data across all subjects to estimate the mean function through local linear regression. They established consistency and uniform convergence rates for both mean and covariance estimators. Furthermore, \cite{zhang2016sparse} provided convergence rates and pointwise limiting distributions for the mean estimator using local linear regression under various weighting schemes. Additionally, they presented a systematic framework for differentiating dense functional data from sparse functional data, which is based on the ratio of the average number of observations per subject to the total number of subjects.

Inspired by the aggregation idea proposed in \cite{FY2005}, we introduce a novel mean function estimator within the RKHS framework. Specifically, we compute the smoothing splines estimate for $\mu_g$ in each group $g$ by minimizing a loss function $\ell(f, \lambda_g)$ defined as follows:
\begin{align}
    \hat{\mu}_{\lambda_g} &= \underset{f \in \HS}{\amin}~ \ell(f, \lambda_g) \nonumber \\
    &:=  \underset{f \in \HS}{\amin} \Bigg[\frac{1}{2M_g}\sum_{i=1}^{n_g}\sum_{j=1}^{N_{gi}}\big\{Y_{gij} -  f(T_{gij})\big\}^2 + \frac{\lambda_g}{2} J(f, f)\Bigg], \label{eq:mean_optimization}
\end{align}
where $M_g = \sum_{i=1}^{n_g}N_{gi}$ denotes the total number of observations across $n_g$ subjects in group $g$, $J(f, f) = \int_\T \{f^{(m)}(t)\}^2dt$ is the roughness penalty, and $\lambda_g > 0$ denotes the smoothing parameter controlling the trade-off between fidelity to the data and roughness of the estimated mean function. {It is important to point out that we assign the same weight to each observation in \eqref{eq:mean_optimization}, while \cite{TCaiAOS11} assigned the same weight to each subject. They showed that their proposed mean estimator attained the minimax convergence rate, but did not develop any inferential tools for the mean function.}

When $\HS$ is equipped with the following (squared) norm
\begin{equation}\label{eq:norm_representer_thm}
    \norm{f}_{R}^2 = \sum_{\nu = 0}^{m-1}\brac{\int_{\T}f^{(\nu)}(t)dt}^2 + \int_{\T}\brac{f^{(m)}(t)}^2dt,
\end{equation}
$\HS$ becomes an RKHS \citep[Chapter~2.3]{gu2013smoothing}. 
By the well-known representer theorem for an RKHS \citep{wahba1990spline}, the mean function estimator, i.e., the solution to the optimization problem \eqref{eq:mean_optimization}, has a finite-dimensional representation. This result is summarized in the following proposition. Let $\HS_0$ denote the null space of $J$, that is, $\HS_0 = \{f \in \HS\mid J(f, f) = 0\}$, and $\HS_1$ the orthogonal complement of $\HS_0$ in $\HS$.
Suppose the basis functions of $\HS_0$ are $\{\phi_1, \phi_2, \ldots, \phi_m\}$ and $R(s, t)$ denotes the reproducing kernel of $\HS_1$ associated with the norm defined in \eqref{eq:norm_representer_thm}.
\begin{prop}\label{prop:representer_thm}
The solution to the minimization problem \eqref{eq:mean_optimization} can be expressed as 
\begin{equation}\label{eq:representer_thm}
    \hat{\mu}_{\lambda_g}(t) = \sum_{k=1}^{m}d_{gk}\phi_k(t) + \sum_{i=1}^{n}\sum_{j=1}^{N_{gi}}c_{gij}R(t_{gij}, t),
\end{equation}
for some coefficients $d_{g1}, d_{g2}, \ldots, d_{gk}$ and $c_{g11}, c_{g12}, \ldots, c_{gn_gN_{gn}}$.
\end{prop}
A suitable choice for $\phi_k$ and $R(\cdot, \cdot)$ is $\phi_k = t^{k - 1}, k=1, 2, \ldots, m$ and $$R(s, t) = \frac{1}{(m!)^2}B_m(t)B_m(s) - \frac{1}{(2m)!}B_{2m}(\abs{t-s}),$$ respectively. Here $B_{l}$ stands for the $l$th Bernoulli polynomials. For further details, please refer to Chapter 2.3 of \cite{gu2013smoothing}.

The choice of $\lambda_g$ has a great impact on the performance of mean function estimation and inference for the mean function difference. A commonly used data-driven method for choosing $\lambda$ is through \textit{generalized cross-validation} (GCV). 
Since it can be shown that $c_{gij}$ and $d_{gk}$ linearly depend on $\mbf{y}_g$ in \eqref{eq:representer_thm}, we express $\hat{\boldsymbol{\mu}}_{\lambda_g} = S(\lambda_g)\mbf{y}_g$, where 
$$\hat{\boldsymbol{\mu}}_{\lambda_g} = (\hat{\mu}_{\lambda_g}(t_{g11}), \hat{\mu}_{\lambda_g}(t_{g12}), \ldots, \hat{\mu}_{\lambda_g}(t_{gn_gN_{n_g}}))^{'},$$ and  
$$\mbf{y}_g = (y_{g11}, y_{g12}, \ldots, y_{gn_gN_g})^{'}.$$ 
Here $S(\lambda_g)$ is the smoothing matrix. To determine $\lambda_g$, we minimize the the GCV score defined as 
\begin{equation}\label{eq:GCV_score}
    \text{GCV}(\lambda_g) = \frac{n^{-1}\|\hat{\boldsymbol{\mu}}_{\lambda_g} - \mbf{y}_g\|_2^2}{[n^{-1}\text{tr}\{I - S(\lambda_g)\}]^2},
\end{equation}
where $\|\cdot\|_2$ denotes the vector Euclidean norm. The choice of $\lambda_g$ is crucial for this procedure, and more details can be found in  \cite{GCV1979} and Chapter 3.2 of \cite{gu2013smoothing}. 

\section{Theoretical Properties} \label{sec:theory}

Before presenting the main results, we establish necessary notations and assumptions in Section \ref{subsec:assumptions} that define the data-generating process for sparse functional data and ensure estimation consistency. {Section \ref{sec:RKHS} introduces a new inner product that makes $\HS$ an RKHS with desirable properties. With the aid of these properties, we establish the functional Bahadur representation (FBR) in Section \ref{sec:FBR}. This representation leads to the pointwise limiting distributions (Section \ref{subsec:pointwise}) and weak convergence (Section \ref{subsec:global}) for the mean estimator of each group; they are then extended to the estimated mean difference. Subscripts are omitted when assumptions and properties apply equally to both groups.} 

\subsection{General Assumptions}\label{subsec:assumptions}
\begin{assp}\label{assp:data_gm}
For $i= 1, 2, \ldots, n, j = 1, 2, \ldots, N_{i}$, $(Y_{ij}, T_{ij})$ generated by model (\ref{eq:model}), are copies of $(Y, T)$ with a bounded joint density function $f(y, t)$. The marginal density function of $T$, denoted as $\pi(t)$, is bounded away from both infinity and zero.
\end{assp}

\begin{assp}\label{assp:measurement_error}
$\E(\eps) = 0$ and $\text{Var}(\eps) = \sigma_{\eps}^{2}< \infty$. Furthermore, $\eps_{ij}, j = 1, \dots, N_i, i = 1, \ldots, n$ are i.i.d. copies of $\eps$ and are independent of $T_{ij}$.
\end{assp}

\begin{assp}\label{assp:num_of_obs}
The number of observations for the $i$th subject, $N_{i}, i = 1, \ldots, n$, are i.i.d. copies of a positive integer-valued random variable $N$ satisfying $\E(N) = \mu_N < \infty$ and $\Var(N) = \sigma^2_N< \infty$. Furthermore, $\p(N > 1) > 0$. Additionally, for $i = 1, \ldots, n$, $N_i$ is independent of $T_{ij}$ and $\eps_{ij}$.
\end{assp}

\begin{assp}\label{assp:smooth_of_X}
    The sample path of the stochastic process $X$ in model \eqref{eq:model} has an $m$th-order derivative, which is denoted by $X^{(m)}$, almost surely. Moreover, $X^{(m)} \in \mathcal{L}_2(\T)$, i.e., $\E[\int_{\T} \{X_{g}^{(m)}(t)\}^2dt] < \infty$. 
\end{assp}
These assumptions are common in the FDA literature. For example, Assumptions \ref{assp:data_gm}-\ref{assp:num_of_obs} were adopted by \cite{FY2005} for fPCA in the context of sparse functional data. Assumption \ref{assp:data_gm} ensures that observation $Y_{ij}$'s are randomly and evenly spread on the domain $\T$ for each subject, facilitating reasonable estimators of the mean function when observations are aggregated. Assumption 2 ensures that the second moment of the error term is finite. The assumption of independent errors within the same subject can be relaxed. For instance, one could assume $\eps(t)$ is a square-integrable and weakly stationary process. In such a case, an additional term accounting for the within-subject variability would appear in the variance expressions provided in Theorems \ref{thm:ptwise_limiting_dist_biased} and \ref{thm:weak_conv}. Assumption \ref{assp:num_of_obs} characterizes the 
sparseness of the observations. To see this, let $M = \sum_{i=1}^{n}N_i$. Then we have $$\E(M) = \E\left(\sum_{i =1}^{n} N_{i}\right) = n\mu_N = O(n)\textrm{~as~}\mu_N < \infty.$$ Further, we can show that $M = O_p(n)$ by the weak law of large numbers. This indicates that the total number of observations per group scales with the sample size, meaning the number of observations per subject does not increase with respect to the sample size. Assumption \ref{assp:smooth_of_X} is introduced for theoretical convenience, which ensures continuity and boundedness of sample paths and the covariance function in the compact domain. This assumption with $m = 2$ was considered in \cite{zhu2014structured} and \cite{kong2016partially}. Let $\mbf{Y}_{i} = (Y_{i1}, Y_{i2}, \ldots, Y_{iN_i})^{'}$ and $\mbf{T}_{i} = (T_{i1}, T_{i2}, \ldots, T_{iN_i})^{'}$. Assumptions \ref{assp:measurement_error} and \ref{assp:num_of_obs} indicate $\{(\mbf{Y}_{i}, \mbf{T}_{i}, N_{i})\}_{i=1}^{n}$ are i.i.d in model \eqref{eq:model}. {The independence of subjects allows for the application of the central limit theorem, to establish asymptotic results for the mean estimator.}

\subsection{Reproducing Kernel Hilbert Space} \label{sec:RKHS}
To test hypotheses in \eqref{eq:pt-hypo} and \eqref{eq:hypo}, we develop the FBR for a first-order approximation to the mean estimator. The FBR entails a properly defined RKHS, where the mean estimator resides. 
Following \cite{zfsannals13}, for $m > 1/2$, we consider $\HS = H^{m}(\T)$ as an RKHS, endowed with the inner product:
\begin{equation}\label{eq:inner_prod}
    \innerproduct{f}{g} = \E\left\{f(T)g(T)\right\} + \lambda J(f, g).
\end{equation}
Henceforth, $\norm{f}$ represents the norm induced by \eqref{eq:inner_prod} for any $f \in \HS$, and $K$ stands for the corresponding reproducing kernel of $\HS$. Consequently, $K$, serving as a function from $\T \times \T$ to $\R$, satisfies the following: 
\begin{equation*}\label{eq:reproducing_property}
   \innerproduct{K_t}{f} = f(t) \text{  for any } f \in \HS~\text{and}~t \in \T,
\end{equation*}
 where $K_t(\cdot) = K(t, \cdot) \in \HS$. 
Moreover, we introduce a positive self-adjoint operator $W_\lambda: \HS \to \HS$ satisfying \begin{center}
    $\innerproduct{W_\lambda f}{g} = \lambda J(f,g)$ for any $f, g \in \HS$. 
\end{center}
Define $V(f, g)$ as $\E\{f(T)g(T)\}$; thus, the inner product $\innerproduct{f}{g}$ defined in \eqref{eq:inner_prod} becomes the sum of $V(f, g)$ and $\innerproduct{W_\lambda f}{g}$. It is evident that if $T$ is uniformly distributed over $[0, 1]$, then $V(f, f)$ simplifies to the squared $L_2$ norm.

Suppose $\{a_\nu\}_{\nu=1}^{\infty}$ and $\{b_\nu\}_{\nu=1}^{\infty}$ are two positive sequences. Denote $a_\nu \asymp b_\nu$ if 
\begin{center}
    $0 < \underline\lim_{\nu \to \infty} a_\nu/b_\nu \leq \overline\lim_{\nu \to \infty} a_\nu/b_\nu < \infty$.
\end{center}
Denote the sup-norm of $g\in \HS$ and the $L_2$ norm of $f \in \mathcal{L}_2(\T)$ as  
$$\norm{g}_{\sup} = \sup_{t\in \T}\abs{g(t)} \textrm{~and~} \|f\|_2 = \left\{\int_{t\in \T} f^2(t) dt\right\}^{1/2},$$
respectively. Finally, let $\sum_{v} = \sum_{v = 1}^{\infty}$ for convenience.

\begin{assp}\label{assp:Fourier_expansion}
There exists a sequence of eigenfunctions $h_v \in \HS$ such that 
$\sup_{v \in \N_+}\norm{h_v}_{\sup} < \infty$ and a non-decreasing sequence of eigenvalues $\gamma_v \asymp v^{2m}$ that satisfy
\begin{equation*}
    V(h_u, h_v) = \delta_{uv} ~~\text{and}~~J(h_u, h_v) = \gamma_u\delta_{uv}, \hspace{3em}u, v \in \N.
\end{equation*}
Particularly, every $g \in \HS$ admits a Fourier expansion $g = \sum_{v} V(g, h_v)h_v$. 
\end{assp}

Assumption \ref{assp:Fourier_expansion} is pivotal in establishing the asymptotic results, as it establishes the connection between the two bilinear forms $V$ and $J$. Such assumptions are widely adopted in the FDA literature to develop consistency and the minimax convergence rate of estimators; see \cite{zfsannals13} and \cite{sun2018optimal} for example. 
This assumption leads to an analytic expression for any function in $\HS$, and plays an important role in determining upper bounds for several key functions. The following two lemmas follow immediately from Assumption \ref{assp:Fourier_expansion} and are crucial for subsequent theoretical developments; for example, see Remark \ref{remark:unbiased_remark}.

\begin{lemma}\label{prop:func_expression_by_h}
Under Assumption \ref{assp:Fourier_expansion}, for any $g \in \HS$, we have $$\norm{g}^2 = \sum_{u}\abs{V(g, h_u)}^2(1 + \lambda\gamma_u),~K_t(\cdot) = \sum_{v} \frac{h_v(t)}{1+\lambda\gamma_v}h_v(\cdot),$$
and $W_\lambda h_v(\cdot) = \lambda \gamma_v/(1 + \lambda \gamma_v)h_v(\cdot).$
\end{lemma}

\begin{lemma}\label{prop:K_and_Wf_bound}
Under Assumption \ref{assp:Fourier_expansion}, $\norm{W_{\lambda}f}^2 \leq \lambda J(f, f)$ and $\norm{K_t}^2 \leq C_K^2 h^{-1}$, for any $f \in \HS$ and $t \in \T$, where $C_K$ is a universal constant that is not associated with $t$.
\end{lemma}

We present the following lemma that confirms the existence of the eigen-system within $\HS$ as outlined in Assumption \ref{assp:Fourier_expansion}. A similar result can be found in \cite{zfsannals13}.

\begin{lemma}\label{lem:exist_of_eigensys}
Suppose Assumption \ref{assp:data_gm} holds, and the density function of $T$, $\pi(t)$, has up to $(2m-1)$th order continuous derivatives on $\T$. 
Then the eigenvalues $\gamma_v$ and the corresponding eigenfunctions $h_v$, solved from the following ordinary differential equation systems,
\begin{gather*}
    (-1)^{m}h_v^{2m}(\cdot) = \gamma_v\pi(\cdot)h_v(\cdot), \\
    h_v^{j}(0) = h_v^{j}(1) =0,~ j=m, m+1, \ldots, 2m-1,
\end{gather*}
satisfy Assumption \ref{assp:Fourier_expansion} if we normalize $h_v$ such that $V(h_v, h_v) = 1$.
\end{lemma}

Finally, we derive the first and second-order Fr\'{e}chet derivatives of $\ell(f, \lambda)$ in (\ref{eq:mean_optimization}) with respect to $f$ as detailed below:
\begin{align}
    D\ell(f, \lambda)\Delta g &= -\frac{1}{M}\sum_{i=1}^{n}\sum_{j=1}^{N_i}\left\{Y_{ij} - f(T_{ij})\right\}\innerproduct{K_{T_{ij}}}{\Delta g} + \innerproduct{W_\lambda f}{\Delta g} \nonumber \\
    &= \innerproduct{S_{M}(f) + W_\lambda f}{\Delta g} \nonumber\\
    &= \innerproduct{S_{M, \lambda}(f)}{\Delta g} \nonumber,
\end{align}
and
\begin{align}
    D^2 \ell(f, \lambda)\Delta g \Delta h &= \frac{1}{M}\sum_{i=1}^{n}\sum_{j=1}^{N_i}\innerproduct{K_{T_{ij}}}{\Delta g}\innerproduct{K_{T_{ij}}}{\Delta h} + \innerproduct{W_\lambda \Delta g}{\Delta h} \nonumber \\
    &\triangleq \innerproduct{DS_{M, \lambda}(f)\Delta g}{\Delta h} \nonumber,
\end{align}
where 
\begin{equation} \label{eq-SM}
\begin{aligned}
&S_{M, \lambda}(f) = -\frac{1}{M}\sum_{i=1}^{n}\sum_{j=1}^{N_i}\{Y_{ij} - f(T_{ij})\}K_{T_{ij}} + W_\lambda f \triangleq S_{M}(f) + W_\lambda f, \\
&DS_{M, \lambda}(f)\Delta g = \frac{1}{M}\sum_{i=1}^{n}\sum_{j=1}^{N_i}\Delta g(T_{ij})K_{T_{ij}} + W_\lambda\Delta g.
\end{aligned}
\end{equation}

\subsection{Functional Bahadur Representation} \label{sec:FBR}
In this section, we present the FBR for $\hat{\mu}_{\lambda}$, the solution to (\ref{eq:mean_optimization}). {For ease of notation, let $h = \lambda^{1/2m}$, where $\lambda$ is the smoothing parameter in \eqref{eq:mean_optimization}.} We define the set  $$\mathcal{G}=\left\{g \in \mathcal{H}:\|g\|_{\text {sup }} \leq 1, J(g, g) \leq C_{K}^{-2} h \lambda^{-1}\right\},$$ where $C_K$ is specified in Lemma \ref{prop:K_and_Wf_bound}. For the formulation of the FBR, we introduce  an empirical process $Z_M(g)$ for any $g \in \mathcal{G}$ and $t \in \T$,
\begin{equation*}
    Z_{M}(g)(t)=\frac{1}{\sqrt{M}} \sum_{i=1}^{n}\sum_{j=1}^{N_{i}}\left[\psi_{n}(Y_{ij}, T_{ij}; g) K_{T_{ij}}(t)-E\left\{\psi_{n}(Y, T; g) K_{T}(t)\right\}\right].
\end{equation*}
Here $\psi_n(T; g)$ is a real-valued function defined on $\T \times \mathcal{G}$, which may depend on $n$, the total number of subjects in the sample. Provided certain smoothness conditions are met for $\psi_n$, we can establish the following concentration inequality.
\begin{lemma}\label{lem:concen_inequa}
Suppose that $\psi_{n}$ satisfies the following Lipschitz continuous condition:
\begin{equation*}\label{assp:Lipschitz}
    \left|\psi_{n}(T ; f)-\psi_{n}(T ; g)\right| \leq C_{K}^{-1} h^{1 / 2}\norm{f-g}_{\text {sup}} 
\end{equation*}
for any $f, g \in \mathcal{G}$, where $C_{K}$ is specified in Lemma \ref{prop:K_and_Wf_bound}. Then we have
\begin{equation*}\label{eq:concen_inequality}
    \lim _{n \rightarrow \infty} \Pr\left(\sup _{g \in \G} \frac{\norm{Z_{M}(g)}}{h^{-(2m-1) /(4m)}\norm{g}_{\sup}^{1-1 /(2 m)}+n^{-1/2}} \leq(5 \log \log M)^{1 / 2}\right)=1.
\end{equation*}
\end{lemma}

Based on Lemma \ref{lem:concen_inequa}, we can proceed to establish the consistency of $\hat{\mu}_\lambda$ as described below. 
\begin{lemma}\label{lem:convergence}
Suppose Assumption \ref{assp:data_gm} - \ref{assp:Fourier_expansion} hold. Further, we assume 
\begin{equation*}
h = o(1) ~\text{and}~    \{nh^{3 - 1/(2m)}\}^{-1/2}(\log\log n)^{1/2} = o(1).
\end{equation*}
Then $\norm{\hat{\mu}_\lambda - \mu} = O_p(d_n)$, where $d_n = h^m + (nh)^{-1/2}$. In particular, $\hat{\mu}_{\lambda}$ achieves the optimal convergence rate if $h \asymp n^{-1/(2m+1)}$.
\end{lemma}
We are now ready to introduce the crucial technical instrument FBR for $\hat{\mu}_\lambda$.
\begin{thm}[functional Bahadur representation]\label{thm:FBR}
Suppose all conditions in Lemma \ref{lem:convergence} are satisfied. Then we have
\begin{equation*}
    \norm[\big]{\hat{\mu}_\lambda - \mu + S_{M, \lambda}(\mu)} = O_{p}(a_n),
\end{equation*}
where $a_n = n^{-1/2}h^{-(6m-1)/4m}\{h^m + (nh)^{-1/2}\}(\log \log n)^{1/2}$.
\end{thm}

Recall that $S_{M, \lambda}(f)$ is defined in \eqref{eq-SM} for any $f \in \HS$ as 
$$S_{M, \lambda}(f) = -\frac{1}{M}\sum_{i=1}^{n}\sum_{j=1}^{N_i}\{Y_{ij} - f(T_{ij})\}K_{T_{ij}} + W_\lambda f.$$ The first term can be treated as the average of i.i.d. random variables indexed by subjects. The second term is the bias induced by the penalty term. Theorem \ref{thm:FBR} provides a first-order approximation to the difference between the mean estimator and the population mean. 

\subsection{Pointwise Asymptotic Distribution} \label{subsec:pointwise}
In this subsection we develop pointwise asymptotic distributions for $\hat{\mu}_{\lambda}$ based on the FBR established in Theorem \ref{thm:FBR}. Subsequently, the pointwise asymptotic distribution for $\hat{\mu}_{\lambda_1} - \hat{\mu}_{\lambda_2}$ can be readily obtained due to the independence between the two groups. We introduce the following notation for convenience:
\begin{equation}\label{eq:G_functions}
\begin{gathered}
    I_1^{uv} = \iint_{t,s \in \T}c(t,s)h_u(t)h_v(s)\pi(t)\pi(s)dtds,\\
    I_2^{uv} = \int_{t \in \T}c(t,t)h_u(t)h_v(t)\pi(t)dt,\\
    \mathcal{S}_i(t,s) = N_i\sum_{u, v}\frac{h_u(t)h_v(s)}{(1+\lambda \gamma_u)(1 + \lambda \gamma_v)}\brac{(N_i-1)I^{uv}_1 + I^{uv}_2 + \sigma_{\err}^2 \delta_{uv}}, i \leq n, 
\end{gathered}    
\end{equation}
where 
$\{h_u\mid u \in \N\}$ are the basis functions of $\HS$ defined in Assumption \ref{assp:Fourier_expansion}. The following assumption ensures that the tail probability of the random objects in model \eqref{eq:model} decays at an exponential rate. 

\begin{assp}\label{assp:exponential_tails}
There exist some positive constants $C_\err$, $C_X$ and $C_N$ such that
\begin{align*}
   (i)~ \E\brac{\exp(C_X\norm{X})} < \infty, 
   ~  (ii)~ \E\brac{\exp(C_\err \abs{\err})} < \infty, ~(iii)~ \E\brac{\exp(C_N N)} < \infty
\end{align*}
for $X$, $\err$, and $N$ in model \eqref{eq:model}. Furthermore, there exists a constant $C > 0$ such that for any $t \in \T$,
$(iv)~ \E[\brac{X(t)-\mu(t)}^4] \leq C \pbrac{\E\sbrac{\brac{X(t)-\mu(t)}^2}}^2.$
\end{assp}
Proposition 3.2 in \cite{zfsannals15} demonstrates that part (i) of Assumption \ref{assp:exponential_tails} is satisfied if $X$ is a Gaussian process with a square-integrable mean function, for any $C_X \in (0, 1/4)$. Part (ii) is met if $\err$ is normally distributed, which was adopted in \cite{FY2005}. If $N$ follows a Poisson distribution, part (iii) is fullfilled. The fourth moment condition (iv) is commonly seen in literature; see Theorem 8.3.5 in \cite{Hsing2015FDA} for example. In particular, $C$ can be taken as 3 in part (iv) when $X(t)$ is a Gaussian process.
Additionally, $\err$ has a finite fourth moment due to (ii).

Under Assumption \ref{assp:exponential_tails}, we can readily derive the pointwise asymptotic distribution of $\hat{\mu}_\lambda$ using Theorem \ref{thm:FBR}. {We define the biased mean function as $\mu_{b} = (\mathrm{id} - W_\lambda)\mu$, where $\text{id}$ is the identity operator in $\HS$ that maps any function within $\HS$ to itself. The bias term comes from the penalization of the function's smoothness in \eqref{eq:mean_optimization}. Indeed, the mean of the asymptotic distribution in the following theorem is $\mu_b$, and the process of bias removal is elaborated in Remark \ref{remark:unbiased_remark}.}

\begin{thm}\label{thm:ptwise_limiting_dist_biased}
Assume all conditions in Lemma \ref{lem:convergence} and Assumption \ref{assp:exponential_tails} are satisfied. We also assume,  
\begin{equation}\label{eq:ptwise_variance}
    \lim_{n\to \infty}h\sum_{i=1}^{n}\E\brac{\frac{n}{M^2}\Sc_i(t,t)} \triangleq \sigma_t^2, ~\forall t \in \T.
\end{equation}
Additionally, if $a_n = o(n^{-1/2})$, where $a_n$ is specified in Theorem \ref{thm:FBR}, and 
$$\log^2(n)\exp\brac{-\frac{(nh)^{1/2}}{\log(n)}} = o(h^2),$$ then we have
\begin{equation*}
    \frac{(nh)^{1/2}\brac{\hat{\mu}_{\lambda}(t) - \mu_{b}(t)}}{\sigma_t} \overset{d}{\to} N(0, 1).
\end{equation*}
\end{thm}
\begin{remark}
To better appreciate the variance term as outlined in \eqref{eq:ptwise_variance}, we consider a specific case where $N_i = m_0$ almost surely, where $m_0$ is a fixed integer no less than one. A direct computation leads to the following result:
\begin{equation*}
    \sigma_t^2 = \lim_{n\to \infty}h\sum_{u, v} \frac{h_u(t)h_v(t)}{(1 + \lambda \gamma_u)(1 + \lambda \gamma_v)}\left(\frac{m_0-1}{m_0}I^{uv}_1 + \frac{1}{m_0}I^{uv}_2 + \frac{\sigma^2_{\err}}{m_0}\delta_{uv}\right).
\end{equation*}
Notice that, the term $\brac{h\sigma_{\err}^2\sum_u h_u^2(t)/(1+\lambda\gamma_u)^2}/m_0$ reflects the variability resulting from the application of smoothing splines within the RKHS framework, a phenomenon frequently noted in the literature. This can be seen in Theorem 3.5 of \cite{zfsannals13} and Theorem 3.3 of \cite{semiCoxJASA2021}, for instance. The remaining terms account for the variability introduced by the stochastic nature of $X$ and $T$.
\end{remark}

If the bias term, $W_{\lambda}(\mu)$, is negligible in comparison to the variance term, we have the following result for $\hat{\mu}(t)$.
\begin{cor}\label{cor:ptwise_limiting_dist_unbiased}
Suppose the conditions in Theorem \ref{thm:ptwise_limiting_dist_biased} hold and 
\begin{equation*}
    \lim_{n\to \infty}(nh)^{1/2}(W_\lambda\mu)(t) = -b_{t} = 0.
\end{equation*}
 Then we have
\begin{equation*}
    \frac{(nh)^{1/2}\{\hat{\mu}_{\lambda}(t) - \mu(t)\}}{\sigma_{t}} \overset{d}{\to} N(0, 1),
\end{equation*}
where $\sigma_t^2$ is defined in \eqref{eq:ptwise_variance}.
\end{cor}
\begin{remark}\label{remark:unbiased_remark}
   It should be noted that even when $\lim_{n\to \infty}(nh)^{1/2}(W_\lambda\mu)(t)$ exists, the bias term $-b_t$ may not necessarily be zero. 
    In the context of nonparametric regression using smoothing splines, \cite{zfsannals13} investigated  appropriate conditions under which the bias term becomes negligible. The key idea is to ensure that $b_t/\sigma_t = o(1)$, i.e., the squared bias is dominated by the variance term. A sufficient condition to achieve this is
   \begin{equation*}
       \sum_{v} V^2(\mu, h_v) \gamma_v^2 < \infty \text{ and } nh^{4m} = o(1).
   \end{equation*}
Given this condition, it can be shown that:
\begin{align*}
    \abs{W_\lambda \mu(t)} & = \abs[\Big]{W_\lambda \sum_v V(\mu, h_v) h_v(t)} \\
    &\leq \sum_v \abs[\Big]{V(\mu, h_v)} \abs[\bigg]{\frac{\lambda \gamma_v}{1 + \lambda \gamma_v}h_v(t)} \\
    &\leq \lambda\sup_{v \in \N}\norm{h_v}_{\sup} \left\{\sum_vV^2(\mu, h_v) \gamma_v^2\right\}^{1/2}\left\{\sum_v \frac{1}{(1 + \lambda \gamma_v)^2}\right\}^{1/2}\\
    &\leq O(\kappa_n), \textrm{~where~}\kappa_n = h^{2m-1/2} \brac{{\sum_vV^2(\mu, h_v) \gamma_v^2}}^{1/2}.
\end{align*}
Denote $N_0$ as the smallest integer satisfying the condition  $$\frac{3\sigma_t^2}{4} \geq h\sum_{i=1}^{n}\E\brac{\frac{n}{M^2}\Sc_i(t,t)}. $$ Then for $n \geq N_0$, it follows that:
\begin{align*}
    \frac{b_t^2}{\sigma_t^2} &\leq \frac{3nhW_\lambda^2 \mu(t)}{4h\sum_{i=1}^{n}\E\brac{n\Sc_i(t,t)/M^2}} \\
    & \leq C_0nh^{4m}\sum_vV^2(\mu, h_v) \gamma_v^2 = o(1).
\end{align*}
\end{remark}

We now establish the asymptotic distribution for the estimator of the mean difference, under the assumption that the two groups are independent.  The sample sizes (and corresponding penalty parameters) for the two groups are denoted as $n_1~(\lambda_1)$ and $n_2~(\lambda_2)$, respectively. Let $h_g = \lambda_g^{1/(2m)}$ for $g = 1, 2$. The following theorem is a direct result of applying Theorem \ref{thm:ptwise_limiting_dist_biased} and Corollary \ref{cor:ptwise_limiting_dist_unbiased}.
\begin{thm}\label{thm:two_group_ptwise_dist}
Suppose the conditions in Corollary \ref{cor:ptwise_limiting_dist_unbiased} are satisfied. Additionally, we assume
\begin{equation*}
    \lim_{n_g\to \infty}(n_gh_g)^{1/2}(W_{\lambda_g}\mu_g)(t) = 0,
\end{equation*}
and
\begin{equation*}
    \lim_{n_g\to \infty}h_g\sum_{i=1}^{n}\E\brac{\frac{n_g}{M_g^2}\Sc_{g,i}(t, t)} = \sigma_{g,t}^2, \forall t \in \T, ~g=1, 2. 
\end{equation*}
In this context, $\Sc_{g,i}(t,s)$ is the group-specific variant of $\Sc_i(t,s)$ as outlined in \eqref{eq:G_functions}. Here, the general terms $c(t,s), \pi(t)$, and $N_i$ are replaced by their respective group-specific equivalents $c_g(t,s), \pi_g(t)$ and $N_{gi}$ for $g=1, 2$. Furthermore, if
\begin{align*}
a_{g,n} = o(n_g^{-1/2}) \text{ and } \log^2(n_g)\exp\brac{-(n_gh_g)^{1/2}/\log(n_g)} = o(h_g^2),
\end{align*}
where $a_{g,n}$ is the group-specific counterpart of $a_n$ as defined in Theorem \ref{thm:FBR}, we have for any $t\in \T$, 
\begin{equation} \label{eq-ptCLT}
    \left(\frac{\sigma_{1, t}^2}{n_1h_1} + \frac{\sigma_{2, t}^2}{n_2h_2}\right)^{-1/2}\left[\{\hat{\mu}_{\lambda_1}(t) - \hat{\mu}_{\lambda_2}(t)\} - \{\mu_1(t) - \mu_2(t)\}\right] \overset{d}{\to} N(0, 1),
\end{equation}
where $\sigma_{g,t}^2$ denotes the pointwise variance for group $g$ as outlined previously.
\end{thm}
\begin{remark}
    Suppose all notation is the same as in Theorem \ref{thm:two_group_ptwise_dist}, and all associated conditions are satisfied. The pointwise hypothesis testing outlined in \eqref{eq:pt-hypo} is conducted as follows. Let
    \begin{equation*}
        \widehat{\mathfrak{T}}_{pt}(t) = \left(\frac{\sigma_{1, t}^2}{n_1h_1} + \frac{\sigma_{2, t}^2}{n_2h_2}\right)^{-1/2}\{\hat{\mu}_{\lambda_1}(t) - \hat{\mu}_{\lambda_2}(t)\}
    \end{equation*}
For every fixed time point $t^{*} \in \T$, under $H_0$, we reject the null hypothesis at level $\alpha$ if $$\Pr\brac{\abs{\widehat{\mathfrak{T}}_{pt}(t^{*})} > q_{\alpha/2}} < \alpha,$$
where $q_{1-\alpha/2}$ denotes the $(1-\alpha/2)\times100\%$ quantile of a standard normal distribution. Under $H_{\alpha}$, 
\begin{align*}
    \widehat{\mathfrak{T}}_{pt}(t) &= \left(\frac{\sigma_{1, t}^2}{n_1h_1} + \frac{\sigma_{2, t}^2}{n_2h_2}\right)^{-1/2}\sbrac{\hat{\mu}_1(t) - \hat{\mu}_2(t) - \brac{\mu_1(t^{*}) - \mu_2(t^{*})}} \\
    &+ \left(\frac{\sigma_{1, t}^2}{n_1h_1} + \frac{\sigma_{2, t}^2}{n_2h_2}\right)^{-1/2}\brac{\mu_1(t^{*}) - \mu_2(t^{*})}.
\end{align*}
Note that, $n_gh_g \to \infty$ as $n \to \infty$ by the condition $\{nh^{3 - 1/(2m)}\}^{-1}(\log\log n) = o(1)$. Therefore, $\widehat{\mathfrak{T}}_{pt}(t) \to \infty$ if $n \to \infty$ under $H_\alpha$.
\end{remark}

\subsection{Weak Convergence} \label{subsec:global}
This subsection is dedicated to establishing weak convergence of $\hat{\mu}_\lambda - \mu$ as a random function in $\HS$, extending it to two groups for a global test for the mean difference. Without loss of generality, we assume that $T_{ij}$ follows a uniform distribution on $[0, 1]$.

We define a covariance operator $\C: \HS \to \HS$ as follows:
\begin{equation*}
    \mathcal{C}(f) = \E\brac{\innerproduct{X-\mu}{f}_{2}(X-\mu)},
\end{equation*}
where $\innerproduct{f}{g}_{2} = \int_{\T} f(t)g(t) dt$ for any $f, g \in \Lp_2(\T)$.
Clearly, we have
\begin{equation}\label{eq:I_{1, ij}}
    \innerproduct{\C(h_u)}{h_v}_{2} = \iint_{t,s \in \T}c(t,s)h_u(t)h_v(s)dtds = I_{1}^{uv},
\end{equation}
where $I_{1}^{uv}$ is defined in \eqref{eq:G_functions} as $T$ is uniformly distributed on [0, 1]. The proof of weak convergence relies on the decay rate of the integral \eqref{eq:I_{1, ij}} as the indices $u, v \to \infty$. 

In the FDA literature, it is common to assume that the covariance operator is a trace-class operator; see \cite{Hall2007Annals} and \cite{zhou_functional_2022} for examples.
If $h_v(t)$ coincides with eigenfunctions of the covariance operator $\C$ for $v \geq 1$, the associated eigenvalue, which is defined as
\begin{align*}
    \theta_v = \innerproduct{\C(h_v)}{h_v}_{2} &= \iint_{t,s \in \T}c(t,s)h_v(t)h_v(s)dtds,
\end{align*}
is usually assumed to have a polynomial decay rate, i.e., 
$\theta_v \leq c v^{-\tau}$
with some universal constant $c$ and $\tau > 1$. This decay rate of $\theta_v$ ensures that $\C$ is a trace-class operator. Therefore, under this assumption, $\C$ is also guaranteed to be a Hilbert-Schmidt operator. Moreover, the Hilbert-Schmidt norm of $\C$ does not depend on the choice of orthonormal basis functions. Besides that, Assumption \ref{assp:Fourier_expansion} implies $\{h_v\}_{v=1}^{\infty}$ forms a basis in $\Lp_2(\T)$ when $T$ follows a uniform distribution on $[0, 1]$. Given these considerations, we introduce the following assumption.
\begin{assp}\label{assp:HS_norm_of_C}
The Hilbert-Schmidt norm of the covariance operator is finite, i.e., $\norm{\C}_{\mathrm{HS}}^2 = \sum_{v}\norm{\C(h_v)}_{2}^{2}< \infty$. Further, 
\begin{gather*}
    \iint_{t,s \in \T}c(t,s)h_u(t)h_v(s)dtds \leq u^{-\tau/2}v^{-\tau/2}\text{~and}\\
   \int_{t \in \T}c(t,t)h_u(t)h_v(t)dt \leq u^{-\tau/2}v^{-\tau/2}\text{~for some~}\tau > 1.
\end{gather*}
\end{assp}

{To establish the weak convergence for $\hat{\mu}_\lambda - \mu$, we denote $C(\T)$ as the collection of all continuous functions on $\T$, which is equipped with the distance metric $$\rho(f,g) = \sup_{t \in \T}\abs{f(t) - g(t)}$$
for any $f, g \in C(\T)$.}

\begin{thm}\label{thm:weak_conv}
Under Assumption \ref{assp:HS_norm_of_C}, and considering all the conditions outlined in  Theorem \ref{thm:ptwise_limiting_dist_biased}, along with the sufficient condition in Remark \ref{remark:unbiased_remark}, and also
\begin{equation}\label{eq:Gaussian_process_kernel}
    \lim_{n\to \infty}h\sum_{i=1}^{n}\E\brac{\frac{n}{M^2}\Sc_i(t,s)} = C_Z(t,s),
\end{equation}
we have
\begin{equation*}
    \Sb_n(t) = (nh)^{1/2}\{\hat{\mu}_{\lambda}(t) - \mu(t)\}_{t\in \T} \leadsto  \{Z(t)\}_{t\in \T}
\end{equation*}
in $C(\T)$, where $Z$ denotes a Gaussian process with zero mean and covariance function $C_Z$ as described in \eqref{eq:Gaussian_process_kernel}.
\end{thm}
The subsequent theorem is an extension of Theorem \ref{thm:weak_conv} and employs the same notation as those used in Theorem \ref{thm:two_group_ptwise_dist}.
\begin{thm}\label{thm:weak_conv_two_group}
Suppose all conditions outlined in Theorem \ref{thm:weak_conv} are satisfied and 
\begin{equation*}
    \lim_{n_1, n_2 \to \infty}\frac{n_1h_1}{n_2h_2} = r
\end{equation*}
for some positive real number $r$.
Additionally, we assume
\begin{equation}\label{eq:Gaussian_process_kernel_group_wise}
    \lim_{n_g\to \infty}h_g\sum_{i=1}^{n_g}\E\brac{\frac{n_g}{M_g^2}\Sc_{g,i}(t,s)} = C_{Z_g}(t,s), ~~ g=1, 2.
\end{equation}
Then,
\begin{equation*}
    \Sb_{d}(t) = \brac{\frac{(n_1h_1)(n_2h_2)}{n_1h_1 + n_2h_2}}^{1/2}\{(\hat{\mu}_{\lambda_1} - \hat{\mu}_{\lambda_2})(t) - (\mu_1 - \mu_2)(t)\}_{t\in \T} \leadsto  \{Z_d(t)\}_{t\in \T}
\end{equation*}
in $C(\T)$, where $Z_d$ is a Gaussian process with zero mean and covariance function $$C_{d} = \frac{1}{r+1}C_{Z_1} + \frac{r}{r+1}C_{Z_2},$$ 
where $C_{Z_1}$ and $C_{Z_2}$ are defined in \eqref{eq:Gaussian_process_kernel_group_wise}.
\end{thm}

\begin{remark}\label{remark:global_test_construction}
For global testing of $\mu_1(t) - \mu_2(t)$ as outlined in \eqref{eq:hypo}, we employ the Cram\'{e}r-von Mises criterion \citep{AndersonCMTest1962}. Indeed, Theorem \ref{thm:weak_conv_two_group} and the continuous mapping theorem allow us to establish the following:
\begin{align*}
    & ~~~~\frac{(n_1h_1)(n_2h_2)}{n_1h_1 + n_2h_2}\int_{t \in \T}\brac{(\hat{\mu}_{\lambda_1} - \hat{\mu}_{\lambda_2})(t) - (\mu_1 - \mu_2)(t)}^2 dt \\
    & \overset{d} \rightarrow \int_{t\in \T}\brac{Z_d(t)}^2dt \\
    & = \int_{t\in\T}\brac{\sum_{\nu}\eta_\nu^{1/2}\zeta_{\nu} \varphi_{\nu}(t)}^2dt\\
    & = \sum_{\nu}\eta_\nu\chi^2_{\nu}(1) \triangleq \chi^2_{w},
\end{align*}
where $(\eta_\nu, \varphi_\nu)$'s stand for eigenvalue-eigenfunction pairs of the covariance function of $Z_d$. Additionally, $\chi^2_{\nu}(1)$'s are i.i.d. $\chi^2$ distributed random variables, each with one degree of freedom. The series comprising these weighted chi-squared random variables can be approximated either through truncation or using a Welch–Satterthwaite type $\chi^2$ approximation \citep{sumWeightedChi2Greg2017}.
\end{remark}
\begin{remark}
We can construct the global testing outlined in \eqref{eq:hypo}. Suppose all conditions in Theorem \ref{thm:weak_conv_two_group} are satisfied. Let
\begin{equation*}
    \widehat{\mathfrak{T}}_{glo}^2 = \frac{(n_1h_1)(n_2h_2)}{n_1h_1 + n_2h_2}\int_{t \in \T}\brac{(\hat{\mu}_{\lambda_1} - \hat{\mu}_{\lambda_2})(t)}^2 dt.
\end{equation*}
We reject the null hypothesis at level $\alpha$ when $\widehat{\mathfrak{T}}_{glo}^2 > q_{\chi^2_{w}, \alpha}$, where $q_{\chi^2_{w}, \alpha}$ is the $(1-\alpha)\times 100\%$ quantile of a $\chi^2_{w}$ defined in Remark \ref{remark:global_test_construction}. Under $H_\alpha$, we have
\begin{align*}
    \widehat{\mathfrak{T}}_{glo}^2 &= \frac{(n_1h_1)(n_2h_2)}{n_1h_1 + n_2h_2}\int_{t \in \T}\sbrac{\hat{\mu}_{\lambda_1}(t) - \hat{\mu}_{\lambda_2}(t) - \brac{\mu_1(t) - \mu_2(t)}}^2 dt \\
    &+ \frac{(n_1h_1)(n_2h_2)}{n_1h_1 + n_2h_2}\int_{t \in \T}\brac{\mu_1(t) - \mu_2(t)}^2 dt.
\end{align*}
The above calculation implies that, under $H_\alpha$, $\widehat{\mathfrak{T}}_{glo}^2 \to \infty$ when $n \to \infty$.
\end{remark}

\subsection{Bootstrap Calibration}\label{subsec:bootstrap_validity}
Ideally, the asymptotic results established in Sections \ref{subsec:pointwise} and \ref{subsec:global} can be directly employed to construct pointwise confidence intervals and test the global null. Nonetheless, implementing these results in practice entails solving ODEs as indicated in Lemma \ref{lem:exist_of_eigensys} and estimating nuisance parameters, such as the covariance function of $X$. 
On the one hand, our numerical studies show that the convergence rates of the limits in equation \eqref{eq:ptwise_variance} and equation \eqref{eq:Gaussian_process_kernel} are slow. On the other hand, the computational time increases exponentially as two-dimensional smoothing is required to estimate the covariance function. To circumvent these technical challenges and enhance computational efficiency, we develop a multiplier bootstrap procedure for performing the inference task. Algorithm \ref{algo:bootstrap_algo} presents the details to implement the bootstrap method to construct pointwise confidence intervals for $\mu_1(t) - \mu_2(t)$ and conduct hypothesis testing for \eqref{eq:hypo}.
% The purpose of this section is to present such a procedure and establish its validity.
\begin{algorithm}[!ht]
  \caption{Proposed multiplier bootstrap procedure}
  \begin{enumerate}
    \item Generate i.i.d. bootstrap weights $U_{i, b}^{g}$ independent of the data $\{T_{gij}, Y_{gij}\}$, where $1\leq i\leq n_g, 1\leq b\leq B$, and $g=1, 2$. Here, $U_{i, b}^{g} = 1 + 2^{-1/2}$ with probability $2/3$ and $U_{i, b}^{g} = 1 + 2^{1/2}$ with probability $1/3$.
    \item The $b$th bootstrap estimator for group $g$ is computed as follows:
    \begin{equation*}
        \hat{\mu}_{\lambda_g,b} = \underset{f \in \HS}{\amin} \Bigg[\frac{1}{2M_g}\sum_{i=1}^{n_g}\sum_{j=1}^{N_{gi}}U_{i, b}^{g}\big\{Y_{gij} -  f(T_{gij})\big\}^2 + \frac{\lambda_g}{2} J(f, f)\Bigg],
    \end{equation*}
    for $b=1, 2, \ldots, B$. {This optimization problem can be solved by applying a weighted variant of Proposition \ref{prop:representer_thm}; see Chapter 3.2.4 in \cite{gu2013smoothing}.}
    \item[3.a] For each $1 \leq b \leq B$, let
    \begin{equation*}
        \Delta_b(t) = \hat{\mu}_{\lambda_1,b}(t) - \hat{\mu}_{\lambda_2,b}(t) - \brac{\hat{\mu}_{\lambda_1}(t) - \hat{\mu}_{\lambda_2}(t)}, ~~t \in \T.
    \end{equation*}
    Denote $\hat{\sigma}_{\Delta_t}$ as the standard error of $\{\Delta_b(t) \mid b = 1, 2, \ldots, B\}$.
    \item[3.b] For each $b$, we construct a variance-adjusted test statistic \citep{Wang2023Sinica}
    \begin{equation*}
        \hat{\kappa}_{b}^2 = \bigintssss_{\T}\sbrac[\bigg]{\frac{\hat{\mu}_{\lambda_1,b}(t) - \hat{\mu}_{\lambda_2,b}(t) - \brac{\hat{\mu}_{\lambda_1}(t) - \hat{\mu}_{\lambda_2}(t)}}{\hat{\sigma}_{\Delta_t}}}^2dt.
    \end{equation*}
    \item[4.a] The $(1-\alpha)\times 100\%$ confidence interval of $\mu_1(t) - \mu_2(t)$ at every specific $t \in \T$ is as follows:
    \begin{equation}\label{eq:bootstrap_CI}
        \brac{\hat{\mu}_{\lambda_1}(t) - \hat{\mu}_{\lambda_2}(t)} \pm q_{1-\alpha/2}\hat{\sigma}_{\Delta_t},
    \end{equation}
    where $q_{1-\alpha/2}$ denotes the $(1-\alpha/2)\times100\%$ quantile of a standard normal distribution. 
    \item[4.b] At a significance level of $\alpha$, we reject the null hypothesis in \eqref{eq:hypo} if $\hat{\kappa} > q_{1-\alpha}$, where
    \begin{equation*}
        \hat{\kappa}^2 = \bigintsss_{\T}\brac{\frac{\hat{\mu}_{\lambda_1}(t) - \hat{\mu}_{\lambda_2}(t)}{\hat{\sigma}_{\Delta_t}}}^2 dt,
    \end{equation*}
    and $q_{1-\alpha}$ is the $(1-\alpha)\times 100\%$ quantile of $\{\hat{\kappa}_{b}\}_{b=1}^{B}$.
\end{enumerate}
\label{algo:bootstrap_algo}
\end{algorithm}

Because two groups of data are independent of each other, we only need to establish the validity of the multiplier bootstrap procedure for each group. The following two theorems establish statistical guarantees for the bootstrap procedure, where the group index $g$ is omitted, and $\hmu_{\lambda,b}$ is the bootstrap mean estimator defined in Algorithm \ref{algo:bootstrap_algo}.
\begin{thm}\label{thm:pt_wise_bootstrap_validity}
Under the conditions of Theorem \ref{thm:ptwise_limiting_dist_biased}, we have for any $t \in \T$,
\begin{equation*}
    \frac{(nh)^{1/2}\brac{\hat{\mu}_{\lambda, b}(t) - \hat{\mu}_\lambda(t)}}{\sigma_t} \to N(0,1).
\end{equation*}
\end{thm}
\begin{thm}\label{thm:weak_conv_bootstrap_validity}
Under the conditions of Theorem \ref{thm:weak_conv}, we have
\begin{equation*}
    \Sb_{n,b}(t) = (nh)^{1/2}\{\hat{\mu}_{\lambda,b}(t) - \hmu_\lambda(t)\}_{t \in \T} \leadsto  \{Z(t)\}_{t\in \T}
\end{equation*}
\end{thm}
\section{Simulation Studies}\label{sec:Simulation_studies}
Throughout this section, we consider $\HS = H^2(\T)$ and $J(f, f) = \int_{\T}\{f^{''}(t)\}^2dt$. We showcase the finite sample performance of the proposed method, henceforth
referred to as SS. Details regarding the simulation settings are provided in Section \ref{subsec:simulation_settings}.  We then elucidate the estimation performance of our method in Section \ref{subsec:simulation_estimations}. Finally, the performance of our proposed method in both pointwise and global testing is thoroughly evaluated in Sections \ref{subsec:simulation_ptwise_inference} and \ref{subsec:simulation_global_inference}, respectively.

\subsection{Simulation Settings}\label{subsec:simulation_settings}
Data are simulated from model \eqref{eq:model} for each group, where $\err_{g} \sim N(0, \sigma_{g, \err}^2)$. Specifically, $\sigma_{1, \err}^2 = 0.09$ and $\sigma_{2, \err}^2 = 0.04$. For $i = 1, \ldots, n_g$,  we generate the random function $X_{gi}(t)$ utilizing the Karhunen–Lo\'{e}ve expansion:
\begin{equation*}
    X_{gi}(t) = \mu_g(t) + \sum_{k=1}^{\infty}\xi_{gik}\varphi_{gk}(t),~i = 1, \ldots, n_g.
\end{equation*}
In this context, $\xi_{gik}$'s, which are fPC scores, follow $N(0, \theta_{gk})$. Additionally, $\varphi_{gk}$ represents the eigenfunction of the covariance function of $X_{gi}$, corresponding specifically to the $k$th eigenvalue denoted as $\theta_{gk}$.

To examine the impact of the homogeneous covariance assumption on mean difference detection, we keep the eigensystem for group 1 constant across all simulation settings, while various configurations are explored for the eigensystem of group 2. In particular, for group 1, the eigenvalues are set as $\theta_{11} = 1,~\theta_{12} = 0.25,~\theta_{13} = 0.09,~\theta_{14} = 0.05$,  and $\theta_{1k}= 0$ for $k \geq 5$. Additionally, the eigenfunctions for group 1, $\varphi_{1k}(t)$, are set to be $\sin(k\pi t)$ for every $t \in \T$ and for each $k \geq 1$. 
Regarding the eigensystem of group 2, we consider the following three settings.
\begin{enumerate}
    \item In setting \textit{c.1}, group 2's eigensystem is identical to that of group 1, meaning $\gamma_{2k} = \gamma_{1k}$ and $\varphi_{2k}(t) = \varphi_{1k}(t)$ for every $k$ and all $t \in \T$.
    \item In setting \textit{c.2}, group 2 has different eigenvalues from group 1, but retains the same eigenfunctions. The eigenvalues are set as $\theta_{21} = 0.81, \theta_{22} = 0.36, \theta_{23} = 0.09, \theta_{24} = 0.01$, and $\theta_{2k} = 0$ for $k \geq 5$.
    \item In setting \textit{c.3}, the eigensystem of group 2 is completely distinct from that of group 1. Specifically, the eigenvalues are $\theta_{21} = 0.64, \theta_{22} = 0.36, \theta_{23} = 0.16, \theta_{24} = 0.04, \theta_{25} = 0.01$, and $\theta_{2k} = 0$ for $k \geq 6$. The eigenfunctions are set as $\varphi_{2k} = \cos(k\pi t)$ for every $k \geq 1$ and all $t \in \T$.
\end{enumerate}
In terms of the mean functions for these two groups, we consider the following design:
\begin{equation}\label{eq:meandesign}
\begin{gathered}
    \mu_1(t) = (2t - 0.3)^3 + 0.5t,\quad{\text{and}}\\
     ~~~\mu_2(t) - \mu_1(t) = \delta n_2^{-1/4}\brac{e^t - (2t - 1)^3 - 1},
\end{gathered}
\end{equation}
where $\delta$, ranging in the set $\{0, 0.5, 1\}$, represents the size of the mean difference between the two groups.
 
To assess the performance of the proposed estimator and testing procedure, we consider different combinations of subject counts in the two groups: specifically, $n_1 = 200$, and $n_2$ varies among $\{100, 200, 400\}$. The number of observations per subject in both groups is determined by random sampling from a discrete uniform distribution over the set $\{2, 3, \ldots, N_{\max}\}$, where $N_{\max} \in \{6, 10, 14, 18\}$. The larger $N_{\max}$ results in, on average, a greater number of observations per subject. Comparing these settings allows us to evaluate the performance of our method under varying levels of data sparsity.

The observation times for each group are sampled randomly from a uniform distribution over the interval $[0, 1]$. For each of the aforementioned settings, we perform 1000 Monte Carlo simulations. For both the pointwise testing and the global testing outlined in Algorithm \ref{algo:bootstrap_algo}, we choose to use 300 bootstrap samples, denoted as $B = 300$.

\begin{table}[!ht]
\caption{Mean and standard error (in parenthesis) of IMSEs across 1000 Monte Carlo runs with $\delta = 0.5$.}
\centering
\addtolength{\tabcolsep}{-4.7pt} 
\def\arraystretch{1.5}
\begin{tabular}[t]{c c c c c c c c}
\hline
\hline
\multirow{2}{*}{Setting } & \multirow{2}{*}{$N_{\max}$} & \multicolumn{2}{c}{$n_1=200, n_2 = 100$} & \multicolumn{2}{c}{$n_1=200, n_2 = 200$} & \multicolumn{2}{c}{$n_1=200, n_2 = 400$}\\ 
 & & SS & PACE & SS & PACE & SS & PACE\\
\hline
\multirow{2}{*}{c.1} & $10$& 0.020(0.019) & 0.021(0.018) & 0.013(0.010) & 0.014(0.010) & 0.011(0.008) & 0.011(0.008)\\
& $18$ & 0.018(0.016) & 0.018(0.015) & 0.012(0.010) & 0.012(0.010) & 0.009(0.008) & 0.009(0.008)\\
\hline
\multirow{2}{*}{c.2} & $10$ & 0.019(0.015) & 0.020(0.015) & 0.014(0.010) & 0.014(0.010) & 0.010(0.008) & 0.010(0.007)\\

& $18$ & 0.017(0.015) & 0.017(0.015) & 0.011(0.009) & 0.011(0.008) & 0.009(0.008) & 0.009(0.008)\\
\hline
\multirow{2}{*}{c.3} & $10$ & 0.020(0.012) & 0.022(0.012) & 0.014(0.009) & 0.015(0.009) & 0.011(0.007) & 0.011(0.007)\\

& $18$ & 0.019(0.012) & 0.019(0.011) & 0.012(0.008) & 0.012(0.007) & 0.009(0.007) & 0.009(0.006)\\
\hline
\end{tabular}
\label{tb:IMSE_comparison}
\end{table}

\subsection{Estimation Performance}\label{subsec:simulation_estimations}
In evaluating the estimation efficacy of SS, we compare it against PACE \citep{FY2005}, a widely adopted method for estimating the mean function in the context of sparsely observed functional data. The performance of both methods is assessed using the Integrated Mean Squared Error (IMSE), defined as follows:
\begin{equation*}
    \int_{\T}[\{\mu_1(t) - \mu_2(t)\} - \{\hat{\mu}_{\lambda_1}(t) - \hat{\mu}_{\lambda_2}(t)\}]^2dt.
\end{equation*}
We also compare our method to PACE in terms of out-of-sample performance, with detailed results provided in Appendix \ref{appendix:extra_numerical_results}. It turns out that these two methods have similar out-of-sample performance. 

Table \ref{tb:IMSE_comparison} presents the average values and standard errors of the IMSEs for both methods, calculated over 1000 simulation runs under various settings. It is evident that the proposed method attains estimation accuracy comparable to that of PACE. As anticipated, the accuracy of both methods improves with an increase in $n_2$. These numerical findings underscore the effectiveness of aggregating observations from all subjects in estimating the mean function.

\subsection{Pointwise Inferences Performance}\label{subsec:simulation_ptwise_inference}

To evaluate the performance of the pointwise confidence intervals outlined in Algorithm \ref{algo:bootstrap_algo}, we compute their actual coverage probabilities at values of $t_p\in \{0, 0.02, \ldots, 1\}$. For these pointwise confidence intervals, the nominal coverage probability is set as $1 - \alpha = 0.95$. Figure \ref{fig:coverProb} illustrates the actual coverage probabilities of these intervals under different simulation settings.
\begin{figure}[!ht]
    \centering
    \includegraphics[scale = 0.72]{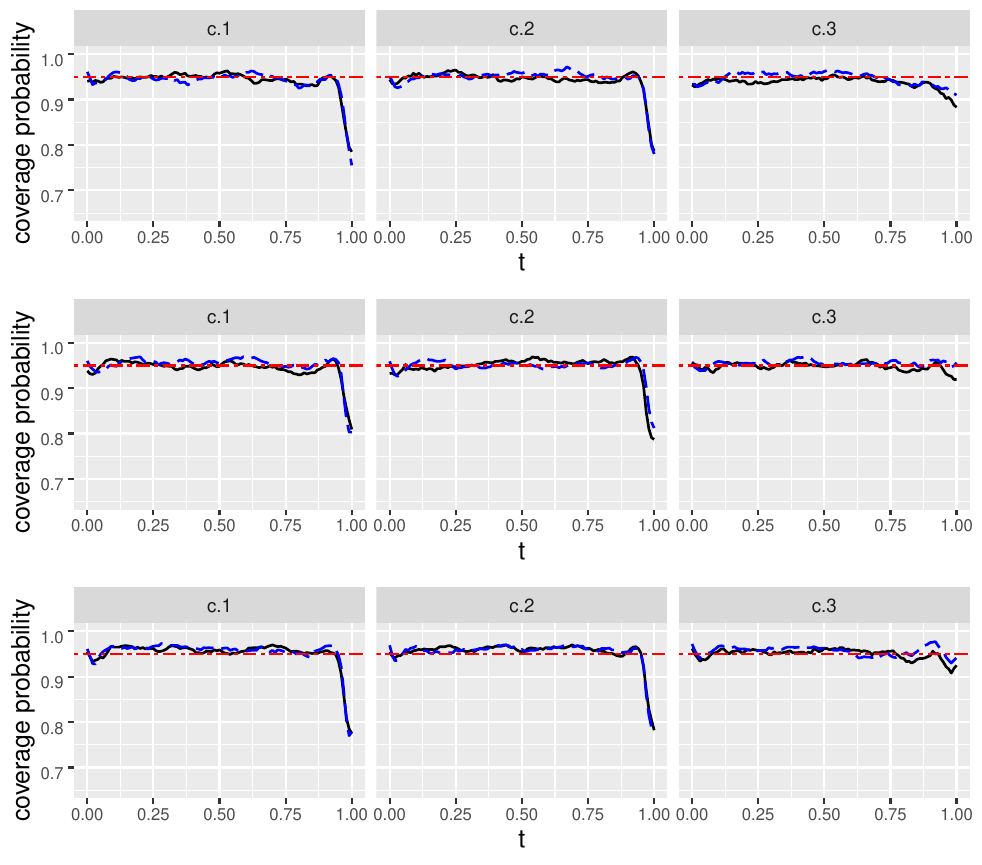}
    \caption{In every panel, the red dotted-dash  line represents the nominal coverage probability. The actual pointwise coverage probabilities for $N_{\max} = 10$ and $N_{\max} = 18$ are represented by black solid and blue long-dashed lines, respectively. The number of subjects in group 2, denoted by $n_2$, with values of 100, 200, and 400, arranged from the top to the bottom.}
    \label{fig:coverProb}
\end{figure}

Our observations reveal that the actual coverage probabilities of our proposed method closely align with the nominal level, regardless of the groups having a common covariance structure or not. The only deviation from this pattern occurs near the boundary, where the performance drops. This is likely due to a boundary effect, which implies that the scarcity of observations near the boundary leads to less reliable estimates. 

\subsection{Global Inference Performance}\label{subsec:simulation_global_inference}
In this subsection, we evaluate the performance of the global testing procedure introduced in Remark \ref{remark:global_test_construction}, where the implementation follows Algorithm \ref{algo:bootstrap_algo}. Our approach is compared with the fPCA-type test statistic (hereinafter referred to as SH), as proposed in \cite{ejs2sample}, which has shown superior performance against several alternatives. Specifically, SH assumes the covariance functions of these two groups share common eigenfunctions. It estimates these eigenfunctions by aggregating data from both groups. Each individual curve is then projected onto the space spanned by these estimated eigenfunctions. Assuming that the projected scores within each group follow a multivariate normal distribution, \cite{ejs2sample} developed a test statistic that asymptotically follows a $\chi^2$ distribution, with the degrees of freedom equal to $p$, the dimensionality of the projection space. On the suggestion of a referee, we also compare to the functional additive mixed models proposed by \cite{scheipl2015functional} (hereinafter referred to as pffr) to estimate the mean function only explained by the group indicator.
\begin{figure}[!ht]
    \centering
    \caption{Rejection rates across 1000 Monte Carlo runs at a 5\% significance level with $N_{\max} = 10$ with $n_1 = 100$ under setting c.1.}
    \includegraphics[scale = 0.68]{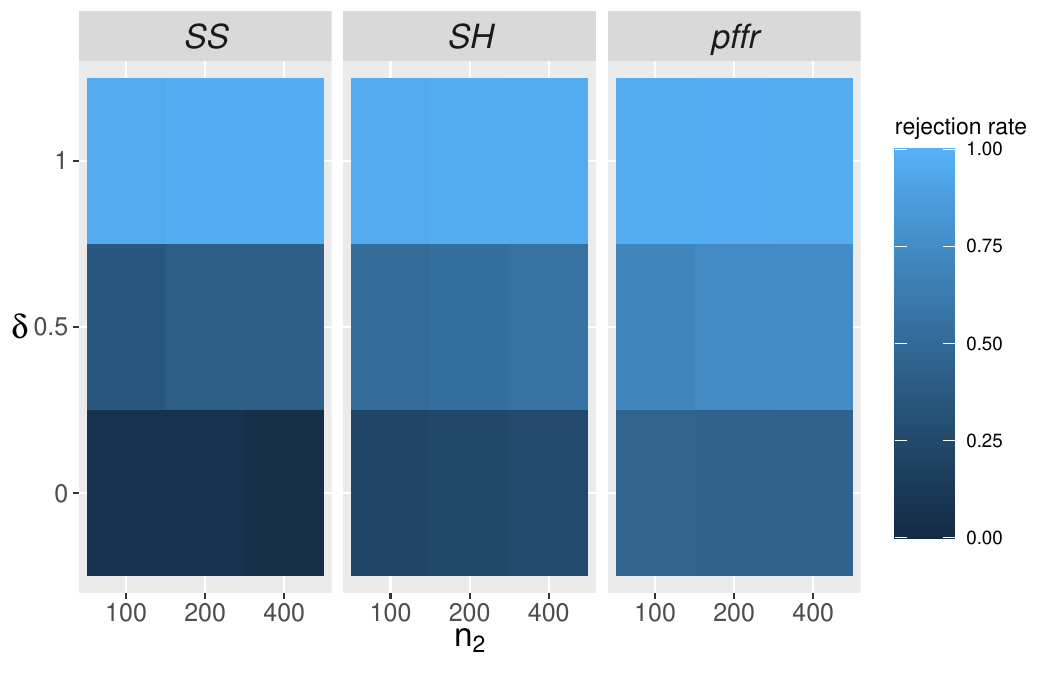}
    \label{fig:weak_convergence_N10c1}
\end{figure}
\begin{figure}[!ht]
    \centering
    \caption{Rejection rates across 1000 Monte Carlo runs at a 5\% significance level with $N_{\max} = 10$ with $n_1 = 100$ under setting c.2.}
    \includegraphics[scale = 0.68]{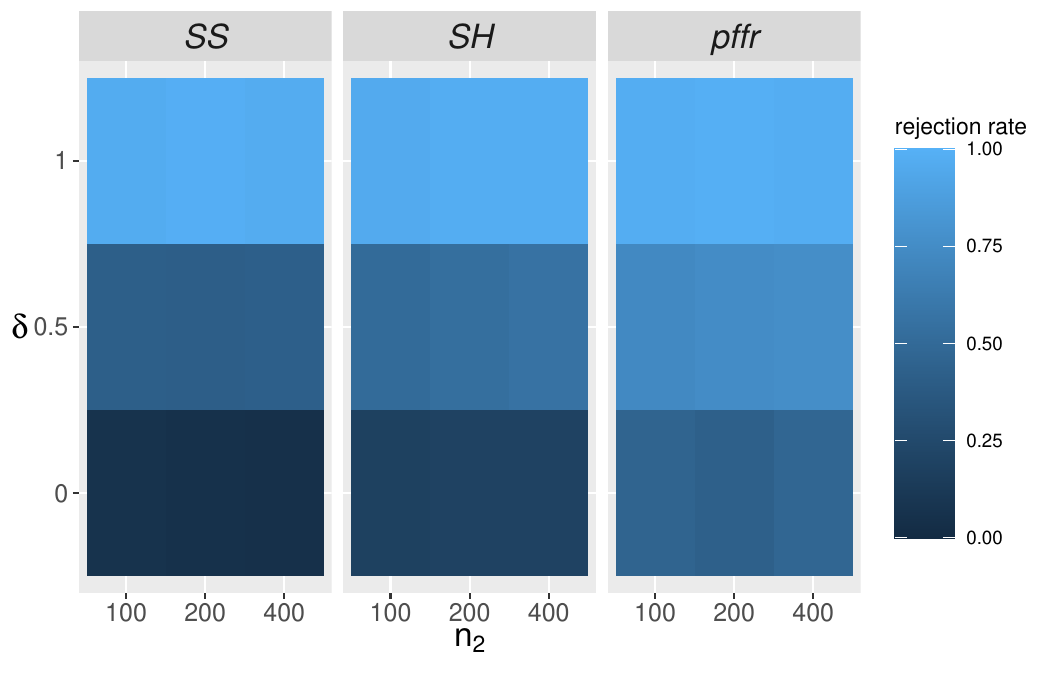}
    \label{fig:weak_convergence_N10c2}
\end{figure}

\begin{figure}[!ht]
    \centering
    \caption{Rejection rates across 1000 Monte Carlo runs at a 5\% significance level with $N_{\max} = 10$ with $n_1 = 100$ under setting c.3.}
    \includegraphics[scale = 0.68]{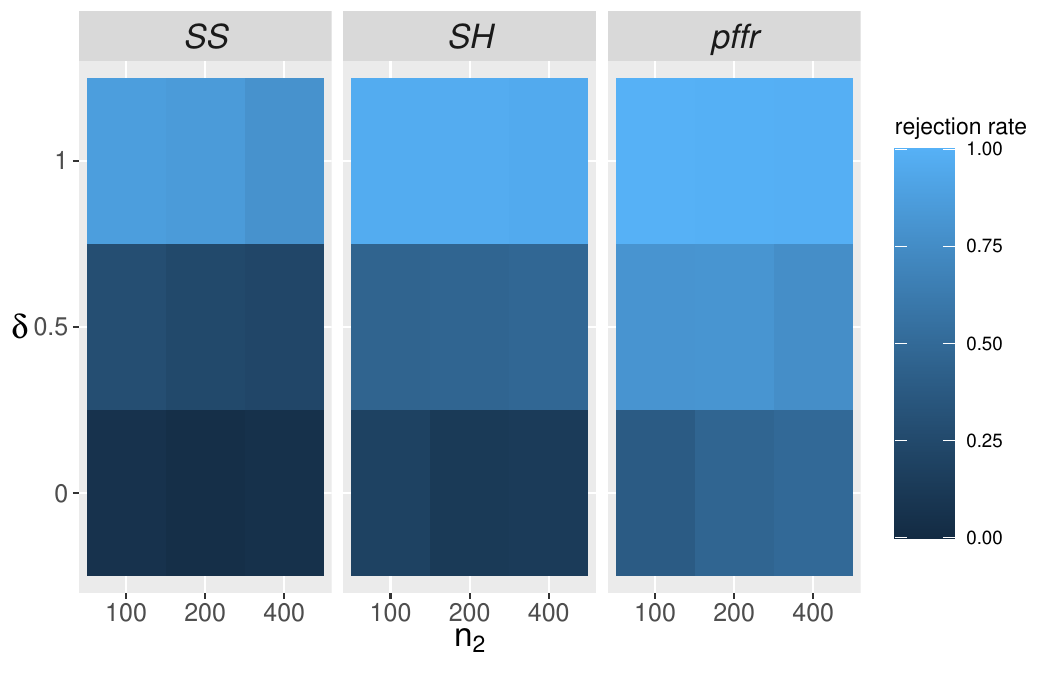}
    \label{fig:weak_convergence_N10c3}
\end{figure}

\begin{table}[!ht]
\caption{Rejection rates across 1000 Monte Carlo runs at a 5\% significance level with $N_{\max} = 10$}
\centering
\addtolength{\tabcolsep}{-2.5pt} 
\def\arraystretch{1.3}
\begin{tabular}[t]{c c @{\hspace{1.6em}} c c c @{\hspace{2em}} c c c @{\hspace{2em}} c c c}
\hline\\[-3ex]
\multirow{2}[2]{*}{Setting} & \multirow{2}[2]{*}{Test} & \multicolumn{3}{c}{$n_1=200, n_2 = 100$} & \multicolumn{3}{c}{$n_1=200, n_2 = 200$} & \multicolumn{3}{c}{$n_1=200, n_2 = 400$}\\[0.3ex]
\cline{3-11}\\[-2.8ex]
 & & $\delta = 0$ & $\delta = 0.5$ & $\delta = 1$ & $\delta = 0$ & $\delta = 0.5$ & $\delta = 1$ & $\delta = 0$ & $\delta = 0.5$ & $\delta = 1$\\[0.5ex]
\hline\\[-3.5ex]
\multirow{3}{*}{c.1} & \textbf{SS} & 0.056 & 0.358 & 0.955 & 0.062 & 0.414 & 0.979 & 0.046 & 0.417 & 0.976\\[0.5ex]
& SH & 0.213 & 0.506 & 0.954 & 0.232 & 0.538 & 0.970 & 0.252 & 0.560 & 0.963\\
& pffr & 0.453 & 0.696 & 0.963 & 0.446 & 0.737 & 0.976 & 0.437 & 0.743 & 0.966\\[0.2ex]
\hline\\[-3.5ex]
\multirow{3}{*}{c.2} & \textbf{SS} & 0.066 & 0.416 & 0.964 & 0.053 & 0.411 & 0.983 & 0.042 & 0.420 & 0.966\\[0.2ex]
& SH & 0.180 & 0.506 & 0.954 & 0.193 & 0.536 & 0.974 & 0.192 & 0.559 & 0.973\\
& pffr & 0.453 & 0.727 & 0.974 & 0.421 & 0.748 & 0.986 & 0.466 & 0.758 & 0.975\\
\hline
\multirow{3}{*}{c.3} & \textbf{SS} & 0.071 & 0.297 & 0.876 & 0.070 & 0.274 & 0.836 & 0.055 & 0.227 & 0.781\\[0.2ex]
& SH & 0.202 & 0.458 & 0.947& 0.157 & 0.457 & 0.950 & 0.138 & 0.451 & 0.951\\
& pffr & 0.387 & 0.801 & 0.997 & 0.459 & 0.805 & 0.990 & 0.494 & 0.759 & 0.985\\
\hline
\end{tabular}
\label{tb:weak_convergence_N10}
\end{table}

Figures \ref{fig:weak_convergence_N10c1} - \ref{fig:weak_convergence_N10c3} display the rejection rates at a $5\%$ significance level for these two methods across different simulation designs with sampling frequency $N_{\max} = 10$, and Table \ref{tb:weak_convergence_N10} presents more details that may not be able to obtained directly from these figures. Generally, the power of all methods improves with a larger number of observations and/or a greater magnitude of mean difference between the two groups. This increase in observations can be attributed to either a higher sampling frequency or a greater number of subjects in the second group. The Type-I error of the proposed testing procedure aligns closely with the pre-determined significance level under various settings. In contrast, both SH and pffr tend to produce inflated type I errors. Moreover, our method tends to be more conservative in rejecting the null hypothesis when there is no mean difference ($\delta = 0$) or the mean difference is relatively small ($\delta = 0.5$), especially when the covariance structures of the two groups are distinct. Additional simulation results for other values of $N_{\max}$ are available in Appendix \ref{appendix:extra_numerical_results}.

\section{Real Data Examples}\label{sec:real_data_examples}
\subsection{Diffusion Tensor Imaging for Fractional Anisotropy Tract Profiles}

The first dataset we focus on originates from a neuroimaging study on individuals with multiple sclerosis (MS). MS is recognized as an immune-mediated inflammatory ailment, manifesting alterations in white-matter tracts. Diffusion tensor imaging (DTI), a type of magnetic resonance imaging, is employed to examine these white-matter tracts.
\begin{figure}[!ht]
    \centering
    \includegraphics[scale = 0.67]{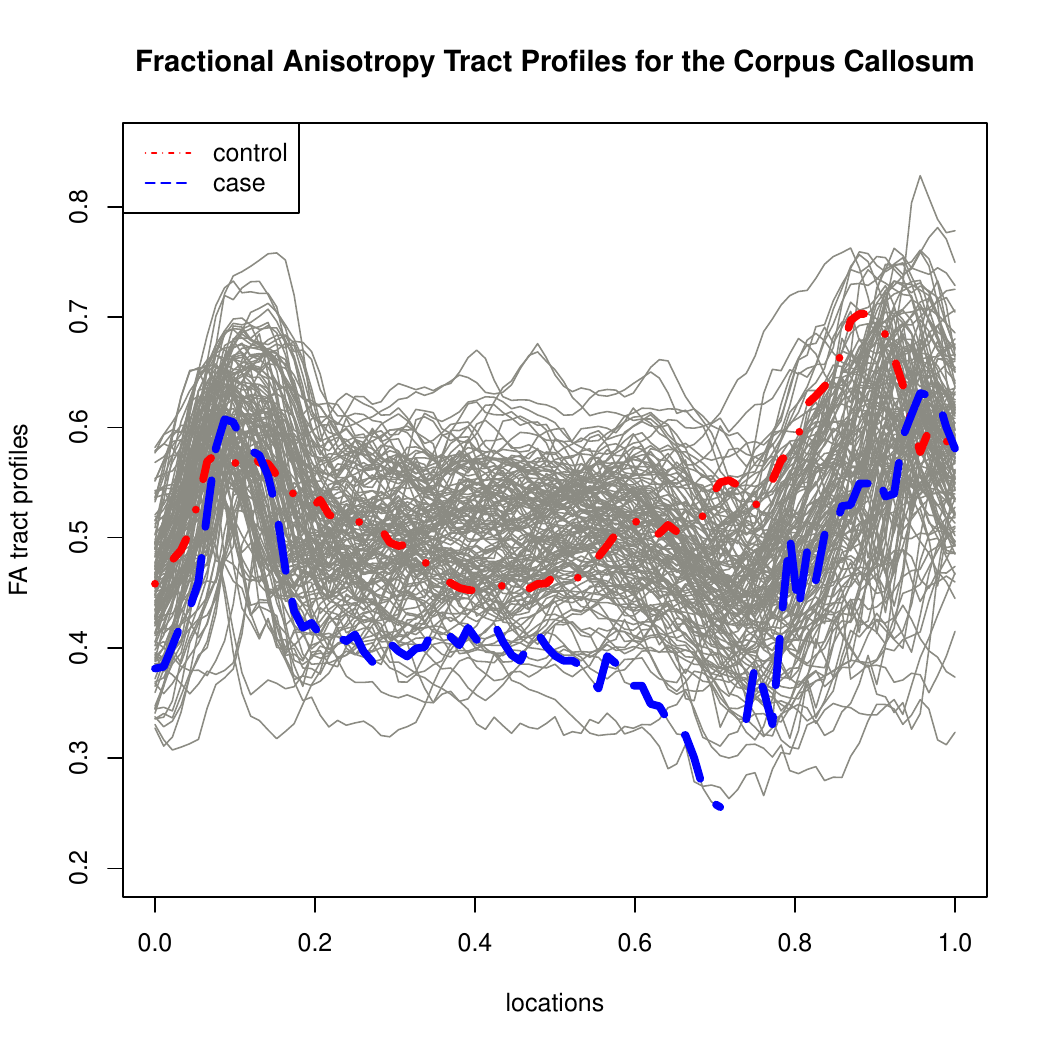}
    \caption{The FA profiles along the CC tract for all participants are shown. The graph features a highlighted control subject represented by a red dot-dashed line and a case subject denoted by a blue dashed line.}
    \label{fig:FA profiles}
\end{figure}

DTI facilitates the generation of detailed images of white-matter anatomy within the brain by assessing the overall directionality of water diffusion. A commonly employed measure is fractional anisotropy (FA), which characterizes the degree of directional diffusion of water. This study is centered on the analysis of FA  profiles parameterized by distance along the tract, the corpus callosum (CC) in particular, to investigate brain structural differences between MS patients and healthy controls.

This dataset consists of 42 control subjects and 100 cases and is publicly available via the \texttt{refund} package \citep{refundRPkg}. Our analysis focuses on data from the baseline visit of the patients. For each individual, FA readings are recorded at 93 distinct locations along the CC. These readings are then interpreted as regularly spaced observations from a latent smooth function specific to each subject, with the presence of random noise in the measurements.

Figure \ref{fig:FA profiles} displays trajectories of FA profiles for both the healthy and control groups. While FA observations are typically collected using a regular grid sampling approach, there are instances of missing data in some subjects, as illustrated by the curve marked with a blue dashed line in Figure \ref{fig:FA profiles}. More detailed information regarding the design, methodologies, and rationale behind the medication used in this study is available in \cite{DTIBioStudy2010}.
\begin{figure}[!ht]
    \centering
    \includegraphics[scale = 0.65]{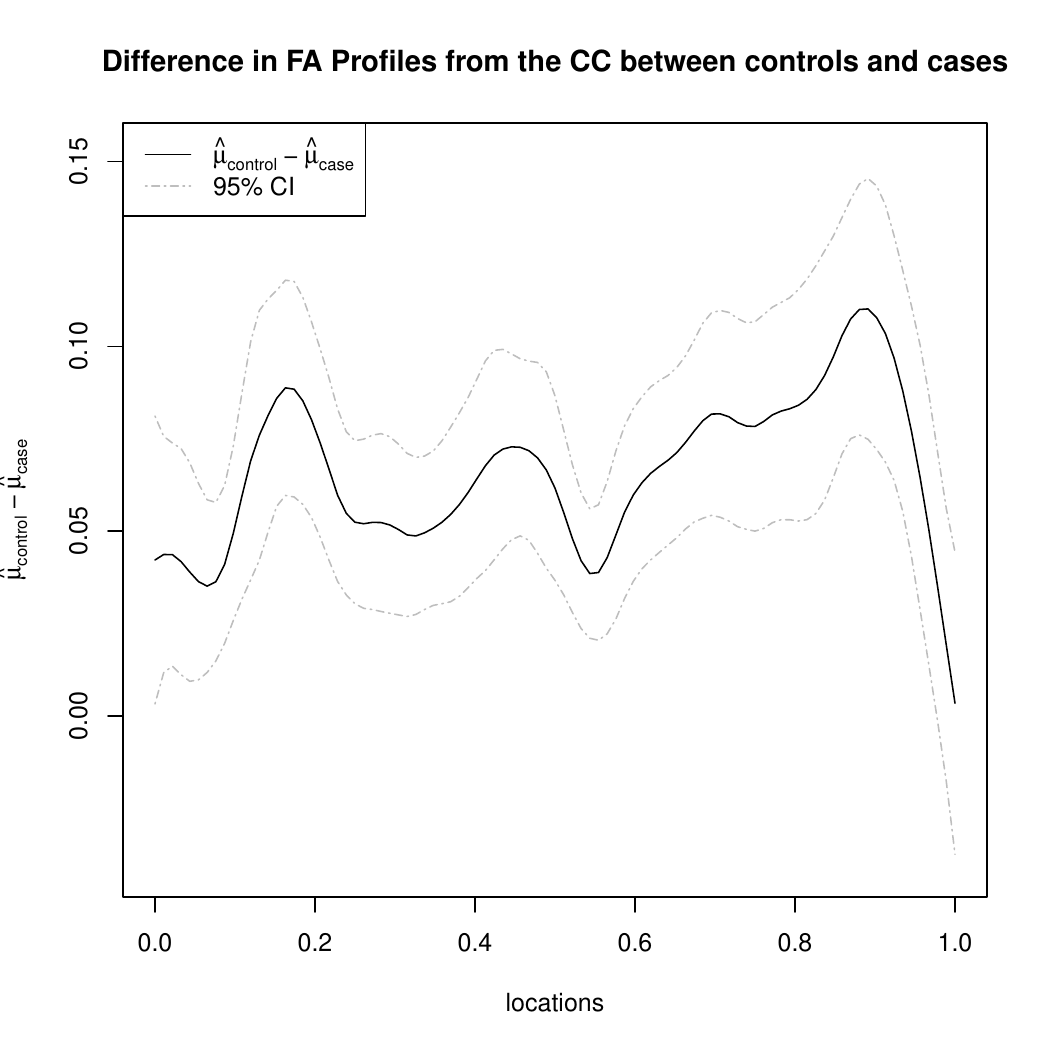}
    \caption{The mean difference between cases and controls (black solid line), along with $95\%$ pointwise confidence intervals (blue dot-dashed lines) over the tract of CC.}
    \label{fig:DTIMeanEst}
\end{figure}

To showcase the performance of the proposed method, we deliberately `sparsify' the data, which originally consists of regularly spaced and densely collected observations. For each subject, the number of observations is randomly chosen to be between 2 and 18, with each number being equally likely. The estimated difference in mean functions between the control and case groups, along with $95\%$ pointwise confidence intervals, is shown in Figure \ref{fig:DTIMeanEst}.
This illustration uncovers a complex pattern that is not readily apparent from just observing Figure \ref{fig:FA profiles}. It identifies specific locations along the CC where there are notable differences in white-matter between healthy individuals and MS patients. Furthermore, the null hypothesis of identical mean FA profiles across the two groups is rejected at a 5\% significance level. This result corroborates existing research, which indicates considerable neuronal loss in the CC of MS patients at the advanced stages, as referenced in \citep{DTIBioStudy2010, JRSSC2016}. Given that this analysis exclusively utilizes baseline data, it suggests that our method could be effective for early diagnosis of MS. 

\subsection{Beijing Air Pollutants Data}
\begin{figure}[!ht]
    \centering
    \includegraphics[scale = 0.67]{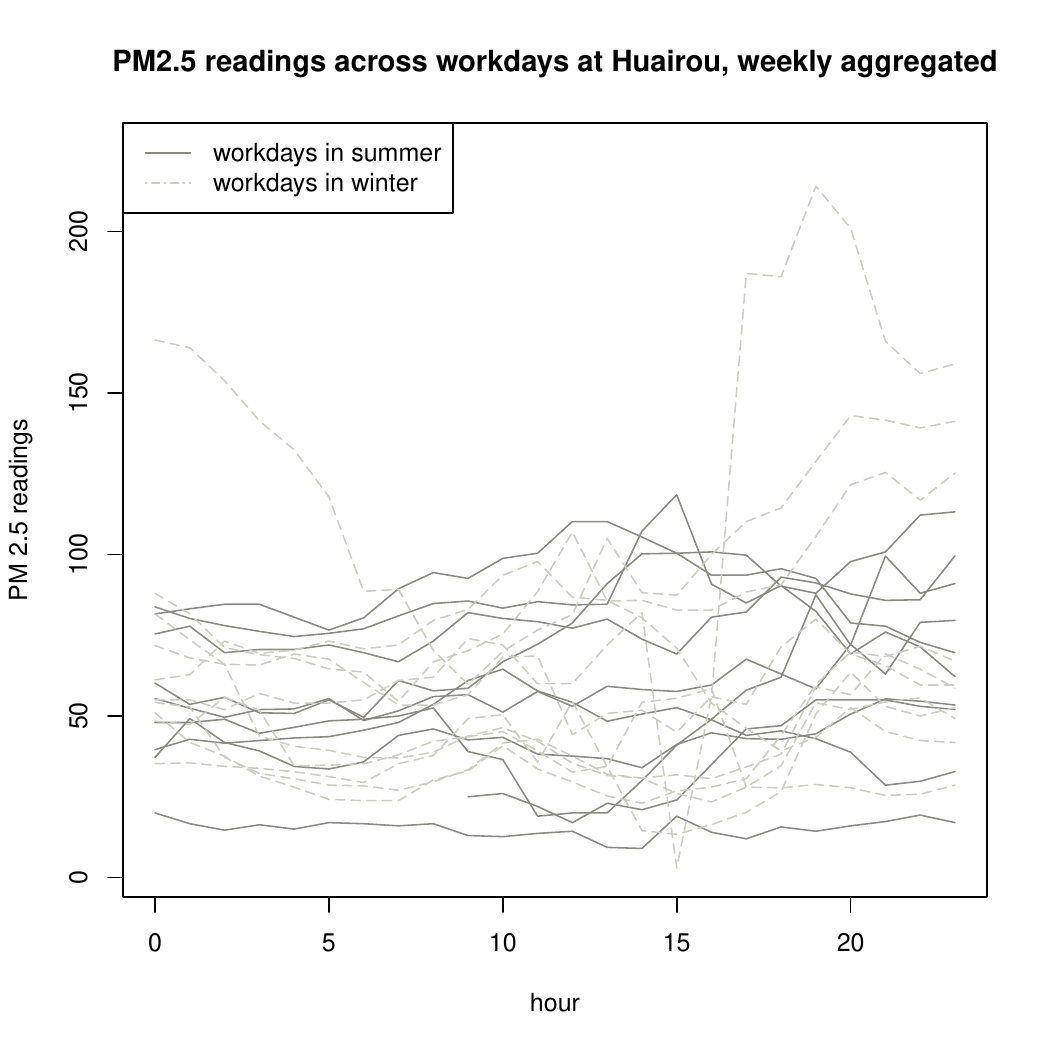}
    \caption{PM2.5 readings (weekly aggregate data) from 10 randomly selected weeks from summer (solid lines) and winter (dashed lines) in Huairou, Beijing.}
    \label{fig:airPollutants}
\end{figure}
Due to rapid industrialization and urbanization, numerous cities in China are experiencing severe air pollution, with PM$_{2.5}$ being a major pollutant. PM$_{2.5}$ refers to fine particulate matter (PM) with an aerodynamic diameter of less than 2.5 $\mu m$. These fine particles can penetrate the circulatory system directly, leading to the development of cardiovascular and respiratory diseases, and even lung cancers, as reported in studies like \citep{PopeJAMA2002, hoek_long-term_2013, Lelieveld2015Nature}.

Our investigation focuses on hourly air pollution data from Huairou, located near the northern part of Beijing, collected over a five-year period from March 2013 to February 2017. This data is accessible through the UCI repository \url{(https://doi.org/10.24432/C5RK5G)}. The high levels of PM$_{2.5}$ observed in this area can be attributed to factors like industrial emissions, human activities, and  winter heating, as discussed in \cite{LiangJRSSA2015} and \cite{Zhang2023Envir}.
\begin{figure}[!ht]
    \centering
    \includegraphics[scale = 0.67]{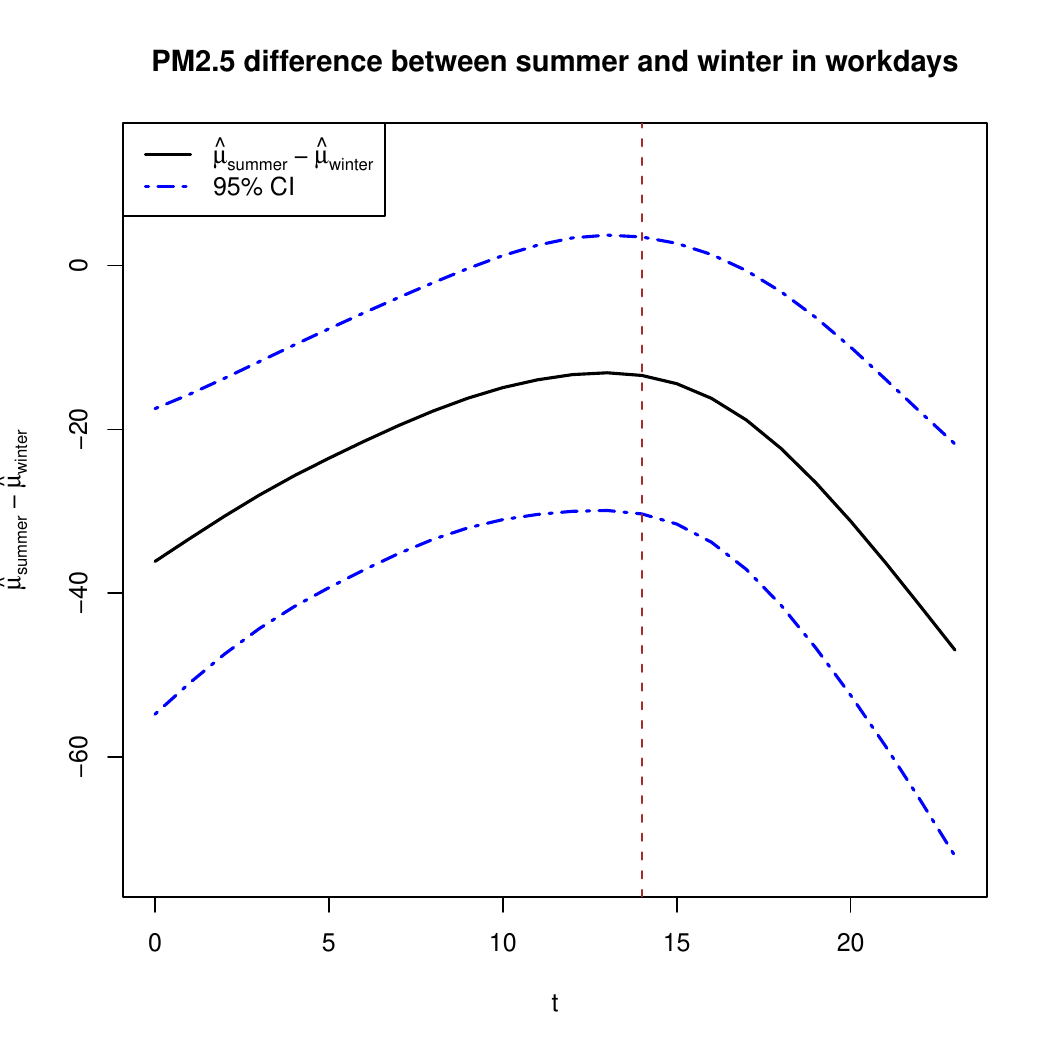}
    \caption{Estimated mean differences (black solid line) and $95\%$ confidence intervals (blue dashed lines) of hourly readings of PM$_{2.5}$ between winter and summer in Huairou. The vertical brown dash line indicates 14:00 at the local time.}
    \label{fig:airPollutantsEst}
\end{figure}
To analyze air pollution patterns in Huairou, we start by consolidating the daily data into weekly averages, reducing data dependency. Our analysis is  focused on workday data, from Monday to Friday, specifically during the summer months (June, July, and August) and winter months (December, January, February). From this weekly dataset, we calculate average hourly readings. As a result, for each week within these summer or winter periods, the PM$_{2.5}$ data comprise 24 hourly averages.

Lastly, we `sparsify' the dataset by randomly selecting between 2 and 10 observations per week, with each number being equally likely. Figure \ref{fig:airPollutants} illustrates these sparsified PM$_{2.5}$ readings for workdays, showcasing data from ten randomly selected weeks in both the summer and winter seasons. Notably, the PM$_{2.5}$ levels during certain hours in the winter months appear to be exceptionally high.

Figure \ref{fig:airPollutantsEst} displays the estimated mean difference in PM$_{2.5}$ levels between summer and winter, along with 95\% pointwise confidence intervals. Notably, this difference is more pronounced during nighttime, which could be due to the prevalent use of coal for heating in the winter. As the day progresses, the PM$_{2.5}$ disparity gradually lessens, suggesting that daytime human activities become a more significant source of PM$_{2.5}$ emissions. As expected, a global test confirms a significant difference in the average PM$_{2.5}$ levels between summer and winter, aligning with the findings in \cite{LiangJRSSA2015}.

\section{Concluding Remarks}\label{sec:conclusion}

In this paper, we introduce a statistical framework designed to identify mean differences between two groups of functional data, particularly in situations where each subject has only a limited number of irregularly spaced observations. Unlike typical fPCA-based approaches, our method does not require the assumption of a common covariance structure across two groups.

Specifically, under the RKHS framework, we build a functional Bahadur representation for the mean function estimator. This estimator is constructed by applying smoothing splines to the aggregated data within each group. The functional Bahadur representation facilitates the development of  pointwise limiting distributions and the weak convergence of the mean function estimator for each group. Subsequently, we present a testing procedure for the mean difference, as the functional data from these two groups are independent. 

It is worth emphasizing that the theoretical results obtained in this paper are in a distribution-wise manner. Specifically, the type-I error is controlled for sufficient large sample sizes under a fixed probability distribution. However, these asymptotic results generally fail to guarantee that the type-I error is well-controlled across a family of distributions, even when the sample size is large; see \cite{uniform_testing_2004} and \cite{waudbysmith2024arXiv1} for more details. Currently, many distribution-free methods focus on testing conditional independence \citep{conditional_independence_2022_JRSSB, conditional_independence_AOS2020}. In the context of FDA, the development of uniform inference procedures remains an underexplored area. 
It is still unclear what proper assumptions should be imposed to restrict the family of distributions to perform uniform inference for functional data. 
Most of current work in FDA on statistical inference focuses on asymptotic properties for a fixed probability distribution, as seen in \cite{ejs2sample} and \cite{JRSSC2016}. In future work, we plan to study these problems in depth.
%%%%%%%%%%%%%%%%%%%%%%%%%%%%%%%%%%%%%%%%%%%%%%
%% Example with single Appendix:            %%
%%%%%%%%%%%%%%%%%%%%%%%%%%%%%%%%%%%%%%%%%%%%%%
%\begin{appendix}
%\section*{Title}\label{appn} %% if no title is needed, leave empty \section*{}.
%Appendices should be provided in \verb|{appendix}| environment,
%before Acknowledgements.

%If there is only one appendix,
%then please refer to it in text as \ldots\ in the \hyperref[appn]{Appendix}.
%\end{appendix}
%%%%%%%%%%%%%%%%%%%%%%%%%%%%%%%%%%%%%%%%%%%%%%
%% Example with multiple Appendixes:        %%
%%%%%%%%%%%%%%%%%%%%%%%%%%%%%%%%%%%%%%%%%%%%%%
\newpage
\addcontentsline{toc}{section}{Appendices}
\setcounter{lemma}{0}
\setcounter{equation}{0}
\setcounter{figure}{0}
\setcounter{table}{0}
\renewcommand{\thelemma}{\Alph{section}.\arabic{lemma}}
\renewcommand{\theequation}{\Alph{section}.\arabic{equation}}
\renewcommand{\thefigure}{\Alph{section}.\arabic{figure}}
\renewcommand{\thetable}{\Alph{section}.\arabic{table}}

\begin{appendix}
Here is how the appendices are organized.
\begin{enumerate}
    \item[(i)] Appendix \ref{appendix:lemmas_sec3} contains the proofs for the lemmas in Section \ref{sec:theory}. 
    \item[(ii)] The proof of Theorem \ref{thm:FBR}, which establishes the functional Bahadur representation (FBR), is found in Appendix \ref{appendix:proof_thm3.1}. 
    \item[(iii)] Appendix \ref{appendix:proof_thm3.2} presents the proof of Theorem \ref{thm:ptwise_limiting_dist_biased}, discussing  the asymptotic pointwise normality based on FBR.
    \item[(vi)] The proof of Theorem \ref{thm:weak_conv}, addressing the weak convergence results for the mean estimator, is detailed in Appendix \ref{appendix:proof_thm3.4}.
    \item[(v)] The proofs for Theorem in Section 3.6, justifying the validity of the bootstrap procedure.
    \item[(vi)] Appendix \ref{appendix:proof_extra_technial_lemmas} includes various lemmas that provide technical support for the proofs.
    \item[(vii)] Lastly, Appendix \ref{appendix:extra_numerical_results} provides results of additional numerical studies.
\end{enumerate} 

Throughout this paper, we represent insignificant constants using $c_1, c_2, c_3, \ldots$. On the other hand, we denote important constants related to a specific variable by subscripting them with a letter, such as $C_{x}$ for constants associated with $x$. For example, $C_K$ represents a constant related to the maximum norm of the kernel function $K$ in a RKHS, where the inner product is defined in \eqref{eq:inner_prod}.

\section{Proofs for the Lemmas in Section \ref{sec:theory}}\label{appendix:lemmas_sec3}
\begin{proof}[Proof of Lemma \ref{prop:func_expression_by_h}]
According to Assumption \ref{assp:Fourier_expansion}, for any $g \in \HS$, 
\begin{align*}
    \norm{g}^2&= \innerproduct{g}{g} \\
    &= \innerproduct{\sum_{u}V(g, h_u)h_u}{\sum_{v}V(g, h_v)h_v}\\
    &= \sum_{u}\sum_v V(g, h_u)V(g, h_v)\innerproduct{h_u}{h_v} \\
    &= \sum_{u}\sum_v V(g, h_u)V(g, h_v)(\delta_{uv} + \lambda \gamma_u \delta_{uv}) \\
    &= \sum_{u}|V(g, h_u)|^2(1 + \lambda\gamma_u).
\end{align*}
Additionally, Assumption \ref{assp:Fourier_expansion} suggests $K_t = \sum_{v}V(K_t, h_v)h_v$ as $K_t \in \HS$. Thus,
\begin{align*}
    h_v(t) &= \innerproduct{K_t}{h_v} \\
    &= \innerproduct{\sum_{u}V(K_t, h_u)h_u}{h_v}\\
    &= \sum_{u}V(K_t, h_u)\innerproduct{h_u}{h_v} \\
    &=V(K_t, h_v)(1 + \lambda \gamma_v).
\end{align*}
Therefore, 
\[
V(K_t, h_v) = h_v(t)/(1 + \lambda\gamma_v) \text{~and~}
K_t = \sum_v \frac{h_v(t)}{1+\lambda\gamma_v}h_v.
\]

Regarding $W_\lambda h_v$,  Assumption \ref{assp:Fourier_expansion} ensures that $$\innerproduct{W_\lambda h_v}{h_u} = \lambda J(h_v, h_u) = \lambda \gamma_v \delta_{vu}.$$ On the other hand, $W_\lambda h_v = \sum_{k}V(W_\lambda h_v, h_k)h_k$. Therefore, 
\begin{align*}
    \innerproduct{W_\lambda h_v}{h_u} &= \innerproduct{\sum_{k} V(W_\lambda h_v, h_k)h_k}{h_u} \\
    &= \sum_k V(W_\lambda h_v, h_k)\innerproduct{h_k}{h_u} \\
    &= V(W_\lambda h_v, h_u)(1 + \lambda \gamma_u).
\end{align*}

In summary,
\begin{equation*}
    W_\lambda h_v = \sum_{u} V(W_\lambda h_v, h_u)h_u = \sum_u\frac{\lambda \gamma_v \delta_{vu}}{1+\lambda \gamma_u} h_u = \frac{\lambda \gamma_v}{1 + \lambda \gamma_v}h_v.
\end{equation*}
\end{proof}
Leveraging the results provided in the previous lemma, we can determine the upper bounds for various critical functions outlined in Lemma \ref{prop:K_and_Wf_bound}. This step is crucial in simplifying our forthcoming development of the FRB. 

\begin{proof}[Proof of Lemma \ref{prop:K_and_Wf_bound}]
According to Assumption \ref{assp:Fourier_expansion}, for any $f \in \HS$, we have $f = \sum_{v} V(f, h_v)h_v$. Therefore, in conjunction with Lemma \ref{prop:func_expression_by_h}, 
\begin{equation*}
 W_\lambda f = \sum_{v}V(f, h_v)W_\lambda h_v = \sum_v V(f, h_v)\frac{\lambda \gamma_v}{1 + \lambda \gamma_v}h_v.
\end{equation*}
Subsequently, 
\begin{align*}
    \norm{W_\lambda f}^2 &= \innerproduct{W_\lambda f}{W_\lambda f}\\
    &= \innerproduct{\sum_u V(f, h_u)\frac{\lambda \gamma_u}{1 + \lambda \gamma_u}h_u}{\sum_v V(f, h_v)\frac{\lambda \gamma_v}{1 + \lambda \gamma_v}h_v} \\
    &= \sum_u |V(f, h_u)|^2 \frac{\lambda^2\gamma_u^2}{1 + \lambda \gamma_u}.
\end{align*}
Furthermore, 
\begin{align*}
    \innerproduct{W_\lambda f}{f} &= \innerproduct{\sum_u V(f, h_u)\frac{\lambda \gamma_u}{1 + \lambda \gamma_u} h_u}{\sum_v V(f, h_v)h_v} \\
    &= \sum_u |V(f, h_u)|^2 \lambda \gamma_u.
\end{align*}
Since $\lambda$ and $\gamma_u$'s are all positive, we can then arrive at this conclusion, $$\norm{W_\lambda f}^2 \leq \innerproduct{W_\lambda f}{f}=\lambda J(f, f),\text{~for any~}f \in \HS.$$

Finally, as per Assumption \ref{assp:Fourier_expansion}, $\gamma_v \asymp v^{2m}$, implying  the existence of constants $c_0$ and $V_1$ such that 
$$\gamma_v \geq c_0 v^{2m}\text{~for all~}v \geq V_1.$$ Furthermore, we can identify constants $c_1$, $c_2$, and $V_2$ such that 
$$c_2 \geq \lambda V_2^{2m} \geq c_1.$$ Let's define $V= \max(V_1, V_2)$. Therefore, we can establish the following upper bound for the the norm of $K_t$:
\begin{align*}
    \norm{K_t}^2 &= \innerproduct{K_t}{K_t} \\
    &= \innerproduct{\sum_{v} \frac{h_v(t)}{1+\lambda\gamma_v}h_v(\cdot)}{\sum_{v} \frac{h_u(t)}{1+\lambda\gamma_u}h_u(\cdot)} \\
    &= \sum_{v}\frac{h_v^2(t)}{1 + \lambda \gamma_v} \\
    & \leq \left(\sup _{v\in \N} \norm{h_v}_{\sup}\right)^2\sum_{v}\frac{1}{1 + \lambda\gamma_v^{2m}}\\
    & \leq \left(\sup _{v\in \N} \norm{h_v}_{\sup}\right)^2 \left(\sum_{v=1}^{V}\frac{1}{1 + \lambda \gamma_v^{2m}} + \sum_{v=V}^{\infty}\frac{1}{1 + \lambda \gamma_v^{2m}}\right)\\
    & \leq \left(\sup _{v\in \N} \norm{h_v}_{\sup}\right)^2 \left(V + \sum_{v=V}^{\infty} \frac{1}{\lambda c_0 v^{2m}}\right)\\
    & \leq \left(\sup _{v\in \N} \norm{h_v}_{\sup}\right)^2 \left(V + \frac{1}{\lambda}\frac{c_4}{V^{2m-1}}\right) \\
    & \leq \left(\sup _{v\in \N} \norm{h_v}_{\sup}\right)^2 c_5 V^{-1} \\
    & \leq C_K^2 h^{-1}.
\end{align*}
\end{proof}
\begin{proof}[Proof of Lemma \ref{lem:exist_of_eigensys}]
This is a direct consequence of Proposition 2.2 as stated in \cite{zfsannals13}.
\end{proof}

\begin{proof}[Proof of Lemma \ref{lem:concen_inequa}]
This proof is directly derived from the proof of Lemma 3.2 presented in \cite{zfsannals13}.
\end{proof}

The upcoming proof is a vital step in deriving the FBR, which is a first-order approximation for the smoothing-spline mean estimator from \eqref{eq:mean_optimization}. Let's introduce some key notations. Recall that $D\ell(f, \lambda)\Delta g = \innerproduct{S_{M, \lambda}(f)}{\Delta g}$ where 
\begin{align*}
    S_{M, \lambda}(f) &= -\frac{1}{M}\sum_{i=1}^{n}\sum_{j=1}^{N_i}\{Y_{ij} - f(T_{ij})\}K_{T_{ij}} + W_\lambda f \in \HS\\
    &\triangleq S_M(f) + W_\lambda f.
\end{align*}
Furthermore, we define 
$$\E\{S_{M, \lambda}(f)(t)\} \triangleq S_{\lambda}(f)(t)\text{~and~}\E\{S_{M}(f)(t)\} \triangleq S(f)(t)\text{~for all~}t \in \T.$$ 
For ease of notation, let's define the following integrals:
\begin{equation}\label{eq:I_and_G_intgrals}
\begin{aligned}
     I_{t,s} &= \int_{t,s \in \T}c(t,s)\pi(t)\pi(s)dtds,\\
    I_{t, t} &= \int_{t \in \T}c(t,t)\pi(t)\pi(t)dtds,\\
    \G_1(s, t) &= \int_{(u, v)\in \T^2}c(u,v)K(u, t)K(v,s)\pi(u)\pi(v)dudv,\\
    \G_2(s, t) &= \int_{\T}c(u,u)K(u, t)K(u,s)\pi(u)du,\\
    \G_3(s,t) &= \int_{\T} K(u,t)K(u,s)\pi(u)du.
\end{aligned}
\end{equation}
It's evident that $\abs{I_{t,s}}$ and $\abs{I_{t,t}}$ have upper bounds, as $\pi(t)$ and $c(t,s)$ have upper bounds by Assumptions \ref{assp:data_gm} and \ref{assp:smooth_of_X}, respectively.

\begin{proof}[Proof of Lemma \ref{lem:convergence}]
The proof consists of the following two parts. 
\begin{enumerate}
    \item[1.]The first part aims to find $f_\lambda$ s.t. $S_\lambda(f_\lambda) = 0$ where 
    $$S_\lambda(f)(t) = \E\{S_{M,\lambda}(f)(t)\}\text{~for all~}t \in \T.$$
    It should be noted that $S_{M,\lambda}(f)$, as per the Riesz representation theorem, is the unique representer determined by the first Fr\'{e}chet derivative of the loss function $\ell(f, \lambda)$ in \eqref{eq:mean_optimization}. Here $S_\lambda(f_\lambda)$ represents the expectation of this representer. This implies that $f_\lambda$ can serve as an estimator when ``$n\to\infty$''. The discrepancy between $f_\lambda$ and the true mean function $\mu$ reflects the bias introduced by the penalty term in the model outlined in \eqref{eq:model}. 
    \item[2.]The second part aims to bound the distance between $f_\lambda$ and the finite-sample mean estimator $\hat{\mu}_\lambda$. This distance illustrates the impact of having only a finite number of data points.
\end{enumerate}

Both parts utilize the following Banach contraction mapping theorem (CMT):

\textit{Suppose a mapping $T: \mathcal{X} \to \mathcal{X}$ satisfies 
$$\norm{T(x) - T(y)} \leq k\norm{x-y}\text{~for some~}k \in [0, 1),$$ 
where $\mathcal{X} \subseteq \HS$ is a closed set. Then, there is a unique $x_0$ s.t. $T(x_0) = x_0$.}

For more details on the CMT, readers are directed to \cite[Chapter 9, The Contraction Principle]{rudin1976principles}. 

\begin{proofpart}{1}{} Let us first prove the existence of $f_\lambda$. \\

Denote $B_1 \triangleq B_1(r_{1}) := \{g \in \HS \mid \norm{g} \leq r_{1}\}$ where 
$$r_{1} = \sqrt{\lambda\{J(\mu, \mu) + 1\}} = h^{m}\sqrt{\{J(\mu, \mu) + 1\}}$$ and $T_1(\delta) = \delta - S_{\lambda}(\mu + \delta)$ for any $\delta \in \HS$. By the triangle inequality, we have 
$$\norm{T_1(\delta)} = \norm{\delta - S_\lambda(\mu + \delta)} \leq  \norm{\delta - S_\lambda(\mu + \delta) + S_\lambda(\mu)} + \norm{S_\lambda(\mu)}.$$ 
Note that, $\norm{S_{\lambda}(\mu)} = \norm{S(\mu) + W_\lambda\mu}$ where 
\begin{align*}
    S(\mu) &= -\E[\{Y - \mu(T)\}K_T] \\
    &= -\E\left(\E\Big[\{Y - \mu(T)\}K_T~\big|~T\Big]\right)\\ 
    &= 0.
\end{align*}
Due to Lemma \ref{prop:K_and_Wf_bound}, we have
\begin{equation}\label{eq:S_lambda_mu_norm_eq1}
    \norm{S_\lambda(\mu)} = \norm{W_\lambda \mu} \leq \sqrt{\lambda J(\mu, \mu)} < r_{1}.
\end{equation}
Additionally, $DS_\lambda(\mu)\delta = \delta$ and $D^{2}S_\lambda(\mu + ss^{'} \delta) = 0$ as $DS_\lambda(f) = \mathrm{id}$ for any $f \in \HS$ by Lemma \ref{prop:identity_op}. Thus,
\begin{equation} \label{eq:S_lambda_mu_norm_eq2}
    \norm{\delta - S_\lambda(\mu + \delta) + S_\lambda(\mu)} = \norm[\Big]{\delta - (DS_\lambda(\mu)\delta + \int_{\T}\int_{\T}sD^{2}S_\lambda(\mu + ss'\delta)\delta\delta dsds')} = 0.
\end{equation}
Combining \eqref{eq:S_lambda_mu_norm_eq1} and \eqref{eq:S_lambda_mu_norm_eq2}, we can establish that $T_1(B_1) \subseteq B_1$ for every $\delta \in B_1$.\\

To conclude that $T_1$ is a contraction, we consider the following: 
\begin{align*}
    T_1(\delta_1) - T_1(\delta_2) &= \delta_1 - \delta_2 - S_\lambda(\mu + \delta_1) + S_\lambda(\mu + \delta_2) \\
    &= \delta_1 - \delta_2 + \E\left[\left\{Y - (\mu + \delta_1)(T)\right\}K_{T}\right] - W_\lambda(\mu + \delta_1) \\
    &- \E\left[\left\{Y - (\mu + \delta_2)(T)\right\}K_{T}\right] + W_\lambda(\mu + \delta_2)\\
    &= \delta_1 - \delta_2 + \E\left[\left\{(\delta_2 - \delta_1)(T)\right\}K_{T}\right] + W_\lambda(\delta_2 - \delta_1).
\end{align*}
Therefore, for any $g \in \HS$, 
\begin{align*}
&~~~~\innerproduct{T_1(\delta_1) - T_1(\delta_2)}{g} \\
&= \innerproduct{\delta_1 - \delta_2}{g} + \E\left[\left\{(\delta_2 - \delta_1)(T)\right\}\innerproduct{K_{T}}{g}\right] + \innerproduct{W_\lambda(\delta_2 - \delta_1)}{g} \\
&= \innerproduct{\delta_1 - \delta_2}{g} + \innerproduct{\delta_2 - \delta_1}{g}\\
& =0.
\end{align*}
The calculations above indicate that 
$$T_1(\delta_1) - T_1(\delta_2) = 0\text{~for any~}\delta_1, \delta_2 \in B_1.$$ Thus, $T_1$ is a contraction defined on $B_{1}$.\\

Concluding Part 1, we apply the CMT, leading to the assertion that there exits a $\delta_\lambda$ such that $T_1(\delta_\lambda) = \delta_\lambda$. This implies $S_\lambda(\mu + \delta_\lambda) = 0$. Defining $f_\lambda = \mu + \delta_\lambda$,  we can then conclude that $$\norm{\mu - f_\lambda} \leq r_{1} = O(h^{m}).$$
\end{proofpart}

\begin{proofpart}{2}{}
We will define a bounded linear operator $T_2$ such that $$T_2(\delta) = \delta - S_{M, \lambda}(f_\lambda + \delta)\text{~for any~}\delta \in \HS.$$ 
It is apparent that 
$$\norm{T_2(\delta)} \leq \norm{S_{M, \lambda}(f_\lambda)} + \norm{\delta - S_{M, \lambda}(f_\lambda + \delta) + S_{M, \lambda}(f_\lambda)},$$ 
with $f_\lambda$ derived from Part 1. 

To establish an appropriate radius of $B_{2} \triangleq B(r_{2})$ so that $T_2$ is a contraction on $B_2$, we will consider $\E\{\norm{S_{M, \lambda}(f_\lambda)}^2\}$. Recalling $\E\{S_{M, \lambda}(f)\} = S_{\lambda}(f)$, we define 
$$O_{ij} = \left\{Y_{ij} - f_\lambda(T_{ij})\right\}K_{T_{ij}}\text{~and~}D_{ij} =O_{ij} - \E(O_{ij}).$$ Since $S_{\lambda}(f_\lambda) = 0$ as shown in Part 1, we proceed with the following
\begin{align*}
    &\hspace{1.2em} \E\{\norm{S_{M, \lambda}(f_\lambda)}^2\} \\
    &= \E\{\norm{S_{M, \lambda}(f_\lambda) - S_{\lambda}(f_\lambda)}^2\}\\
    &= \E\{\norm{S_{M}(f_\lambda) - S(f_\lambda)}^2\}\text{ where } S(f) \text{ is calculated in } \eqref{eq:S(f)}\\
    &= \E\left[\norm[\bigg]{\frac{1}{M}\sum_{i=1}^{n}\sum_{j=1}^{N_i}\Big\{O_{ij} - \E(O_{ij})\Big\}}^2\right]\\
    &= \E\sbrac[\bigg]{\frac{1}{M^2}\norm[\Big]{\sum_{i=1}^n\brac{\mathcal{W}_i - \E(\mathcal{W}_i)}}^2}\\
    &\leq \E\brac{\frac{n}{M^2}\norm[\big]{\mathcal{W}_1 - \E(\mathcal{W}_1)}^2}\\
    &\leq \E\brac{\frac{n}{M^2}\norm{\mathcal{W}_1}^2}, \label{eq:simplified_E(S_M)}\numberthis
\end{align*}
where $\mathcal{W}_i = \sum_{j=1}^{N_i}O_{ij}$ for $i=1, 2, \ldots, n$. The expectation of this norm is calculated as follows: 
\begin{align*}
    \E\left(\norm[\Big]{\sum_{j=1}^{N_1}O_{1j}}^2\right) &= \sum_{k, l}\E\left\{\innerproduct{O_{1k}}{O_{1l}} \right\}\\
    &= \sum_{k,l}\E\Big[\{Y_{1k} - f_\lambda(T_{1k})\}\{Y_{1l} - f_\lambda(T_{1l})\}\innerproduct{K_{T_{1k}}}{K_{T_{1l}}}\Big]\\
    &\leq \underset{t \in \T}{\sup}\norm{K_t}^2 \sum_{k,l}\E\Big[\{Y_{1k} - f_\lambda(T_{1k})\}\{Y_{1l} - f_\lambda(T_{1l})\}\Big], \numberthis
\end{align*}
where the expectation in the final step can be reformulated as follows
\begin{align}
    & ~~~~\E\big[\{Y_{1k} - f_\lambda(T_{1k})\}\{Y_{1l} - f_\lambda(T_{1l})\}\big] \nonumber\\
    &= \E\big[\{Y_{1k} - \mu(T_{1k}) + \mu(T_{1k}) - f_\lambda(T_{1k})\}\{Y_{1l} - \mu(T_{1l}) + \mu(T_{1l}) - f_\lambda(T_{1l})\}\big]\nonumber\\
    &= \E[\{Y_{1k} - \mu(T_{1k})\}\{Y_{1l} - \mu(T_{1l})\}] + \E[\{Y_{1k} - \mu(T_{1k})\}\{\mu(T_{1l}) - f_\lambda(T_{1l})\}] \nonumber\\
    &+ \E[\{Y_{1l} - \mu(T_{1l})\}\{\mu(T_{1k}) - f_\lambda(T_{1k})\}]\nonumber\\
    &+ \E[\{\mu(T_{1k}) - f_\lambda(T_{1k})\}\{\mu(T_{1l}) - f_\lambda(T_{1l})\}]. 
\end{align}
Observe that both terms involving cross products are zero. For instance,
\begin{align*}
    &~~~~\E[\{Y_{1k} - \mu(T_{1k})\}\{\mu(T_{1l}) - f_\lambda(T_{1l})\}]\\
    &= \E\Big(\E[\{Y_{1k} - \mu(T_{1k})\}\{\mu(T_{1l}) - f_\lambda(T_{1l})\}\mid T]\Big) \\
    &=\E\Big(\{\mu(T_{1l}) - f_\lambda(T_{1l})\}\E\{Y_{1k} - \mu(T_{1k})\mid T\}\Big) \\
    &= 0. \numberthis
\end{align*}

Additionally,
$$\E[\{\mu(T_{1k}) - f_\lambda(T_{1k})\}\{\mu(T_{1l}) - f_\lambda(T_{1l})\}] = \E^2\{\mu(T) - f_\lambda(T)\}$$ because $T_{1k}$ and $T_{1l}$ are i.i.d. for all observations regardless of whether they are from the same or different subjects. Furthermore, based on the conclusions from Part 1, we can establish the following results:
\begin{align*}
    \E^2\{\mu(T) - f_\lambda(T)\} &\leq \E[\{\mu(T) - f_\lambda(T)\}^2]\\
    & = V(\mu - f_\lambda, \mu - f_\lambda)\\
    &\leq \norm{\mu - f_\lambda}^2 \\
    &= r_{1N}^2. \numberthis
\end{align*}

The remaining term to consider is $\E[\{Y_{1k} - \mu(T_{1k})\}\{Y_{1l} - \mu(T_{1l})\}]$. Note that,
\begin{align*}
    &~~~~\E[\{Y_{1k} - \mu(T_{1k})\}\{Y_{1l} - \mu(T_{1l})\}]\\ &= \E\left(\E\Big[\{Y_{1k} - \mu(T_{1k})\}\{Y_{1l} - \mu(T_{1l})\}\mid T_{1k}, T_{1l}\Big]\right)\\
    &= \E\{c(T_{1k}, T_{1l}) + \delta_{kl}\sigma_{\eps}^2\}\\
    &= \iint_{t,s \in \T} c(t, s)\pi (t)\pi (s)dtds + \sigma_{\eps}^2\delta_{kl} \\
    &= I_{t,s} + \sigma_{\eps}^2\delta_{kl}\\
    &< \infty, \numberthis \label{eq:cov_dist}
\end{align*}
where $c(t,s)$ is the covariance function of $X(t)$, and $\delta_{kl} =1$ if $k=l$ and $0$ otherwise.

Combining the results from equations \eqref{eq:simplified_E(S_M)} - \eqref{eq:cov_dist}, we have
\begin{align*}
    \E\left(\norm[\Big]{\sum_{j=1}^{N_1}O_{ij}}^2\right) \leq \underset{t\in \T}{\sup}\norm{K_t}^2\Big[& N_1^2 \E^2\{\mu(T) - f_\lambda(T)\} + N_1 \sigma^2_{\eps} \nonumber \\
    &+ N_1\E\brac{c(T, T)}+ (N_1^2 - N_1)\E\brac{c(T_1, T_2)}\Big].
\end{align*}
Additionally, we have
\begin{align*}
    &~~~~\E\left\{\frac{n}{M^2} \E\Big(\norm{\mathcal{W}_1}^2\Big)\right\} \\
    &\leq \E\left\{\underset{t \in \T}{\sup}\norm{K_t}^2\frac{n}{M^2}\left\{N_1^2r_{1N}^2 + N_1\sigma^2_{\eps} + N_1I_{t, t} + (N_1^2 - N_1)I_{t,s}\right\}\right\}\\
    &\leq \frac{nC_K^2h^{-1}}{n^2}\left\{(\sigma_N^2 + \mu_N^2)h^{2m} + \mu_N\sigma_{\eps}^2 + \mu_N I_{t,t} + (\sigma_N^2 + \mu_N^2 - \mu_N)I_{t,s}\right\}\\
    &= O((nh)^{-1}),
\end{align*}
where the final step is based on the established bounds for $I_{t,t}$ and $I_{t,s}$. Consequently, this leads to the following:
$$\E\{\norm{S_{M,\lambda}(f_\lambda)}^2\} = O((nh)^{-1}),$$ which implies $\norm{S_{M,\lambda}(f_\lambda)}^2 = O_{p}((nh)^{-1})$. \\

This result enables us to choose a radius of $B_2$, 
$$B_2(r_{2}) \equiv \{f\in \HS \mid \norm{f} \leq r_{2}\}\text{~where~}r_{2} = 2C_{B_2}(nh)^{-1/2}.$$ 
Therefore, we ascertain that, with a probability approaching one,
$$\norm{S_{M,\lambda}(f_\lambda)} < C_{B_2}(nh)^{-1/2}.$$ 

 The next step is to handle $\norm{\delta - S_{M, \lambda}(f_\lambda + \delta) + S_{M, \lambda}(f_\lambda)}$. As in Part 1, using Taylor expansion results in
\begin{equation*}
\begin{aligned}
    \norm{\delta - S_{M, \lambda}(f_\lambda + \delta) + S_{M, \lambda}(f_\lambda)} \leq & \norm{\delta - DS_{M, \lambda}(f_\lambda)\delta} + \\
    &\norm[\bigg]{\iint_{s, s^{'} \in \T}sD^{2}S_{M, \lambda}(f_\lambda + ss'\delta)\delta\delta dsds'}.
\end{aligned}
\end{equation*}
The latter term becomes zero as $D^{2}S_{M, \lambda}(f) = 0$ for any $f \in \HS$. Observing that $DS_{\lambda}(f_\lambda) = \mathrm{id}$, we can then rewrite this as
\begin{align*}
    \norm{\delta - DS_{M, \lambda}(f_\lambda)\delta} &= \norm[\Big]{DS_\lambda(f_\lambda)\delta - DS_{M, \lambda}(f_\lambda)\delta} \\
    & = \norm[\bigg]{\frac{1}{M}\sum_{i=1}^{n}\sum_{j=1}^{N_i}\left[\delta(T_{ij})K_{T_{ij}} - \E  \{\delta(T)K(T)\}\right]}. \numberthis\label{eq:diff_in_empirical_process}
\end{align*}
Let $\psi_n(T, \tilde{\delta}) = C_{K}^{-1}h^{1/2}\tilde{\delta}(T)$ where $\tilde{\delta}(T)= C_K^{-1}h^{1/2}\norm{\delta}^{-1}\delta(T)$ for any $\delta \in B_{2}$. Note that,
\begin{align*}
    \norm{\tilde{\delta}}_{\sup} &= C_K^{-1}h^{1/2}\norm{\delta}^{-1}\norm{\delta}_{\sup}\\
    &\leq C_K^{-1}h^{1/2}\norm{\delta}^{-1} C_K h^{-1/2}\norm{\delta} \leq 1,
\end{align*}
where the final line employs Lemma \ref{prop:sup_norm}. Additionally, $$J(\tilde{\delta}, \tilde{\delta}) \leq \lambda^{-1}\norm{\tilde{\delta}}^2 = \lambda^{-1}C_K^{-2}h\norm{\delta}^{-2}\norm{\delta}^2 \leq \lambda^{-1}C_K^{-2}h.$$ Thus, $\tilde{\delta}(T) \in \G$. Furthermore,
\begin{equation*}
    \abs{\psi_n(T, \tilde{\delta}_1) - \psi_n(T, \tilde{\delta}_2)} \leq C_K^{-1}h^{1/2}\norm{\tilde{\delta}_1-\tilde{\delta}_2}_{\sup},
\end{equation*}
which satisfies the Lipschitz condition in Lemma \ref{lem:concen_inequa}. Therefore, for any $\delta \in B_{2}$ and sufficiently large $n$, we have the following result:
\begin{equation*}
\begin{aligned}
    &~~~~\norm[\Big]{\sum_{i=1}^{n}\sum_{j=1}^{N_i}\psi_n(T_{ij}, \tilde{\delta})K_{T_{ij}} - \E\{\psi_n(T, \tilde{\delta})K(T)\}} \\
    &\leq c_1\brac{M^{1/2}h^{-\frac{2m-1}{4m}} + \left(\frac{M}{n}\right)^{1/2}}(5\log\log M)^{1/2}.
\end{aligned}
\end{equation*}
Note that, $\psi_n(T, \tilde{\delta}) = C_K^{-2}\norm{\delta}^{-1}h\delta(T)$. Therefore,
\begin{align*}
    (\ref{eq:diff_in_empirical_process}) &\leq \sqrt{5}c_1C_K^2(Mh)^{-1}\brac{h^{-\frac{2m-1}{4m}}M^{1/2} + \left(\frac{M}{n}\right)^{1/2}}(\log\log M)^{1/2} \norm{\delta} \\
    &\leq o_p(1)\norm{\delta}.
\end{align*}

Therefore, $T_2(B_2) \subseteq B_2$ with a probability that approaches one as sample size increases infinitely. To conclude Part 2, it's necessary to verify that $T_2$ is a contraction. Note that,
\begin{align*}
    &~~~~S_{M, \lambda}(f_\lambda + \delta_1) - S_{M, \lambda}(f_\lambda + \delta_2)\\
    &= S_{M,\lambda}(f_\lambda) + DS_{M,\lambda}(f_\lambda)\delta_1 + \int_{0}^{1}\int_{0}^{1}sD^{2}S_{M, \lambda}(f_\lambda + ss'\delta_1)\delta_1\delta_1 dsds'\\
    &- \left\{S_{M,\lambda}(f_\lambda) + DS_{M,\lambda}(f_\lambda)\delta_2 + \int_{0}^{1}\int_{0}^{1}sD^{2}S_{M, \lambda_2}(f_\lambda + ss'\delta_2)\delta_2\delta_2 dsds'\right\}\\
    &= DS_{M,\lambda}(f_\lambda)(\delta_1 - \delta_2).
\end{align*}
For any $\delta_1, \delta_2 \in B_2$. If $\delta_1$ and $\delta_2$ are identical, then $\norm{T_2(\delta_1) - T_2(\delta_2)}$ equals zero. Otherwise, when $\norm{\delta_1 -\ \delta_2} > 0$,
\begin{align*}
    \norm{T_2(\delta_1) - T_2(\delta_2)} &= \norm{\delta_1 - \delta_2 - \{S_{N, \lambda}(f_\lambda + \delta_1) - S_{M, \lambda}(f_\lambda + \delta_2)\}}\\
    &\leq \norm{(\delta_1 - \delta_2) - DS_{M, \lambda}(f_\lambda)(\delta_1 - \delta_2)}\\
    &= \norm{DS_\lambda(f_\lambda)(\delta_1 - \delta_2) - DS_{M, \lambda}(f_\lambda)(\delta_1 - \delta_2)}.
\end{align*}
The same arguments used for equation (\ref{eq:diff_in_empirical_process}) apply here. Let $u = \delta_1 - \delta_2$ and $\tilde{u} = C_{K}^{-1}h^{1/2}\norm{u}^{-1}u(T)$, while keeping $\psi(T, f)$ unchanged. It is straightforward to verify that $\tilde{u} \in \G$. Therefore,
\begin{align*}
    &~~~\norm{u - DS_{M, \lambda}(f_\lambda)u} \\
    &\leq \sqrt{5}c_2C_K^2(Mh)^{-1}\brac{h^{-\frac{2m-1}{4m}}M^{1/2} + \left(\frac{M}{n}\right)^{1/2}}(\log\log M)^{1/2} \norm{u}\\
    & = o_p(1)\norm{u}.
\end{align*}
This suggests that, in probability, $T_2$ is indeed a contraction mapping on $B_2$. 

Consequently, by applying the fixed point theorem, it follows that there exists a $\delta_{M, \lambda} \in B_2$ such that $T_2(\delta_{M,\lambda}) = \delta_{M, \lambda}$. This means that $S_{M, \lambda}(f_\lambda + \delta_{M, \lambda}) = 0$. Additionally, $S_{M, \lambda}(\hat{\mu}_\lambda) = 0$. Therefore, we conclude that $f_\lambda + \delta_{M, \lambda} = \hat{\mu}_\lambda$, which implies $\norm{f_\lambda - \hat{\mu}_\lambda} \leq r_{2} = O((nh)^{-1})$.
\end{proofpart}

To summarize, {Part 1} establishes that $\norm{f_\lambda - \mu} = O(h^{m})$, and {Part 2} demonstrates that $\norm{f_\lambda - \hat{\mu}} = O_p((nh)^{-1})$, where $\mu$ is the true mean function. Combining these two results with the triangle inequality leads to 
$$\norm{\hat{\mu} - \mu} = O_p(h^{m} + (nh)^{-1/2}).$$ 
In particular, the optimal convergence rate is achieved when $h = O(n^{1/(2m+1)})$, or equivalently, when $\lambda = O(n^{2m/(2m+1)})$.
\end{proof}

\section{Proof of Theorem \ref{thm:FBR}}\label{appendix:proof_thm3.1}

We can now establish a key theorem for this paper, the \textit{functional Bahadur representation}. This theorem provides a first-order approximation of the mean estimator.
\begin{proof}[Proof of Theorem \ref{thm:FBR}]
Let's represent the true mean function as $\mu$, and denote $\hat{\mu}_\lambda$ as the estimator derived from equation \eqref{eq:mean_optimization}. Additionally, define $\mu_\delta = \hat{\mu}_\lambda - \mu$. First, notice that
\begin{equation}\label{eq:first_order_diff}
\begin{aligned}
&~~~~S_{M,\lambda}(\hat{\mu}_\lambda) - S_{\lambda}(\hat{\mu}_\lambda) - \{S_{M,\lambda}(\mu) - S_{\lambda}(\mu)\} \\ 
    &= S_{M}(\hat{\mu}_\lambda) - S(\hat{\mu}_\lambda) - \{S_{M}(\mu) - S(\mu)\}.
\end{aligned}
\end{equation}
Given that $S_{M,\lambda}(\hat{\mu}_\lambda) = 0$, the left-hand-side of \eqref{eq:first_order_diff} becomes
\begin{align*}
    \text{LHS of } \eqref{eq:first_order_diff} &= - S_\lambda(\mu) - DS_{\lambda}(\mu)\mu_\delta - \iint_{s, s' \in \T}sD^{2}S_{\lambda}(\mu + ss'\mu_\delta)\mu_\delta\mu_\delta dsds'\\
    &\hspace{1.1em}-\{S_{M,\lambda}(\mu) - S_{\lambda}(\mu)\}\\
    &= -\left\{DS_{\lambda}(\mu)\mu_\delta + S_{M, \lambda}(\mu)\right\} \\
    &= -\{\mu_\delta + S_{M, \lambda}(\mu)\},
\end{align*}
where $D^{2}S_{\lambda}(f) = 0$ as $DS_{\lambda}(f) = \mathrm{id}$ for any $f \in \HS$. 

The calculation above suggests that to complete the proof, it is sufficient to establish an upper bound for the norm of the right-hand-side of \eqref{eq:first_order_diff}. Note that,
\begin{align*}
    \text{RHS of }\eqref{eq:first_order_diff} &= -\frac{1}{M}\sum_{i=1}^{n}\sum_{j=1}^{N_i}\left[\{Y_{ij} - \hat{\mu}_\lambda(T_{ij})\}K_{T_{ij}} -  \{Y_{ij} - {\mu}(T_{ij})\}K_{T_{ij}}\right]\\
    &- \big(\E[\{Y-\hat{\mu}_\lambda(T)\}K_T] - \E[\{Y-\mu(T)\}K_T]\big)\\
    &= \frac{1}{M}\sum_{i=1}^{n}\sum_{j=1}^{N_i}\big[\mu_\delta(T_{ij})K_{T_{ij}} - \E\{\mu_\delta(T)K_{T}\}\big]. \label{eq:first_order_diff_woPenalty_simplified}\numberthis
\end{align*}
Therefore, we can utilize the same reasoning as in Lemma \ref{lem:convergence}. Specifically, define 
$$\tilde{\mu}_\delta(T) = C_K^{-1}h^{1/2}\norm{\mu_\delta}^{-1}\mu_\delta(T)\text{~and~}\psi(T, f) = C_K^{-1}h^{1/2}f(T).$$ It's straightforward to verify that $\tilde{\mu}_\delta \in \G$, and $\psi(T, f)$ satisfies the Lipschitz condition outlined in Lemma \ref{lem:concen_inequa}. Thus, with a probability approaching one,
\begin{equation*}
\begin{aligned}
    &~~~~\norm[\Big]{\sum_{i=1}^{n}\sum_{j=1}^{N_i}\big[\psi(T_{ij}, \tilde{\mu}_\delta)K_{T_{ij}} - \E\{\psi(T, \tilde{\mu}_\delta)K_{T}\}\big]} \\
    &\leq c_1\brac{M^{1/2}h^{-\frac{2m-1}{4m}} + \left(\frac{M}{n}\right)^{1/2}}(5\log\log M)^{1/2}.
\end{aligned}
\end{equation*}
It's worth noting that $$\psi(T, \tilde{\mu}_\delta) = C_K^{-2}h\norm{\mu_\delta}^{-1}\mu_\delta(T),$$ which implies
\begin{align*}
    \eqref{eq:first_order_diff_woPenalty_simplified} &\leq \sqrt{5}c_1C_K^2(Mh)^{-1}\brac{M^{1/2}h^{-\frac{2m-1}{4m}} + \left(\frac{M}{n}\right)^{1/2}}(\log\log M)^{1/2}\norm{\mu_\delta}.
\end{align*}
We know that $\norm{\mu_\delta} \leq c_2\{h^m + (nh)^{-1/2}\}$ with a probability approaching one, as indicated by Lemma \ref{lem:convergence}. Therefore,
\begin{equation*}
    \eqref{eq:first_order_diff} = O_p(a_n'),
\end{equation*}
where 
\begin{align*}
    a_n' &= \frac{h^m + (nh)^{-1/2}}{Mh}\brac{M^{1/2}h^{-\frac{2m-1}{4m}} + \left(\frac{M}{n}\right)^{1/2}}(5\log\log M)^{1/2}\\
    &\leq O_p(n^{-1/2}h^{-(6m-1)/4m}\{h^m + (nh)^{-1/2}\}(\log \log n)^{1/2})\\
    &= O_p(a_n).
\end{align*}
This concludes the proof for the theorem.
\end{proof}
\section{Proof of Theorem \ref{thm:ptwise_limiting_dist_biased}}\label{appendix:proof_thm3.2}
As discussed in the main text, FBR offers a first-order approximation for the smoothing-spline estimators. Given that $\hat{\mu}_\lambda = \mu + S_{M, \lambda}(\mu) + o_p(1)$ where
$$S_{M, \lambda}(f) = -\frac{1}{M}\sum_{i=1}^{n}\sum_{j=1}^{N_i}\{Y_{ij} - f(T_{ij})\}K_{T_{ij}} + W_\lambda f,$$ we can apply Lindeberg's Central Limit Theorem (CLT) to deduce the limiting distribution at any $t \in \T$. Let's now present the detailed proof.
\begin{proof}[Proof of Theorem \ref{thm:ptwise_limiting_dist_biased}]
Let's consider $r_\delta(t) = \hat{\mu}_{\lambda}(t) - \mu_b(t) + S_M(\mu)(t)$, where $\mu_b(t) = (\mathrm{id} - W_\lambda)\mu(t)$ and 
$$S_M(\mu)(t) = -\frac{1}{M}\sum_{i=1}^{n}\sum_{j=1}^{N_i}\{Y_{ij} - \mu(T_{ij})\}K_{T_{ij}}(t)\text{~for any~}t \in \T.$$ 
We can re-write $r_\delta$ as follows:
\begin{align*}
    r_\delta &= \hat{\mu}_{\lambda} - \mu_b + S_{M}(\mu)= (\hat{\mu}_{\lambda} - \mu) + S_{M, \lambda}(\mu).
\end{align*}
According to Theorem \ref{thm:FBR} and the assumptions in Theorem \ref{thm:ptwise_limiting_dist_biased}, we understand that $\norm{r_\delta} = O_p(a_n) = o_p(n^{-1/2})$. Regarding $S_{M, \lambda}(\mu)$, define $D_{ij}(\mu) \triangleq \{Y_{ij} - \mu(T_{ij})\}K_{T_{ij}}$. Then,
\begin{align*}
    \E\big[\norm{S_{M}(\mu)}^2\big] &= \E\bigg[\norm[\Big]{-\frac{1}{M}\sum_{i=1}^{n}\sum_{j=1}^{N_i}\{Y_{ij} - \mu(T_{ij})\}K_{T_{ij}}}^2\bigg]\\
    &= \E\left\{\frac{1}{M^2}\sum_{i}\sum_{j}\sum_{k}\sum_{l} \innerproduct{D_{ik}(\mu)}{D_{jl}(\mu)}\right\} \\
    &= \E\left[\frac{1}{M^2}\sum_{i,j,k,l}\E\{\innerproduct{D_{ik}(\mu)}{D_{jl}(\mu)}\mid N_i, N_j\}\right]\\
    &\leq \underset{t\in \T}{\sup}\norm{K_t}^2 \E\left(\frac{1}{M^2}\sum_{i,j,k,l}\E[\{Y_{ik} - \mu(T_{ik})\}\{Y_{jl} - \mu(T_{jl})\}]\right)\\
    &= \underset{t\in \T}{\sup}\norm{K_t}^2 \E\left(\frac{1}{M^2}\sum_{i,k,l}\E[\{Y_{ik} - \mu(T_{ik})\}\{Y_{il} - \mu(T_{il})\}]\right)\\
    &= \underset{t\in \T}{\sup}\norm{K_t}^{2}\E\left[\frac{1}{M^2}\sum_{i=1}^{n}\Big \{N_i(N_i -1) I_{t,s} + N_i\big(I_{t,t} + \sigma^2_{\eps}\big)\Big\}\right]\\
    &\leq \underset{t\in \T}{\sup}\norm{K_t}^{2}\E\left\{\frac{n N(N-1)}{M^2} I_{t,s} + \frac{nN}{M}\big(I_{t,t} + \sigma^2_{\eps}\big)\right\} \\
    &\leq C_K^2 (nh)^{-1}\left\{\Big(\mu_N^2 + \sigma_N^2\Big)I_{t,s} + \mu_N(I_{t,t} + \sigma_{\eps}^2)\right\}\\
    &= O((nh)^{-1}).
\end{align*}
Observe that the final step takes advantage of the bounded nature of $I_{t,s}$ and $I_{t,t}$. Thus, $$S_{M}(\mu) = O_p((nh)^{-1/2}),$$ and $a_n = o_p(S_{M}(\mu))$. Additionally, it can be readily verified that
\begin{align*}
    \abs{(nh)^{1/2}r_\delta(t)} &= (nh)^{1/2}\abs{\innerproduct{K_t}{r_\delta}} \\
    &\leq (nh)^{1/2}\norm{K_t}\norm{r_\delta}\\
    &= O_p((nh)^{1/2}h^{-1/2}a_n)\\
    & = o_p(1).
\end{align*}

Therefore,
\begin{align*}
    (nh)^{1/2}\{\hat{\mu}_\lambda(t) - \mu_b(t)\} &= \frac{(nh)^{1/2}}{M}\sum_{i=1}^{n}\sum_{j=1}^{N_i}\{Y_{ij} - \mu(T_{ij})\}K_{T_{ij}}(t) + (nh)^{1/2}r_\delta(t) \\
    &= \frac{(nh)^{1/2}}{M}\sum_{i=1}^{n}\sum_{j=1}^{N_i}\{Y_{ij} - \mu(T_{ij})\}K_{T_{ij}}(t) + o_p(1).
\end{align*}
The limiting distribution of the first term on the RHS of the above equation can be established by the Lindeberg's CLT. Specifically, let 
\begin{align*}
    S_{n}(t) &= (nh)^{1/2}\left[\frac{1}{M}\sum_{i=1}^{n}\sum_{j=1}^{N_i}\{Y_{ij} - \mu(T_{ij})\}K_{T_{ij}}(t)\right]\\
    &= \sum_{i=1}^{n}\frac{(nh)^{1/2}}{M}\sum_{j=1}^{N_i}\{Y_{ij} - \mu(T_{ij})\}K_{T_{ij}}(t) \\
    &\triangleq \sum_{i=1}^{n}\widetilde{\Theta}_i(t).
\end{align*}
It is important to note that $\E\brac{\widetilde{\Theta}_i(t)} = 0$ and the $\widetilde{\Theta}_i(t)$'s are independent. Thus, for any $t\in \T$,
\begin{align*}
    &~~~~\E\brac{\widetilde{\Theta}_i^2(t)}\\
    &= \E\sbrac[\bigg]{\frac{nh}{M^2}\sum_{j=1}^{N_i}\sum_{k=1}^{N_i}\brac{Y_{ij} - \mu(T_{ij})}\brac{Y_{ik} - \mu(T_{ik})}K_{T_{ij}}(t)K_{T_{ik}}(t)}\\
    &=\E\pbrac[\bigg]{\frac{nh}{M^2}\sum_{j,k=1}^{N_i}K_{T_{ij}}(t)K_{T_{ik}}(t)\E\sbrac[\big]{\brac{Y_{ij} - \mu(T_{ij})}\brac{Y_{ik} - \mu(T_{ik})}\mid T_{ij}, T_{ik}}}\\
    &= \E\sbrac[\bigg]{\frac{nh}{M^2}\sum_{j,k=1}^{N_i}K_{T_{ij}}(t)K_{T_{ik}}(t)\brac{c(T_{ij}, T_{ik}) + \delta_{jk}\sigma_{\err}^2}}\\
    &= \E\sbrac[\bigg]{\frac{nh}{M^2}\brac{(N_i^2 - N_i)\G_1(t,t) + N_i\G_2(t,t) + \sigma^2_{\err}N_i\G_3(t,t)}} \\
    &= h\E\brac{\frac{n}{M^2}S_i(t,t)},
\end{align*}
where the final line is based on the assumption stated in Theorem \ref{thm:ptwise_limiting_dist_biased} and Lemma \ref{lem:G_functions}. Thus, $$\Var\brac{{\sum_{i=1}^n {\widetilde{\Theta}_i}(t)}} = h \sum_{i=1}^n \E\brac{n\Sc_i(t,t)/M^2} = \sigma_t^2 + o(1).$$

The application of Lindeberg's conditions to the triangular array $\{\widetilde{\Theta}_i(t)\}_{i=1}^n$ results
\begin{align*}
    &\hspace{1.2em}\frac{1}{\sigma_t^2}\sum_{i=1}^{n}\E\sbrac[\Big]{\widetilde{\Theta}_i^2(t)\bbi_{\brac{\abs{\widetilde{\Theta}_i(t)} > e}}}\\
    &\leq\frac{1}{\sigma_t^2}\sum_{i=1}^{n}\sbrac[\Big]{\E\brac{\widetilde{\Theta}_i^4(t)}}^{1/2}\brac{\p\pbrac{\abs{\widetilde{\Theta}_i(t)} > e}}^{1/2}\\
    &\leq \frac{1}{\sigma_t^2}\sum_{i=1}^{n}\sbrac[\bigg]{\E\brac{\innerproduct{\widetilde{\Theta}_i}{K_t}^4}}^{1/2}\brac{\p\pbrac{\abs{\widetilde{\Theta}_i(t)} > e}}^{1/2}. \numberthis \label{eq:Lindeberg_eq_ptwise_variance}
\end{align*}
Note that, according to Lemma \ref{lem:tilde_Theta_upper_bound} and Lemma \ref{prop:K_and_Wf_bound}, we can establish that
\begin{equation*}
    \E\brac{\innerproduct{\widetilde{\Theta}_i}{K_t}^4} \leq \E\left(\norm{\widetilde{\Theta}}^4\norm{K_t}^4\right) = c_0^2C_K^4(nh)^{-2}.
\end{equation*}
Regarding the probability, it is indicated in Lemma \ref{lem:exp_tail_prob_Theta} that
\begin{align*}
    \p\brac{\sup_{t\in\T} \abs{\widetilde{\Theta}_i}> e} &\leq c_1\log^2(n)\exp\brac{-\frac{c_2(nh)^{1/2}}{\log(n)}} + 2\p\brac{N_i > c_3\log(n)}\\
    &\leq c_1 \log^2(n)\exp\brac{-\frac{c_1(nh)^{1/2}}{\log(n)}} + 2\E\brac{\exp(C_NN_i)}n^{-c_3C_N}.
\end{align*}

Putting all the aforementioned results together and choosing $c_3 > C_N^{-1} + 1$, we can establish that
\begin{align*} 
\eqref{eq:Lindeberg_eq_ptwise_variance} &\leq n \left(c_0C_K^2\frac{1}{nh}\right)\sbrac[\bigg]{c_1
    \log^2(n)\exp\brac{-\frac{c_2(nh)^{1/2}}{\log(n)}} + \frac{2\E\brac{\exp(C_NN_i)}}{n^{c_3C_N}}}^{1/2}\\
    &\leq c_0C_K^2\sbrac[\bigg]{c_1\frac{\log^2(n)}{h^2}\exp\brac{-\frac{c_2(nh)^{1/2}}{\log(n)}} + \frac{2\E\brac{\exp(C_NN_i)}}{n^{c_3C_N}h^2}}^{1/2}\\
    &= o(1),
\end{align*}
given that $c_3C_N>1$ and $\log^2(n)\exp\brac{-(nh)^{1/2}/\log(n)} = o(h^2)$. Thus, for any $t \in \T$,
\begin{equation*}
    (nh)^{1/2}\brac{\hat{\mu}(t) - \mu(t)} \overset{d}{\to} N(0, \sigma^2_t),
\end{equation*}
by the Lindeberg's CLT and the symmetry of the standard normal distribution.
\end{proof}
\section{Proof of Theorem \ref{thm:weak_conv}}\label{appendix:proof_thm3.4}
\begin{proof}{}
The proof consists of three parts. 
\begin{enumerate}
    \item[1.]The first part demonstrates the kernel of the limiting process.
    \item[2.]The second part establishes the asymptotic tightness of the $\Sb_n(t)$.
    \item[3.]The third part verifies the weak convergence of the finite-dimensional marginals of $\Sb_n$.
\end{enumerate}  
\begin{proofpart}{1}{}
Note that,
\begin{align*}
\Sb_n(t) &= (nh)^{1/2}\brac{\hat{\mu}(t) - \mu(t)\ + S_{M, \lambda}(\mu)(t) -W_\lambda(\mu)(t) - S_M(\mu)(t)} \\
&= (nh)^{1/2}\brac{\hat{\mu}(t) - \mu(t) + S_{M, \lambda}(\mu)(t)} - (nh)^{1/2}W_\lambda(\mu)(t)\\ 
&- (nh)^{1/2}S_M(\mu)(t)\\
&\triangleq I_{1, n}(t) + I_{2, n}(t) + I_{3, n}(t). \numberthis \label{eq:stochastic_proc_decomp}
\end{align*}

According to Theorem \ref{thm:FBR}, 
\begin{align*}
    \sup_{t\in\T}\abs{I_{1, n}(t)} &= (nh)^{1/2}\innerproduct{K_t}{\hat{\mu} - \mu + S_{M,\lambda}(\mu)} \\
    & \leq (nh)^{1/2}\norm{K_t}\norm{\hat{\mu} - \mu - S_{M, \lambda}(\mu)} \\
    & = o_p(1).
\end{align*}

As outlined in Remark \ref{remark:unbiased_remark}, if the sufficient condition specified therein is met,
$$\sup \abs{I_{2, n}(t)} = O((nh)^{1/2}\kappa_n) = o(1).$$

As for the $I_{3,n}(t)$, recall that
$$S_M(\mu)(t) = -\frac{1}{M}\sum_{i=1}^{n}\sum_{j=1}^{N_i}\brac{Y_{ij} - \mu(T_{ij})}K_{T_{ij}}(t).$$ 
First, let's rewrite the leading term $I_{3, n}$ in \eqref{eq:stochastic_proc_decomp},
\begin{equation*}
    I_{3, n}(t) = \sum_{i=1}^{n}\frac{(nh)^{1/2}}{M}\sum_{j=1}^{N_i}\{Y_{ij} - \mu(T_{ij})\}K_{T_{ij}}(t) \triangleq \sum_{i=1}^{n} \widetilde{\Theta}_i(t),
\end{equation*}
where $\widetilde{\Theta}_i(t) = (nh)^{1/2}\Theta_i(t)/M$ and $\Theta_i(t) = \sum_{j=1}^{N_i}\{Y_{ij} - \mu(T_{ij})\}K_{T_{ij}}(t)$. 

Furthermore, we have
\begin{align*}
    \E\brac{\Theta_i \mid N_i} &= \sum_{j=1}^{N_i}\E\left[\{Y_{ij} - \mu(T_{ij})\}K_{T_{ij}}(t)\right]\\
    &= \sum_{j=1}^{N_i}\E\sbrac[\Big]{K_{T_{ij}}(t)\E\brac{Y_{ij} - \mu(T_{ij}) \mid T_{ij}}}\\
    &= 0.
\end{align*}
A similar argument can be used to compute the expectation of $I_{3,n}(t)$, yielding
\begin{align*}
\E\{I_{3, n}(t)\} &= (nh)^{1/2}\E\sbrac[\bigg]{\frac{1}{M}\sum_{i=1}^{n}\sum_{j=1}^{N_i}\brac{Y_{ij} - \mu(T_{ij})}K_{T_{ij}}(t)}\\
& = (nh)^{1/2}\E\pbrac[\bigg]{\frac{1}{M}\sum_{i=1}^{n}\sum_{j=1}^{N_i}\sbrac[\Big]{K_{T_{ij}}(t)\E\brac{{Y_{ij} - \mu(T_{ij})}\mid T_{ij}}}} \\
& = 0,
\end{align*}
and
\begin{align*}
    &~~~~\E\sbrac[\bigg]{\frac{1}{M^2}\sum_{i=1}^{n}\sum_{\substack{k=1\\k\neq i}}^{n}\sum_{j =1}^{N_i}\sum_{l=1}^{N_k}\brac{Y_{ij} - \mu(T_{ij})}\brac{Y_{kl} - \mu(T_{kl})}K_{T_{ij}}(t)K_{T_{kl}}(s)} \\
    &=\E\pbrac[\bigg]{\frac{1}{M^2}\sum_{\substack{i,k=1 \\ i\neq k}}^{n}\sum_{j =1}^{N_i}\sum_{l=1}^{N_k}K_{T_{ij}}(t)K_{T_{kl}}(s)\E\sbrac[\big]{\brac{Y_{ij} - \mu(T_{ij})}\brac{Y_{kl} - \mu(T_{kl})}\mid T_{ij}, T_{kl}}} \\
    &=0.
\end{align*}
Therefore, the covariance between $I_{3,n}(t)$ and $I_{3,n}(s)$ is as follows:
\begin{align*}
    &\hspace{1.25em}\E\{I_{3,n}(t)I_{3,n}(s)\} \\
    &= \E\pbrac[\bigg]{\frac{nh}{M^2}\sbrac[\Big]{\sum_{i=1}^{n}\sum_{j=1}^{N_i}\brac{Y_{ij} - \mu(T_{ij})}K_{T_{ij}}(t)}\sbrac[\Big]{\sum_{k=1}^{n}\sum_{l=1}^{N_k}\brac{Y_{kl} - \mu(T_{kl})}K_{T_{kl}}(s)}} \\
    &= \E\sbrac[\bigg]{\frac{nh}{M^2}\sum_{i=1}^{n}\sum_{j,l =1}^{N_i}\brac{Y_{ij} - \mu(T_{ij})}\brac{Y_{kl} - \mu(T_{kl})}K_{T_{ij}}(t)K_{T_{il}}(s)}\\
    &= \E\pbrac[\bigg]{\frac{nh}{M^2}\sum_{i=1}^{n}\sum_{j,l =1}^{N_i}K_{T_{ij}}(t)K_{T_{il}}(s)\E\sbrac[\Big]{\brac{Y_{ij} - \mu(T_{ij})}\brac{Y_{kl} - \mu(T_{kl})}\mid T_{ij}, T_{il}}} \\
    &= \E\sbrac[\bigg]{\frac{nh}{M^2}\sum_{i=1}^{n}\sum_{j,l =1}^{N_i}K_{T_{ij}}(t)K_{T_{il}}(s)\brac{c(T_{ij}, T_{il}) + \delta_{jl}\sigma_\err^2}} \\
    &= \E\pbrac[\bigg]{\frac{nh}{M^2}\sum_{i=1}^{n}\sum_{j,l =1}^{N_i}\E\sbrac[\Big]{K_{T_{ij}}(t)K_{T_{il}}(s)\brac{c(T_{ij}, T_{il}) + \delta_{jl}\sigma_\err^2}\mid N_i}} \\
    &= \E\pbrac[\bigg]{\frac{nh}{M^2}\sum_{i=1}^{n}\sbrac[\Big]{(N_i^2 - N_i)\G_1(s, t) + N_i\{\G_2(s, t) + \sigma_\err^2 \G_3(s, t)\}}} \\
    & \triangleq C_Z(t, s) + o(1),
\end{align*}
where $\G_i(t, s), i=1,2,3,$ are defined in \eqref{eq:I_and_G_intgrals}. The final two lines are derived using the assumption stated in Theorem \ref{thm:weak_conv} and Lemma \ref{lem:G_functions}.
\end{proofpart}
\begin{proofpart}{2}{}
The objective of this part is to establish the equicontinuity of $(nh)^{1/2}\{\hat{\mu}(t) - \mu(t)\}$. Given that the supremum of both $I_{1, n}(t)$ and $I_{2,n}(t)$ are $o(1)$, and the equicontinuity can be readily demonstrated for them, we focus our attention on the leading term, $I_{3, n}$. In other words, we aim to show, for any $e > 0$,
\begin{equation}\label{eq:equicontinuity_I3}
    \lim_{\delta \to 0} \lim_{n \to \infty} \p\brac{\sup_{d(s, t) \leq \delta} \abs{I_{3,n}(t) - I_{3,n}(s)} > e} = 0.
\end{equation}
To achieve this, we will apply Lemma A.1 from \cite{Kley2016Bernoulli}. To do so, we introduce a metric function $d: \T \times \T \to \R$ defined as $d(s,t) = \abs{t-s}$ for any $s, t \in \T$, and a Orlicz function $\Psi(x) = x^2$. The Orlicz norm for any real-valued random variable $U$ is defined as follows,
\begin{equation*}
    \norm{U}_\Psi = \inf\{c > 0 \mid \E\{\Psi(\abs{U}/c)\} \leq 1\}\}.
\end{equation*}

Notice that the square of the Orlicz norm of $U$ corresponds to the second moment of $U$. Therefor, we can express 
$$\norm{I_{3, n}(t) - I_{3, n}(s)}_\Psi^2 = \E\sbrac{\brac{I_{3, n}(t) - I_{3, n}(s)}^2}.$$ 
Given that $\E\{\Theta_i(t)\} = 0$ for all $i$ and $t\in \T$, and from Part 1,
$$\E\brac{\Theta_i(t)\Theta_j(s)} = \E\brac{\Theta_i(t)}\E\brac{\Theta_j(s)} = 0\text{~for all~}i \neq j\text{~and any~}(t,s) \in \T \times \T.$$ 

Furthermore,
\begin{align*}    &~~~~\E\brac{\abs{I_{3, n}(t) - I_{3, n}(s)}^2} \\
    &= \E\sbrac[\bigg]{\frac{nh}{M^2}\abs[\Big]{\sum_{i=1}^{n}\Theta_i(t) - \Theta_i(s)}^2} \\
    &= \E\sbrac[\bigg]{\frac{nh}{M^2}\sum_{i=1}^{n}\sum_{j=1}^{n}\{\Theta_i(t) - \Theta_i(s)\}\{\Theta_j(t) - \Theta_j(s)\}} \\
    &= \E\sbrac[\bigg]{\frac{nh}{M^2}\sum_{i=1}^{n}\{\Theta_i^2(t) - 2\Theta_i(t)\Theta_i(s) + \Theta_i^2(s)\}} \\
    &= \E\left[\frac{nh}{M^2}\sum_{i=1}^{n}\Big\{\Theta_i^2(t) - 2\Theta_i(t)\Theta_i(s) + \Theta_i^2(s)\Big\}\right]. \numberthis \label{eq:I3_different_2nd_moment}
\end{align*}
Let us examine each term separately. To begin, $\E\{\Theta_i^2(t) \mid N_i\}$ can be further simplified as follows: 
\begin{align*}
    &~~~~\E\{\Theta_i^2(t) \mid N_i\} \\
    &= \E\pbrac[\bigg]{\sbrac[\Big]{\sum_{j=1}^{N_i}\{Y_{ij} - \mu(T_{ij})\}K_{T_{ij}}(t)}\sbrac[\Big]{\sum_{k=1}^{N_i}\{Y_{ik} - \mu(T_{ik})\}K_{T_{ik}}(t)}~\Big | ~ N_i}\\
    &= \E\sbrac[\Big]{\sum_{j,k=1}^{N_i} K_{T_{ij}}(t)K_{T_{ik}}(t)\{c(T_{ij}, T_{ik}) + \sigma_{\err}^2 \delta_{jk}\} \mid N_i}\\
    &= (N_i^2 - N_i)\G_1(t, t) + N_i \G_2(t, t) + N_i\sigma_{\err}^2\G_3(t, t). \numberthis \label{eq:I3_different_2nd_moment_1st}
\end{align*}

Secondly,
\begin{align*}    &~~~~\E\{\Theta_i(t)\Theta_i(s) \mid N_i\} \\
    &= \E\pbrac[\bigg]{\sbrac[\Big]{\sum_{j=1}^{N_i}\{Y_{ij} - \mu(T_{ij})\}K_{T_{ij}}(t)}\sbrac[\Big]{\sum_{k=1}^{N_i}\{Y_{ik} - \mu(T_{ik})\}K_{T_{ik}}(s)} ~\Big |~ N_i}\\
    &= \E\sbrac[\bigg]{\sum_{j,k=1}^{N_i} K_{T_{ij}}(t)K_{T_{ik}}(s)\{\Cov(T_{ij}, T_{ik}) + \sigma_{\err}^2 \delta_{jk}\} ~\Big | ~ N_i}\\
    &= (N_i^2 - N_i)\G_1(t, s) + N_i \G_2(t, s) + N_i\sigma_{\err}^2\G_3(t, s). \numberthis \label{eq:I3_different_2nd_moment_2nd}
\end{align*}
Therefore, by combining \eqref{eq:I3_different_2nd_moment} through \eqref{eq:I3_different_2nd_moment_2nd}, Lemma \ref{lem:G_functions}, and the condition $M \geq n$, we derive the following result:
\begin{align*}
    &~~~~\E\brac{{\abs{I_{3, n}(t) - I_{3, n}(s)}^2}} \\
    &\leq h\bigg[(\mu_{N^2} - \mu_N)\Delta_{\G_1}(t,s) + \mu_N\Delta_{\G_2}(t,s) + \mu_N\sigma_{\err}^2\Delta_{\G_3}(t,s)\bigg]\\
    &\leq \brac{C_{\G_1}(\mu_{N^2} - \mu_N) + C_{\G_2}\mu_N}h^{\tau}\abs{t-s} + \mu_N\sigma_{\err}^2 C_{\G_3}\abs{t-s}^2\\
    &\triangleq c_1h^{\tau}\abs{t-s} + c_2\abs{t-s}^2.
\end{align*}
In the equation above, $\mu_{N^2}$ represents the second moment of $N$. 

It's worth noting that $d(s,t) = \abs{t-s}$ is a well-defined metric. Additionally, we define $\bar{\eta} = h/2$. Thus, for sufficient large $n$, when $d(s,t) \geq \bar{\eta}$, we can establish the following inequality:
\begin{align*}
    &\hspace{0.8em}\left(c_1h^{\tau}\abs{t-s} + c_2\abs{t-s}^2\right)\abs{t-s}^{-2}\\
    & = c_1h^\tau\abs{t-s}^{-1} + c_2 \\
    &\leq 2c_1h^{\tau - 1} + c_2 \\
    & \leq  2c_1 + c_2.
\end{align*}
Therefore, there exists a universal constant $C_{\Psi}^2 = 2c_1 + c_2$ such that
\begin{equation*}
    \norm{I_{3, n}(t) - I_{3, n}(s)}_\Psi \leq C_\Psi \abs{t-s} = C_\Psi d(s,t),
\end{equation*}
where $C_{\Psi}$ depends on $d(s,t)$ and $\Psi(x)$ when $d(s,t) \geq \bar{\eta}$ for sufficient large $n$.

To handle the case where $d(s,t) \leq \bar\eta/2$, we rely on the Bernstein inequality. Therefore, we need an almost sure upper bound for $\abs{\widetilde{\Theta}_i(t) - \widetilde{\Theta}_i(s)}$. 

According to Lemma \ref{lem:tilde_Theta_upper_bound}, we can establish the following:
\begin{equation*}
    \sup_{t,s \in \T}\abs{\widetilde{\Theta}_i(t) - \widetilde{\Theta}_i(s)} \leq 2\sup_{t\in \T}\abs{\widetilde{\Theta}_i(t)} \leq 2C_{\widetilde{\Theta}}\frac{\log^2 n}{n^{1/2}h}.
\end{equation*}
Furthermore, when $d(s,t) \leq \bar{\eta}$,
\begin{align*}
    &~~~~\E\{\abs{\widetilde{\Theta}_i(t) - \widetilde{\Theta}_i(s)}^2\} \\
    &= \E\brac{\frac{nh}{M}\abs{\Theta_i(t) - \Theta_i(s)}^2}\\
    &= \E\sbrac[\bigg]{\frac{nh}{M^2}\brac{\Theta_i(t)^2 + \Theta_i(s)^2 -2\Theta_i(t)\Theta_i(s)}}\\
    &\leq \frac{h}{n}\E\sbrac[\Big]{(\mu_{N^2}-\mu_N)\Delta_{G_1}(t,s) + \mu_N\Delta_{\G_2}(t,s) + \sigma_{\err}^2\mu_N\Delta_{\G_3}(t,s)} \\
    &\leq \frac{c_1h^{\tau}\abs{t-s} + c_2\abs{t-s}^2}{n} \\
    & \leq \frac{c_3h^2}{n}.
\end{align*}

Let $\mathcal{D}(\epsilon, d)$ represent the associated $\epsilon$-packing number in the metric space $(\T, d)$ where $d$ is the previously defined metric. Notice that $\D(\xi, d) \lesssim \xi^{-1}$, and according to Lemma A.1 in \cite{Kley2016Bernoulli}, there exists a set $\tilde{\T}$ containing at most $\D(\bar{\eta}, d)$ points such that, for any $\eta > \bar{\eta}$ and $\delta, e>0$, as $n \to \infty$,
\begin{align*}
    \eqref{eq:equicontinuity_I3} &\leq c_{4}\brac{\int_{\bar{\eta}/2}^{\eta}\D^{1/2}(\zeta, d)d\zeta + (\delta + 2\bar{\eta})\D(\eta, d)}^2 \\
    &~~~~+ \p\brac{\sup_{\overset{d(s,t) \leq \bar{\eta}}{s,t \in \tilde{\T}}}\abs{I_{3,n}(t) - I_{3,n}(s)} > e/2}\\
    &\leq c_{5}(\eta + \delta^2\eta^{-2}) + \D(\bar{\eta}, d) \sup_{d(s,t) \leq \bar{\eta}, (s,t) \in \T^2}\p\brac{\abs{I_{3,n}(t) - I_{3,n}(s)} > e/4}\\
    &\leq c_{5}(\eta + \delta^2\eta^{-2}) \\
    &\hspace{1.0em}+ c_{6} \bar{\eta}^{-1}\exp\left[-\frac{e^2/32}{\sum_{i=1}^{n}\E\brac{\abs{\widetilde{\Theta}_i(t) - \widetilde{\Theta}_i(s)}^2} + c_7e(\log^2 n)/(n^{1/2}h)}\right]\\
    &\leq c_{5}(\eta + \delta^2\eta^{-2}) + c_{6} \bar{\eta}^{-1}\exp\left[-\frac{e^2/32}{c_{3}\bar{\eta}^2 + c_7e(\log^2 n)/(n^{1/2}h)}\right]\\
    & \leq o(1).
\end{align*}
The final line is obtained by setting $\eta = \delta^{2/3}$ and $\delta \to 0$.

Regarding the remaining $I_{1, n}$ and $I_{2, n}$ in \eqref{eq:stochastic_proc_decomp}, we obtain the following:
\begin{align*}
    &~~~~\lim_{\delta \to 0} \lim_{n \to \infty} \p\brac{\sup_{d(s, t) \leq \delta} \abs{I_{1, n}(t) + I_{2, n}(t) - I_{1, n}(s) - I_{2, n}(s)} > e} \\
    &\leq \lim_{\delta \to 0} \lim_{n \to \infty} \p\brac{\sup_{d(s, t) \leq \delta} \abs{I_{1, n}(t) - I_{1, n}(s)}>e/2 + \sup_{d(s, t) \leq \delta}\abs{I_{2, n}(t) - I_{2, n}(s)} > e/2}\\
    & = 0,
\end{align*}
where the final step is due to $\sup_{t}\abs{I_{1, n}(t)} = o_p(1)$ and $\sup_{t}\abs{I_{2, n}(t)} = o(1)$ in Part 1.

This establishes the equicontinuity of the random process $$(nh)^{1/2}\{\hat{\mu}(t) - \mu(t)\}.$$ The asymptotic tightness is supported by Theorem 1.5.4 and Theorem 1.5.7 in \cite{van1996weak}.
\end{proofpart}

\begin{proofpart}{3}{}
To establish the weak convergence of the finite-dimensional marginal distribution of $\Sb_n(t)$, we will employ the Cram\'{e}r-Wold device. This involves demonstrating that for any $q \in \N$, $(b_1, b_2, \ldots, b_q)$, and any $t_1, t_2, \ldots, t_q \in\T$, the following holds:
\begin{equation}\label{eq:finite_marignal_conv}
    \sum_{k=1}^{q}b_k\Sb_n(t_k) \overset{d}{\to} \sum_{k=1}^{q}b_kZ(t_k)
\end{equation}
Adopting the same notation as in Part 1, for any $t \in \T$, we find that:
\begin{align*}
    \sum_{k=1}^{q}b_k\Sb_n(t) &= \sum_{k=1}^{q}b_k\brac{I_{1,n}(t) + I_{2,n}(t) + I_{3,n}(t)}\\
    &=\sum_{k=1}^{q}b_k\brac{I_{1,n}(t) + I_{2,n}(t)} + \sum_{k=1}^{q}b_kI_{3,n}(t).
\end{align*}
Once more, we focus on the leading term $I_{3,n}$, as the first two terms have been shown to be negligible in Part 1. Note that:
\begin{align*}
    \sum_{k=1}^{q}b_kI_{3,n}(t_k) &= \sum_{k=1}^{q}b_k\sum_{i=1}^{n}\widetilde{\Theta}_{i}(t_k)\\
    &=\sum_{i=1}^{n}\sum_{k=1}^{q}b_k\widetilde{\Theta}_{i}(t_k) \\
    &\triangleq \sum_{i=1}^{n}\U_{i, q},
\end{align*}
where $\U_{i,q} = \sum_{k=1}^{q}b_k\widetilde{\Theta}_i(t_k)$. 

Clearly, $\E(U_{i, q}) = 0$ as $\E\{\widetilde{\Theta}_i(t)\} = 0$ for all $i$ and $t\in \T$. Thus,
\begin{align*}
    \E\{\U_{i,q}^2\} &= \E\sbrac[\Bigg]{\brac{\sum_{k=1}^{q}b_k\widetilde{\Theta}_i(t_k)}\brac{\sum_{l=1}^{q}b_l\widetilde{\Theta}_i(t_l)}}\\
    &=\sum_{k=1}^q\sum_{l=1}^qb_kb_l\E\brac{\widetilde{\Theta}_i(t_k)\widetilde{\Theta}_i(t_l)}\\
    &= \sum_{k=1}^q\sum_{l=1}^qb_kb_lC_Z(t_k, t_l) + o(1),
\end{align*}
where the last line follows the same calculation as in Part 1. 

When $$\sum_{k=1}^q\sum_{l=1}^qb_kb_lC_Z(t_k, t_l) = 0,$$ it implies that $\sum_{k=1}^{q}b_k Z(t_k)$ has a degenerate distribution with a point mass at $0$ and \eqref{eq:finite_marignal_conv} immediately follows. However, if 
$$\sum_{k=1}^q\sum_{l=1}^qb_kb_lC_Z(t_k, t_l) \neq 0,$$ we shall examine the limiting distribution of $\sum_{i=1}^{n} \U_{i, q}$. 

Firstly, $\U_{i, q}$ are independent with respect to the index $i$ for all $q \in \N$. The variance of $\U_{i,q}$ is finite. Thus, the proof is concluded if we can verify the Lindeberg's condition. 

Notice that $C_{b, q} = \sum_{k=1}^{q}\abs{b_k} > 0$. For any $e > 0$,
\begin{align*}
    &\hspace{1em}\sum_{i=1}^{n}\E\sbrac[\Big]{\U_{i,q}^2\bbi_{\{\abs{\U_{i, q}} > e\}}}\\
    &\leq \sum_{i=1}^{n}\brac{\E\pbrac{\U_{i,q}^4}}^{1/2}\brac{\p(\abs{\U_{i, q}} > e)}^{1/2}, \numberthis \label{eq:finite_marginal_Lindeberg}
\end{align*}
where the final line is derived using the Cauchy-Schwarz inequality. It's important to note that:
\begin{align*}
    \p\pbrac{\abs{\U_{i,q}} > e} &= \p\brac{\abs[\Big]{\sum_{k=1}^qb_k\widetilde{\Theta}_i(t_k)} > e}\\
    &\leq \p\brac{\sup_{t\in \T}\abs{\widetilde{\Theta}_i(t)}\abs{\sum_{k=1}^qb_k} > e}\\
    &\leq \p\brac{\sup_{t\in \T}\abs{\widetilde{\Theta}_i(t)}> \frac{e}{C_{b, q}}}\\
    &\overset{a.s.}{\leq} c_0\log^2(n)\exp\brac{-c_1\frac{e(nh)^{1/2}}{\log(n)}} + 2\p\brac{N_i > c_2\log(n)}, \numberthis\label{eq:finite_marginal_prob}
\end{align*}
where the upper bound is obtained from Lemma \ref{lem:exp_tail_prob_Theta}. As for the other term, let's first rewrite $\U_{i,q}$:
\begin{align*}
    \U_{i, q} &= \sum_{k=1}^{q}b_k\widetilde{\Theta}_i(t_k)\\
    &=\sum_{k=1}^{q}b_k \innerproduct{\widetilde{\Theta}_i}{K_{t_k}}\\
    &\leq \norm{\widetilde{\Theta}_i}\norm{K_{t_k}}\sum_{k=1}^{q}\abs{b_k}\\
    &\leq C_Kh^{-1/2}C_{b,q}\norm{\widetilde{\Theta}_i}.
\end{align*}
Therefore, according to Lemma \ref{lem:tilde_Theta_upper_bound},
\begin{align*}
    E\pbrac{\U_{i,q}^4} = E\pbrac{C_K^4h^{-2}C_{b,q}^4\norm{\widetilde{\Theta}_i}^4} \leq c_3^2C_K^4C_{b,q}^4(nh)^{-2}. \numberthis\label{eq:finite_marginal_exp}
\end{align*}

Combining \eqref{eq:finite_marginal_Lindeberg} through \eqref{eq:finite_marginal_exp}, and employing the same calculations used for \eqref{eq:Lindeberg_eq_ptwise_variance} in the proof of Theorem \ref{thm:ptwise_limiting_dist_biased}, we obtain:
\begin{align*}
    &\lim_{n\to\infty}\sum_{i=1}^{n}\E\sbrac[\Big]{\mathcal{U}_{i,q}^2\bbi_{\{\abs{\U_{i, q}} > e\}}} = o(1).
\end{align*}

Finally, by applying the Lindeberg's CLT, we will have \eqref{eq:finite_marignal_conv} established.
\end{proofpart}
\end{proof}

\section{Proofs of Theorems in Section \ref{subsec:bootstrap_validity}}
This subsection illustrates the validity of the bootstrap procedure discussed. Even though two separated theorems are stated in Section \ref{subsec:bootstrap_validity}, we will only provide the proofs for the Theorem \ref{thm:weak_conv_bootstrap_validity}. As been pointed out in Figure 1 from \cite{bucher_note_2019}, Theorem \ref{thm:pt_wise_bootstrap_validity} can be understood as a special case of Theorem \ref{thm:weak_conv_bootstrap_validity}, and the proofs for these two theorems are very similar if not identical.
\begin{proof}
To begin, we introduce the following notation. Let 
$\Sb_{n,b}$ denote the $b$th bootstrap process obtained from Algorithm \ref{algo:bootstrap_algo} for $b =1 ,2, \ldots, B$, where $B$ denotes the bootstrap size. Let $\mathrm{BL}_{1}\brac{\mathrm{C}[0,1]}$ denote the collection of all uniformly Lipschitz functionals. i.e.,
\begin{equation*}
\begin{aligned}
    \mathrm{BL}_{1}\brac{\mathrm{C}[0,1]} := \{\mathfrak{h}: &\mathrm{C}[0,1] \to [-1,1] \mid \\
    &\abs{\mathfrak{h}(g_1) - \mathfrak{h}(g_2)} \leq \norm{g_1 - g_2}_{\sup} \text{ for any } g_1, g_2 \in \mathrm{C}[0,1]\}.
\end{aligned}
\end{equation*}
It is sufficient to show that the 
\begin{equation}\label{eq:boostrap_conditional_expectation}
    \sup_{\mathfrak{h} \in \mathrm{BL}_1\brac{\mathrm{C}[0,1]}} \abs{\E_{\D}\brac{\mathfrak{h}(\Sb_{n,1})} - \E\brac{\mathfrak{h}(\Sb_{n})}},
\end{equation}
where $\E_{\D}$ denotes the conditional expectation given data $\D = \{T_{ij}, Y_{ij} \mid i = 1, 2, \ldots, n, j = 1, 2, \ldots, N_i\}$. For more details, please see Chapter 23 (or more specifically, Theorem 23.7) in \cite{vaart_asymptotic_1998}. Lemma 3.1 in \cite{bucher_note_2019} implies equation \eqref{eq:boostrap_conditional_expectation} is equivalent to prove, for any positive integer $B \geq 2$, as $n \to \infty$,
\begin{equation}\label{eq:bootstrap_validity_proof_goal}
    (\Sb_n, \Sb_{n, 1}, \Sb_{n, 2}, \ldots, \Sb_{n, B}) \leadsto (Z, Z_1, Z_2, \ldots, Z_B),
\end{equation}
where $Z_1, Z_2, \ldots, Z_B$ are i.i.d. copies of $Z$ in Theorem \ref{thm:weak_conv}. We have shown $\Sb_n \leadsto Z$ in Theorem \ref{thm:weak_conv}. Because $Z$ and $Z_i$'s are Gaussian processes, we only need to show $\Sb_{n,l} \leadsto Z_l$ for any $b=1, 2, \ldots, B$. The main idea of this proof is similar to the proof of Theorem \ref{thm:weak_conv}. That is, we will show the finite dimensional convergence and asymptotic tightness of $\Sb_{n,b}$.

For $b$th bootstrap sample, we have
\begin{equation*}
        \hat{\mu}_{b} = \underset{f \in \HS}{\amin} \Bigg[\frac{1}{2M}\sum_{i=1}^{n}\sum_{j=1}^{N_{i}}U_{i, b}\big\{Y_{ij} -  f(T_{ij})\big\}^2 + \frac{\lambda}{2} J(f, f)\Bigg],
\end{equation*}
and
\begin{equation*}
S_{M, \lambda}^{b}(f) = -\frac{1}{M}\sum_{i=1}^{n}\sum_{j=1}^{N_i}U_{i,b}\{Y_{ij} - f(T_{ij})\}K_{T_{ij}} + W_\lambda f.
\end{equation*}
Because $\E U_{i, b} = 1$ and these bootstrap weights are independent of the random components in model $\eqref{eq:model}$. These facts imply we can obtain an variant of Lemma \ref{prop:identity_op}, i.e. $DS_\lambda^{b}(f) = \mathrm{id}$ for any $f\in\HS$ and obtain the same results in Theorem \ref{lem:convergence} and \ref{thm:FBR}.

Besides that, direct calculations yield that
\begin{equation*}
\begin{aligned}
S_{n,b}(t) &= (nh)^{1/2}\brac{\hat{\mu}_{b}(t) - \hmu(t)}\\
&= (nh)^{1/2}\brac{\hat{\mu}_{b} - \mu  + S_{M, \lambda}^{b}(\mu)(t)} \\
& - (nh)^{1/2}\brac{\hmu - \mu + S_{M, \lambda}(\mu)(t)} \\ 
&+(nh)^{1/2}\brac{S_{M, \lambda}(\mu)(t) - S_{M, \lambda}^{b}(\mu)(t)} \\
& \triangleq H_{1,b}(t) + H_2(t) + H_{3,b}(t),
\end{aligned}
\end{equation*}
where $$H_{3,b} = \frac{(nh)^{1/2}}{M}\sum_{i=1}^{n}\sum_{j=1}^{N_i}(U_{i,b} - 1)\{Y_{ij} - \mu(T_{ij})\}K_{T_{ij}}.$$
Lemma \ref{thm:FBR} implies $\norm{\hmu - \mu + S_{M, \lambda}(\mu)} = O_p(a_n)$ and we have further assumed $a_n = o_p(n^{-1/2})$. This assumption is one of the crucial conditions for Theorem \ref{thm:ptwise_limiting_dist_biased} and \ref{thm:weak_conv}. Because
\begin{align*}
    \sup_{t\in\T}\abs{H_2(t)} &= \sup_{t \in \T}~(nh)^{1/2}\abs{\innerproduct{K_t}{\hat{\mu} - \mu + S_{M,\lambda}(\mu)}} \\
    & \leq (nh)^{1/2}\norm{K_t}\norm{\hat{\mu} - \mu - S_{M, \lambda}(\mu)} \\
    & = o_p(1).
\end{align*}
The same result holds for $H_{1,b}(t)$ as $\norm{\hat{\mu}_b - \mu + S_{M,\lambda}^{b}(\mu)} = O_p(a_n).$ Therefore, we need to consider the leading term $H_3^b$, and the validity of equation \eqref{eq:bootstrap_validity_proof_goal} relies on the finite dimensional convergence of $\bH = (H_{3,1}, H_{3,2}, \ldots, H_{3,B})^{\top}$
and asymptotic tightness of $H_{3, b}$.

To begin, for arbitrary $R \in \N$, and $\bc_1, \bc_2, \ldots, \bc_R \in \R^{B}$, we want to show that, for sufficient large $n$,
\begin{equation}\label{eq:bootstrap_dominant_term_H}
    \sum_{r=1}^{R}\bc_r^{\top}\bH(t_r) \to \sum_{r=1}^R \bc_r^{\top} \bZ(t_r),
\end{equation}
where $\bZ = (Z_1, Z_2, \ldots, Z_B)^{\top}$. Note that,
\begin{align*}
    \sum_{r=1}^R\bc_r\bH(t_r) &= \sum_{r=1}^{R}\sum_{b=1}^{B}c_{rb}H_{3,b}(t_r)\\
    &=\sum_{i=1}^{n}\sum_{b=1}^{B}\sum_{r=1}^{R}\sum_{j=1}^{N_i}c_{rb}\frac{(nh)^{1/2}}{M}(U_{i,b} - 1)\brac{Y_{ij} - \mu(T_{ij})}K_{T_{ij}}(t_r)\\
    &= \sum_{i=1}^{n}\fU_i(t_r),
\end{align*}
where
\begin{equation*}
\begin{aligned}
    \fU_i &= \sum_{b=1}^{B}\sum_{r=1}^{R}\sum_{j=1}^{N_i}c_{rb}\frac{(nh)^{1/2}}{M}(U_{i,b} - 1)\brac{Y_{ij} - \mu(T_{ij})}K_{T_{ij}}(t_r)\\
    &= \sum_{b=1}^{B}\sum_{r=1}^{R}c_{rb}(U_{i,b} - 1)\widetilde{\Theta}_i(t_r),
\end{aligned}
\end{equation*}
Here, $\widetilde{\Theta}_i(t_r)$ is defined in equation \eqref{eq:Theta_def}, which includes all random elements in model \eqref{eq:model}. Notice that $\E U_{i, b} = 1$ and $\E \brac{(U_{i, b} - 1)^2} = 1$. Moreover, $U_{i, b}$'s are independent of all random elements in model \eqref{eq:model}, we can easily conclude $\E \fU_i = 0$. As for the second moment of $\fU_i$, we have
\begin{align*}
    \E \fU_i^2 &= \sum_{b, b'}\sum_{r, r'}c_{rb}c_{r'b'}\E\brac{(U_{i, b} - 1)(U_{i, b'} - 1)\widetilde{\Theta}_i(t_r)\widetilde{\Theta}_i(t_{r'})}\\
    &= \sum_{b=1}^{B}\sum_{r, r'}c_{rb}c_{r'b}\E\brac{(U_{i, b} - 1)^2}\E\brac{\widetilde{\Theta}_i(t_r)\widetilde{\Theta}_i(t_{r'})}\\
    &= \sum_{b=1}^{B}\sum_{r, r'}c_{rb}c_{r'b}\E\brac{\widetilde{\Theta}_i(t_r)\widetilde{\Theta}_i(t_{r'})}
\end{align*}
Therefore,
\begin{align*}
    \Var\Big(\sum_{i=1}^{n}\fU_i\Big) &= \sum_{i=1}^{n}\E\fU_i^2\\
    &= \sum_{i=1}^{n}\sum_{b=1}^{B}\sum_{r, r'}c_{rb}c_{r'b}\E\brac{\widetilde{\Theta}_i(t_r)\widetilde{\Theta}_i(t_r')}\\
    &= C_Z(t_r,t_r') + o_p(1),
\end{align*}
where the last line comes from calculating the covariance kernel in Theorem \ref{thm:weak_conv} as
$\sum_{i=1}^{n}\E\fU_i^2 = \E\{I_{3,n}(t)I_{3,n}(s)\}$. The finite dimensional distributions can be proved by using Lindeberg's condition. Note that, $\abs{U_{i, b} - 1} \leq \sqrt{2}$ almost surely. Therefore, for arbitrary fixed $B \in \N$, the Lindeberg condition naturally holds as
$\abs{\fU_i} \leq 2B \abs{\sum_{r = 1}^{R}c_{rb}\widetilde{\Theta}_i(t_r)}$.

As for the asymptotic tightness of $H_{3,b}(t)$, note that, the asymptotic tightness proved for $I_{3,n}(t)$ in Theorem \ref{thm:weak_conv} relies on equation \eqref{eq:equicontinuity_I3}. Because $\abs{U_{i, b} - 1} \leq \sqrt{2}$, we have $H_{3,b}(t) \leq \sqrt{2}I_{3,n}(t)$, i.e. the ratio between $H_{3,b}(t)$ and $\sqrt{2}I_{3,n}(t)$ is bounded above by a constant almost surely. This fact implies the asymptotic tightness of $H_{3,b}$ for all $b=1, 2, \ldots, B$.

Note that, $$\frac{(nh)^{1/2}}{M}\sum_{i=1}^{n}\sum_{j=1}^{N_i}\{Y_{ij} - \mu(T_{ij})\}K_{T_{ij}} = I_{3, n}.$$ The asymptotic tightness is proved in Theorem \ref{thm:weak_conv}. By Theorem 1.5.4 in \cite{van1996weak}, Lemma 3.1 in \cite{bucher_note_2019}, and Theorem \ref{thm:weak_conv}, equation \eqref{eq:bootstrap_dominant_term_H} is valid for any positive integer $B \geq 2$, as $n \to \infty$. This conclude the proof.

\end{proof}

\section{Auxiliary Lemmas}\label{appendix:proof_extra_technial_lemmas}
This subsection contains several technical lemmas that are used throughout all the proofs of the main theorems. An immediate result, following from Lemma \ref{prop:func_expression_by_h} and \ref{prop:K_and_Wf_bound}, establishes a connection between the supremum norm and the norm induced by the inner product defined in \eqref{eq:inner_prod}. It's important to note that 
$$\abs{f(t)}^2 = \innerproduct{K_t}{f}^2 \leq \norm{K_t}^2\norm{f}^2$$ and $\norm{K_t}^2 \leq C_K^2h^{-1}$ as shown in Lemma \ref{prop:K_and_Wf_bound}. Thus, we have the following lemma:
\begin{lemma} \label{prop:sup_norm}
For any $f \in \HS$, there exists a constant $C_K$ such that $$|f(t)| \leq C_K h^{-1/2}\norm{f}$$ for any $t\in \T$, where $C_K$, as defined in Lemma \ref{prop:K_and_Wf_bound}, does not depend on the choice of $f$ or $t$. In other words, $\norm{f}_{\sup} \leq C_K h^{-1/2}\norm{f}$.
\end{lemma}
The next lemma is a direct result of computing the Frech\'{e}t derivatives of $\ell(f, \lambda)$.
\begin{lemma}\label{prop:identity_op}
$DS_{\lambda}(f) = \mathrm{id}$ where $\mathrm{id}$ is the identity operator on $\HS$ for any $f \in \HS$.
\end{lemma}
\begin{proof}
Recall that  $D\ell(f, \lambda)\Delta g = \innerproduct{S_{M, \lambda}(f)}{\Delta g}$ where 
\begin{equation*}
    S_{M, \lambda}(f) = -\frac{1}{M}\sum_{i=1}^{n}\sum_{j=1}^{N_i}\{Y_{ij} - f(T_{ij})\}K_{T_{ij}} + W_\lambda f \in \HS.
\end{equation*}
Let's define $\E\{S_{M, \lambda}(f)\} \triangleq S_{\lambda}(f)$ and $\E\{S_{M}(f)\} \triangleq S(f)$. It can be shown that
\begin{align}
    \E\{S_{M}(f)\} &= \E\left\{-\frac{1}{M}\sum_{i=1}^{n}\sum_{j=1}^{N_i}\left\{Y_{ij} - f(T_{ij})\right\}K_{T_{ij}}\right\}\nonumber\\
    &= \E\left[-\E\left\{\frac{1}{M}\sum_{i=1}^{n}\sum_{j=1}^{N_i}\left\{Y_{ij} - f(T_{ij})\right\}K_{T_{ij}} \Big \vert ~N_i, \text{ for } i = 1, 2, \ldots n\right\}\right]\nonumber\\\
    &= \E\left[-\left\{Y_{11} - f(T_{11})\right\}K_{T_{11}}\right] \nonumber\\
    &= S(f). \label{eq:S(f)}
\end{align}
Notice that $S_\lambda(f) = S(f) + W_\lambda f$. Therefore, we have $DS_\lambda(f) = DS(f) + W_\lambda$. Additionally, the second order Fr\'echet derivative w.r.t $f$ is
\begin{align*}
    D^2 \ell(f, \lambda)\Delta g \Delta h &= \frac{1}{M}\sum_{i=1}^{n}\sum_{j=1}^{N_i}\innerproduct{K_{T_{ij}}}{\Delta g}\innerproduct{K_{T_{ij}}}{\Delta h} + \innerproduct{W_\lambda \Delta g}{\Delta h} \nonumber \\
    &\triangleq DS_{M, \lambda}(f)\Delta g \Delta h \\
    & = \text{const w.r.t. } f.
\end{align*}
Therefore, $$DS_{M,\lambda}(f)g = \frac{1}{M}\sum_{i=1}^{n}\sum_{j=1}^{N_i}g(T_{ij})K_{T_{ij}} + W_\lambda g.$$ 
Similar to the previous calculation, it can be shown that 
$$DS_\lambda(f)g = \E[g(T)K_T] + W_\lambda g.$$ 

Therefore, for any $f, g, h \in \HS$, we can obtain:
\begin{align*}
    \innerproduct{DS_\lambda(f) g}{h} &= \E[\innerproduct{K_{T}}{g}\innerproduct{K_{T}}{h}] + \innerproduct{W_\lambda g}{h} \\
    &= \E[g(T)h(T)] + \lambda J(g, h)\\
    &= \innerproduct{g}{h}.
\end{align*}

This concludes that $DS_\lambda(f) = \mathrm{id}$ for any $f\in\HS$.
\end{proof}
\begin{lemma} \label{lem:G_functions}
Suppose that Assumption \ref{assp:data_gm}, Assumption \ref{assp:smooth_of_X} and Assumption \ref{assp:Fourier_expansion} hold, then the following functions defined in \eqref{eq:I_and_G_intgrals} can be expressed in the form of the basis function $\{h_v\}_{v \in \N_+}$ of $\HS$,
\begin{align*}
    \G_1(s, t) &= \sum_{i, j}\frac{h_i(t)h_j(s)}{(1+\lambda \gamma_i)(1 + \lambda \gamma_j)}I_{1}^{ij}, \\
    \G_2(s, t) &= \sum_{i, j}\frac{h_i(t)h_j(s)}{(1+\lambda \gamma_i)(1 + \lambda \gamma_j)}I_{2}^{ij}, \\
    \G_3(s,t) &= \sum_{i}\frac{h_i(t)h_i(s)}{(1+\lambda\gamma_i)^2}.
\end{align*}
In particular, if $m \geq 1$ for $\HS = H^{m}$ and Assumption \ref{assp:HS_norm_of_C} holds, then
\begin{equation*}
    \Delta_{\G_1}(t,s) = \abs{\G_1(t, t) + \G_1(s, s) - 2\G_1(t,s)} \leq C_{\G_1}h^{\tau - 1}\abs{t-s},
\end{equation*}
\begin{equation*}
    \Delta_{\G_2}(t,s) = \abs{\G_2(t, t) + \G_2(s, s) - 2\G_2(t,s)} \leq C_{\G_2}h^{\tau - 1}\abs{t-s},
\end{equation*}
\begin{equation*}
    \Delta_{\G_3}(t,s) = \abs{\G_3(t, t) + \G_3(s, s) - 2\G_3(t,s)} \leq C_{\G_3}h^{-1}\abs{t-s}^2.
\end{equation*}
\end{lemma}
\begin{proof}
Note that Lemma \ref{prop:func_expression_by_h} implies that $K(t, \cdot) = \sum_{i} h_i(t)h_i(\cdot)/(1+\lambda \gamma_i)$. Let's start by considering $\G_3(s,t)$ first.
\begin{align*}
    \G_3(s,t) &= \int_{u\in \T}\brac{\sum_i \frac{h_i(t)}{1 + \lambda\gamma_i}h_i(u)}\brac{\sum_j \frac{h_j(s)}{1 + \lambda\gamma_j}h_j(u)}\pi(u)du\\
    &=\sum_{i,j}\frac{h_i(t)h_j(s)}{(1+\lambda \gamma_i)(1+\lambda \gamma_j)}\int_{u\in \T}h_i(u)h_j(u)\pi(u)du\\
    &= \sum_{i,j}\frac{h_i(t)h_j(s)}{(1+\lambda \gamma_i)(1+\lambda \gamma_j)}V(h_i, h_j)\\
    &= \sum_{i}\frac{h_i(t)h_i(s)}{(1+\lambda \gamma_i)^2}.
\end{align*}
Hence,
\begin{align*}
    \Delta_{\G_3}(t,s) &= \abs{\G_3(t, t) + \G_3(s, s) - 2\G_3(t,s)}\\
    &= \sum_{i}\frac{\brac{h_i(t)-h_i(s)}^2}{(1+\lambda \gamma_i)^2}\\
    &= \sum_{i}\frac{\{\df{h_i}(x)\}^{2}}{(1+\lambda \gamma_k)^2}(t-s)^2 ~\text{ for some $x$ between $t$ and $s$}\\
    &\leq C_{\G_3}h^{-1}\abs{t-s}^2,
\end{align*}
where the second-to-last line uses the Mean Value Theorem, and the last line relies on the boundedness of the first derivative of $h_i$, for all $i$, as $h_i \in \HS = H^{m}$. Furthermore,
\begin{align*}
    \G_1(s, t) &= \int_{(u, v) \in \T^2} c(u,v)\brac{\sum_i \frac{h_i(t)}{1 + \lambda\gamma_i}h_i(u)}\brac{\sum_j \frac{h_j(s)}{1 + \lambda\gamma_j}h_j(u)}\pi(u)du \\
    &= \sum_{i, j}\frac{h_i(t)h_j(s)}{(1+\lambda \gamma_i)(1 + \lambda \gamma_j)}\int_{(u,v) \in \T^2}c(u,v)h_i(u)h_j(v)\pi(u)\pi(v)du dv\\
    &= \sum_{i, j}\frac{h_i(t)h_j(s)}{(1+\lambda \gamma_i)(1 + \lambda \gamma_j)}I_{1}^{ij}.
\end{align*}
Since $\abs{\G_1(t, t) + \G_1(s, s) - 2\G_1(t,s)} \leq \abs{\G_1(t, t) - \G_4(t,s)} + \abs{\G_1(s, s) - \G_4(s, t)}$, and if we examine each term, we obtain that
\begin{align*}
   \abs{\G_1(t, t) - \G_1(t,s)} &\leq \sum_{i, j}\frac{\abs{h_i(t)\brac{h_j(s) - h_j(t)}I_{1}^{ ij}}}{(1+\lambda \gamma_i)(1 + \lambda \gamma_j)}\\
   &= \sum_{i, j}\frac{\abs{h_i(t)\df{h_j}(x)I_{1}^{ij}}}{(1+\lambda \gamma_i)(1 + \lambda \gamma_j)}\abs{t-s}\\
   &\leq c_1\sum_{i, j}\frac{\int_{(u,v) \in \T^2}c(u,v)h_i(u)h_j(v)du dv}{(1+\lambda \gamma_i)(1 + \lambda \gamma_j)}\abs{t-s}\\
   &\leq c_2\abs{t-s}\brac{\sum_{i}\frac{i^{-\tau}}{(1+\lambda \gamma_i)^2}}^{1/2}\brac{\sum_{j}\frac{j^{-\tau}}{(1 + \lambda \gamma_j)^2}}^{1/2}\\
   &\leq C_{\G_1}\abs{t-s}h^{\tau - 1}.
\end{align*}

Thus, we have $\Delta_{\G_1}(t,s) \leq C_{\G_1}\abs{t-s}h^{\tau - 1}$. The same argument applies to $\Delta_{\G_2}(t,s)$.
\end{proof}

The following technical lemma aims to establish the properties of
\begin{equation}\label{eq:Theta_def}
    \widetilde{\Theta}_i(t) = \frac{(nh)^{1/2}}{M}\Theta_i(t),
\end{equation}
where $\Theta_i(t) = \sum_{j=1}^{N_i}\{Y_{ij} - \mu(T_{ij})\}K_{T_{ij}}(t)$ for $i=1, 2, \ldots, n$, and $\mu$ is the true mean function. These results facilitate the derivation of point-wise normality and weak convergence.
\begin{lemma}\label{lem:tilde_Theta_upper_bound}
Suppose Assumptions \ref{assp:data_gm} - \ref{assp:exponential_tails} hold, then, for $i=1, 2, \ldots, n$,
\begin{align*}
    (i)~\sup_{t\in \T}\abs{\widetilde{\Theta}_i(t)} \overset{a.s}{\leq} c_1\frac{\log^2(n)}{n^{1/2}h}\text{~and~} (ii)~\E\left(\norm{\widetilde{\Theta}_i}^4\right) \leq c_2n^{-2}.
\end{align*}
\end{lemma}
\begin{proof}
Let $c_1 > C_{\err}^{-1}$, then
\begin{align*}
    \sum_{n}\Pr(\abs{\err_n} \geq c_1 \log(n)) &\leq \sum_{n} \frac{\E\{\exp(C_\err \abs{\err})\}}{n^{c_1 C_\err}}\\
    &= E\{\exp(C_\err \abs{\err})\} \sum_{n} n^{-c_1C_\err} \\
    & < \infty,
\end{align*}
where the final line arises from Assumptions \ref{assp:exponential_tails}. Consequently, we have $\abs{\err} \leq (c_1\log n )$ almost surely. For the random function $X$ and the number of observations per subject $N$, similar arguments apply, leading to $\norm{X} \leq (c_2\log n )$ and $\abs{N} \leq (c_3\log n)$ almost surely, where $c_2 > C_{X}^{-1}$ and $c_3 > C_{N}^{-1}$. 

Note that $$\underset{t \in \T}{\sup} \abs{\widetilde{\Theta}_i(t)} = \frac{(nh)^{1/2}}{M}\underset{t \in \T}{\sup}\abs{\Theta(t)}.$$
As a result, we can obatin the following:
\begin{align*}
    \underset{t \in \T}{\sup} \abs{\Theta_i(t)} &= \underset{t \in \T}{\sup} \abs[\bigg]{\sum_{j=1}^{N_i}\{Y_{ij} - \mu(T_{ij})\}K_{T_{ij}}(t)\}}\\
    &= \underset{t \in \T}{\sup} \abs[\Big]{\sum_{j=1}^{N_i}\{Y_{ij} - X_i(T_{ij}) + X_i(T_{ij}) - \mu(T_{ij})\}K_{T_{ij}}(t)\}}\\
    & \leq \underset{t \in \T}{\sup} \sum_{j=1}^{N_i}\abs{\err_{ij}K_{T_{ij}}(t)} + \underset{t \in \T}{\sup} \sum_{j=1}^{N_i} \abs[\Big]{\{X_{i}(T_{ij}) - \mu(T_{ij})\}K_{T_{ij}}(t)}.  \numberthis \label{eq:theta_tilde_upper_bound_as}
\end{align*}
The first term on the RHS of \eqref{eq:theta_tilde_upper_bound_as} can be  computed as follows:
\begin{align*}
     \underset{t \in \T}{\sup} \sum_{j=1}^{N_i}\abs{\err_{ij}K_{T_{ij}}(t)} &\leq (\sum_{j=1}^{N_i}\abs{\err_{ij}})\underset{t \in \T}{\sup}\abs{K_{T}(t)} \\
     &\leq C_Kh^{-1/2}\norm{K_t}\sum_{j=1}^{N_i}\abs{\err_{ij}} \text{ by Lemma \ref{prop:sup_norm}}\\
     &\overset{a.s.}{\leq} C_h\brac{c_1\log(n)}\brac{c_3\log(n)} h^{-1}\\
     &\leq c_4 h^{-1}\log^2(n). \numberthis \label{eq:as_error_bound}
\end{align*}
The second term on the RHS of \eqref{eq:theta_tilde_upper_bound_as} can be analyzed using a similar argument. Notice that
\begin{align*}
    \abs{X_i(T_{ij}) - \mu(T_{ij})} &= \innerproduct{X_i - \mu}{K_{T_{ij}}} \\
    & \leq \norm{X_i - \mu}\norm{K_{T_{ij}}}\\
    & \overset{a.s.}{<} c_2C_Kh^{-1/2}\log(n).
\end{align*}
Additionally,
\begin{align*}
    &\hspace{1.2em}\sup_{t\in \T}\sum_{j=1}^{N_i}\abs{\{X_i(T_{ij}) - \mu(T_{ij})\}K_{T_{ij}}(t)} \\
    &\leq \sum_{j=1}^{N_i}\abs[\big]{\{X_i(T_{ij}) - \mu(T_{ij})\}}\sup_{t\in \T}\abs{K_{T_{ij}}(t)} \\
    &\leq C_Kh^{-1/2}\norm{K_t}\sum_{j=1}^{N_i}\abs[\big]{\{X_i(T_{ij}) - \mu(T_{ij})\}} \text{ by Lemma \ref{prop:sup_norm}}\\
    &\leq C_K^2h^{-1}\pbrac[\Big]{\sum_{j=1}^{N_i}\innerproduct{X_i - \mu}{K_{T_{ij}}}}\\
    &\leq C_K^2h^{-1}N_i\norm{X_i-\mu}\norm{K_t}\\
    &\overset{a.s.}{\leq} c_2C_K^3h^{-3/2}\log(n)N_i\\
    &\overset{a.s.}{\leq} c_5 h^{-3/2}\log^2(n). \numberthis \label{eq:as_mean_differ_bound}
\end{align*}
Because \eqref{eq:as_mean_differ_bound} dominates \eqref{eq:as_error_bound} when $h = o(1)$, for sufficient large $n$, we can obtain
\begin{align*}
    \eqref{eq:theta_tilde_upper_bound_as} \overset{a.s.}{\leq} c_6\frac{\log^2 n}{h^{3/2}}.
\end{align*}
Immediately, it can be shown that
\begin{align*}
    \sup_{t\in\T} \abs{\widetilde{\Theta}(t)} \overset{a.s.}{\leq} c_6\frac{(nh)^{1/2}}{n}\frac{\log^2(n)}{h^{3/2}} = c_6\frac{\log^2(n)}{n^{1/2}h}.
\end{align*}

As for the second part of the lemma, consider that:
\begin{align*}
    &~~~~\innerproduct{\widetilde{\Theta}_i}{\widetilde{\Theta}_i}\\
    &= \frac{nh}{M^2}\sum_{j=1}^{N_i}\sum_{k=1}^{N_i}\brac{Y_{ij} - \mu(T_{ij})}\brac{Y_{ik} - \mu(T_{ik})}\innerproduct{K_{T_{ij}}}{K_{T_{ik}}}\\
    &\leq \frac{1}{n}\sum_{j,k=1}^{N_i}\brac{Y_{ij} - \mu(T_{ij})}\brac{Y_{ik} - \mu(T_{ik})}\\
    &= \frac{1}{n}\sum_{j,k=1}^{N_i}\brac{\err_{ij} + \delta_{ij}}\brac{\err_{ik} + \delta_{ik}},
\end{align*}
where $\delta_{ij} = X_i(T_{ij}) - \mu(T_{ij})$. We can see that:
\begin{align*}
(i)~\E(\delta_{ij}) = 0, (ii)~\E(\delta_{ij}^2) = I_{t,t}, (iii)~\E(\delta_{ij}\delta_{ik}) = I_{t,s}, j\neq k,
\end{align*}
where $I_{t,s}, I_{t,t}$ are defined as in \eqref{eq:I_and_G_intgrals} for all $j, k=1,2, \ldots, N_i$. Thus,
\begin{align*}
    \E\left(\norm{\widetilde{\Theta}_i}^4\right) &= \E\left(\innerproduct{\widetilde{\Theta}_i}{\widetilde{\Theta}_i}^2\right)\\
    &\leq \frac{1}{n^2}\E\sbrac[\bigg]{\sum_{j,k,u,v=1}^{N_i}\pbrac{\err_{ij} - \delta_{ij}}\pbrac{\err_{ik} - \delta_{ik}}\brac{\err_{iu} - \delta_{iu}}\brac{\err_{iv} - \delta_{iv}}}\\
    &= \frac{1}{n^2}\E\brac{N_i \err_{i1}^4 + (N_i^2 - N_i)\err_{i1}^2\delta_{i1}\delta_{i2} + N_i\err_{i1}^2\delta_{i1}^2 + N_i \delta_{i1}^4}. \numberthis\label{eq:4th_moment_tilde_Theta}
\end{align*}
Note that there always exists an $\epsilon$ such that $\epsilon^4 \leq \exp(C_\eps\epsilon)$. Thus,
\begin{align*}
    \E(\eps_{i1}^4) &= \int_{\err \leq \epsilon} \err^4h_{\err}(\err)d\err + \int_{\err > \epsilon} \err^4h_{\err}(\err)d\err \\
    &\leq \epsilon^4 + \E\brac{\exp(C_\err \abs{\eps})} \\
    & = O(1) \text{~by Assumption \ref{assp:exponential_tails}},
\end{align*}
where $h_\err$ is the density function of $\err$.

Furthermore,
\begin{align*}
    \E[\brac{X_i(T_{ij})-\mu(T_{ij})}^4] &\leq C(\E[\brac{X_i(T_{ij})-\mu(T_{ij})}^2])^2\\
    &\leq C I_{t,t}^2 \\
    & = O(1),
\end{align*}
where the last line uses the the marginal density function for $T$, and $I_{t,t}$ is bounded from above by Assumption \ref{assp:data_gm} and Assumption \ref{assp:smooth_of_X}. Therefore, combining the above two results, and Assumption \ref{assp:data_gm} through \ref{assp:smooth_of_X}, we can conclude that $$\eqref{eq:4th_moment_tilde_Theta} = O(1/n^2).$$
\end{proof}

The final technical lemma calculates the decay rate for the supremum of $\widetilde{\Theta}_i$. A general statement without restrictions on $h$ is given below. This lemma is applied when one needs to verify the Lindeberg's condition.
\begin{lemma}\label{lem:exp_tail_prob_Theta}
Suppose Assumptions \ref{assp:data_gm} - \ref{assp:exponential_tails} hold, then for any $e > 0$ and $i=1, 2, \ldots, n$, we have
\begin{equation*}
    \p\brac{\sup_{t\in\T} \abs{\widetilde{\Theta}_i}> e} \overset{a.s.}{\leq} c_0\log^2(n)\exp\brac{-c_1\frac{e(nh)^{1/2}}{\log(n)}} + 2\p\brac{N_i > c_2\log(n)},
\end{equation*}
where $c_2 > C_N^{-1} + 1$ for $C_N$ in Assumption \ref{assp:exponential_tails}.
\end{lemma}
\begin{proof}
This proof is a repetitive one that involves the repeated use of the Borel-Cantelli Lemma and the Markov inequality. 

Notice that we first decompose $\widetilde{\Theta}$ into two parts, which is the same as in \eqref{eq:theta_tilde_upper_bound_as}. Before delving into the detailed calculations, we first establish that for any random variable $T$ whose marginal distribution is bounded between $0$ and infinity, as stated in Assumption \ref{assp:data_gm}, the kernel function $K_T(t)$ is bounded on the order of $O(h^{-1})$. Hence, 
\begin{equation*}
    \sup_{t\in \T}\abs{K_T(t)} = \norm{K_t}_{\sup} \leq C_Kh^{-1/2}\norm{K_t} \leq C_K^2h^{-1},
\end{equation*}
where the final step uses Lemma \ref{prop:sup_norm}. Let $X_{ij} \triangleq X_i(T_{ij})$, we have
\begin{align*}
    &\hspace{1.3em}\p\pbrac{\sup_{t \in \T}\abs{\widetilde{\Theta}_i(t)} > e} \\
    &= \p\sbrac[\bigg]{\sup_{t \in \T}\abs[\Big]{\sum_{j=1}^{N_i}\brac{Y_{ij}-\mu(T_{ij})}K_{T_{ij}}(t)} > \frac{eM}{(nh)^{1/2}}}\\
    &\leq \p\sbrac[\bigg]{\sup_{t \in \T}\abs[\Big]{\sum_{j=1}^{N_i}\brac{Y_{ij}- X_{ij} + X_{ij} - \mu(T_{ij})}K_{T_{ij}}(t)} \geq \frac{en}{(nh)^{1/2}}}\\
    &\leq \p\sbrac[\bigg]{\sup_{t \in \T}\abs[\Big]{\sum_{j=1}^{N_i}\err_{ij}K_{T_{ij}}(t)} \geq \frac{en}{2(nh)^{1/2}}} \numberthis\label{eq:Theta_tail_prob_part_err}\\
    &+ \p\sbrac[\bigg]{\sup_{t \in \T}\abs[\Big]{\sum_{j=1}^{N_i}\brac{X_{ij} - \mu(T_{ij})}K_{T_{ij}}(t)} \geq \frac{en}{2(nh)^{1/2}}}. \numberthis\label{eq:Theta_tail_prob_part_X}
\end{align*}
Note that both the error term $\err$ and $X$ decay at an exponential rate, almost surely,  according to  Assumption \ref{assp:exponential_tails}. 

Specifically, for some $c_1^{'} > C_\err^{-1}$ and $c_2^{'} > C_N^{-1}$,
$$\abs{\err} \overset{a.s.}{\leq} c_1^{'}\log(n)\text{~and~}\norm{X} \overset{a.s.}{\leq} c_2^{'}\log(n).$$ 

The detailed computations are outlined in Lemma \ref{lem:tilde_Theta_upper_bound}. For conciseness, this part specifically elaborates on the detailed analysis of \eqref{eq:Theta_tail_prob_part_err}. 
\begin{align*}
    \eqref{eq:Theta_tail_prob_part_err} &\leq \p\brac{\sup_{t\in \T}\abs{K_T(t)}\Big(\sum_{j=1}^{N_i}\abs{\err_{ij}}\Big) \geq \frac{en^{1/2}}{2h^{1/2}}}\\
    &\leq \p\brac{C_K^2h^{-1}\sum_{j=1}^{N_i}\abs{\err_{ij}} \geq \frac{en^{1/2}}{2h^{1/2}}} \\
    &\leq \p\brac{\sum_{j=1}^{N_i}\abs{\err_{ij}} \geq \frac{e(nh)^{1/2}}{2C_K^2}\mid N_i \leq c_{1}\log(n)}\p\brac{N_i \leq c_1\log(n)} \\
    &+ \p\brac{\sum_{j=1}^{N_i}\abs{\err_{ij}} \geq \frac{e(nh)^{1/2}}{2C_K^2}\mid N_i > c_{1}\log(n)}\p\brac{N_i > c_1\log(n)}\\
    &\leq \p\brac{\sum_{j=1}^{N_i}\abs{\err_{ij}} \geq \frac{e(nh)^{1/2}}{2C_K^2}\mid N_i \leq c_{1}\log(n)} + \p\brac{N_i > c_1\log(n)}\\
    & \triangleq p_1 + p_2, \numberthis \label{eq:prob_err_exponential_tail}
\end{align*}
where $c_1 > C_N^{-1} + 1$ with $C_N$ being the constant defined in Assumption \ref{assp:exponential_tails}. We will first delve into the analysis of $p_1$. 
\begin{align*}
    &\hspace{1.25em}\p\brac{\sum_{j=1}^{N_i}\abs{\err_{ij}} \geq \frac{e(nh)^{1/2}}{2C_K^2}\mid N_i \leq c_{1}\log(n)}\\
    &\leq\sum_{j=1}^{N_i} \p\brac{\abs{\err_{ij}} \geq \frac{e(nh)^{1/2}}{2C_K^2N_i}\mid N_i \leq c_1\log(n)} \\
    &\leq N_i\left(\sum_{k=1}^{\ceil{c_1\log(n)}}\p\brac{N_i = k \mid N_i \leq c_1\log(n)}\right)\\
    &\hspace{3em}\times\p\brac{\abs{\err_{i1}} \geq \frac{e(nh)^{1/2}}{2C_K^2c_1\log(n)}~\Big | ~N_i \leq c_1\log(n)}\\
    &\overset{a.s.}{\leq} 2c_1\log^2(n)\E\brac{\exp(C_\err\abs{\err_{ij}})}\exp\brac{- \frac{eC_\err(nh)^{1/2}}{2C_K^2c_1\log(n)}}.
\end{align*}
 The final line uses Assumption \ref{assp:exponential_tails} and the Markov inequality. In this context, $\ceil{\cdot}$ represents the smallest integer that is great than or equal to $c_1\log(n)$. Thus, we establish the following result:
\begin{align*}
    &\eqref{eq:prob_err_exponential_tail} \overset{a.s.}{\leq} c_2\log^2(n)\exp\brac{-c_3\frac{e(nh)^{1/2}}{\log(n)}} + \p\brac{N_i > c_1\log(n)}.
\end{align*}
The same argument is applicable to \eqref{eq:Theta_tail_prob_part_X}. The only difference lies in the calculation process. 
$$\Pr\brac{\abs{X_i(T_{ij} - \mu(T_{ij})} \geq e(nh)^{1/2}/2C_K^2c_1\log(n)\mid N_i \leq c_1\log(n)}.$$ Note that $\abs{X_{i}(T_{ij}) - \mu(T_{ij})} \leq \norm{X - \mu}_{\sup} \leq h^{-1/2}\norm{X-\mu}$. Thus,
\begin{align*}
    \eqref{eq:Theta_tail_prob_part_X} &\overset{a.s.}{\leq} 2c_1\log^2(n)\E\brac{\exp(C_\err\abs{\err_{ij}})}\exp\brac{- \frac{eC_\err n^{1/2}h}{2C_K^2c_1\log(n)}}\\
    &\leq c_4\log^2(n)\exp\brac{-c_5\frac{en^{1/2}h}{\log(n)}} + \p\brac{N_i > c_1 \log(n)}.
\end{align*}
It's important to note that $nh \geq nh^2$ given that $h=o(1)$. Therefore, in the asymptotic sense, we obtain the following results:
\begin{equation*}
    \p(\sup_{t\in\T}\abs{\widetilde{\Theta}_i(t)} > e) \overset{a.s.}{\leq}  c_7\log^2(n)\exp\brac{-c_8\frac{(nh)^{1/2}}{\log(n)}} + 2\p\brac{N_i > c_1\log(n)}.
\end{equation*}
\end{proof}
\section{Additional simulation results}\label{appendix:extra_numerical_results}
\subsection{Prediction errors}
In this section, the out-of-sample performance (OOSP) of the proposed method, compared with PACE from \cite{FY2005}, is presented in Table \ref{tb:OOSP_comparison_delta_0}.

Note that the OOSP is well-defined only if the data from the two groups are collected over a common grid. Otherwise, $Y_{1ij} - Y_{2i'j'}$ may not be comparable, as responses could be collected at different observation times, or one subject may have more observations than another. Consequently, we present the OOSP results only for the second group. Specifically, let $\D$ denote the entire dataset, from which $10\%$ is randomly selected as the test set. Define
\begin{equation*}
    \text{OOSP} = \sum_{i=1}^{n'}\sum_{j=1}^{N_{i}}\brac{Y_{ij}' - \hat{\mu}(T_{ij}')}^2,
\end{equation*}
where $\hat{\mu}$ is derived from the training set, and $\{(Y_{ij}', T_{ij}')\}$ is taken from the test set. The results are summarized in \ref{tb:OOSP_comparison_delta_0}.
\begin{table}[!ht]
\caption{Prediction errors and standard error (in parenthesis) of OOSPs across 1000 Monte Carlo runs for group 2 with $\mu(t) = (2t-0.3)^3 + 0.5t$.}
\centering
\addtolength{\tabcolsep}{-4.7pt} 
\def\arraystretch{1.5}
\begin{tabular}[t]{c c c c c c c c}
\hline
\hline
\multirow{2}{*}{Setting } & \multirow{2}{*}{$N_{\max}$} & \multicolumn{2}{c}{$n_1=200, n_2 = 100$} & \multicolumn{2}{c}{$n_1=200, n_2 = 200$} & \multicolumn{2}{c}{$n_1=200, n_2 = 400$}\\ 
%\multirow{2}{c}{Designs} & 1 & 2 & 3 & 4 & 5 & 6\\
 & & SS & PACE & SS & PACE & SS & PACE\\
\hline
\multirow{4}{*}{c.1} & $6$& 0.770(0.084) & 0.771(0.085) & 0.741(0.039) & 0.741(0.039) & 0.743(0.020) & 0.743(0.020)\\
& $10$& 0.753(0.082) & 0.753(0.082) & 0.748(0.041) & 0.748(0.041) & 0.736(0.019) & 0.736(0.019)\\
& $14$& 0.748(0.068) & 0.748(0.068) & 0.736(0.034) & 0.736(0.034) & 0.737(0.017) & 0.737(0.017)\\
& $18$ & 0.743(0.073) & 0.743(0.073) & 0.739(0.034) & 0.739(0.034) & 0.736(0.016) & 0.736(0.016)\\
\hline
\multirow{4}{*}{c.2} & $6$& 0.694(0.064) & 0.695(0.065) & 0.676(0.033) & 0.676(0.033) & 0.681(0.015) & 0.681(0.015)\\
& $10$& 0.707(0.063) & 0.707(0.063) & 0.667(0.027) & 0.666(0.027) & 0.680(0.014) & 0.680(0.014)\\
& $14$& 0.678(0.055) & 0.679(0.055)& 0.679(0.026) & 0.679(0.026) & 0.677(0.014) & 0.677(0.014)\\
& $18$ & 0.685(0.052) & 0.685(0.053) & 0.681(0.024) & 0.681(0.024) & 0.686(0.013) & 0.686(0.013)\\
\hline
\multirow{4}{*}{c.3} & $6$& 0.657(0.052) & 0.660(0.053) & 0.644(0.024) & 0.645(0.024) & 0.648(0.013) & 0.649(0.013)\\
& $10$& 0.651(0.046) & 0.652(0.046) & 0.657(0.023) & 0.657(0.023) & 0.645(0.011) & 0.646(0.011)\\
& $14$& 0.659(0.043) & 0.660(0.044) & 0.652(0.023) & 0.653(0.023) & 0.650(0.011) & 0.650(0.011)\\
& $18$ & 0.660(0.040) & 0.661(0.041) & 0.650(0.021) & 0.650(0.021) & 0.647(0.010) & 0.647(0.010)\\
\hline
\end{tabular}
\label{tb:OOSP_comparison_delta_0}
\end{table}
\subsection{Global testing results}
In this section, we provide additional global testing results of the proposed method and the competing methods for $N_{\max} \in \{6, 14, 18\}$. These results are shown in Figures \ref{fig:weak_convergence_N6c1}-\ref{fig:weak_convergence_N18c3}. Corresponding numerical tables are presented in Tables \ref{tb:weak_convergence_N6} - \ref{tb:weak_convergence_N18}.
% N = 6
% Table N = 6
\begin{table}[hb!]
\caption{Rejection rates across 1000 Monte Carlo runs at a 5\% significance level with $N_{\max} = 6$}
\centering
\addtolength{\tabcolsep}{-2.5pt} 
\def\arraystretch{1.3}
\begin{tabular}[t]{c c @{\hspace{1.6em}} c c c @{\hspace{2em}} c c c @{\hspace{2em}} c c c}
\hline
\hline\\[-3ex]
\multirow{2}[2]{*}{Setting} & \multirow{2}[2]{*}{Test} & \multicolumn{3}{c}{\hspace{-2.4em} $n_1=200, n_2 = 100$} & \multicolumn{3}{c}{\hspace{-2.2em}$n_1=200, n_2 = 200$} & \multicolumn{3}{c}{\hspace{-0.5em}$n_1=200, n_2 = 400$}\\[0.3ex]
%\multirow{2}{c}{Designs} & 1 & 2 & 3 & 4 & 5 & 6\\
\cline{3-11}\\[-3ex]
 & & $\delta = 0$ & $\delta = 0.5$ & $\delta = 1$ & $\delta = 0$ & $\delta = 0.5$ & $\delta = 1$ & $\delta = 0$ & $\delta = 0.5$ & $\delta = 1$\\[0.5ex]
\hline\\[-3.5ex]
\multirow{3}{*}{c.1} & \textbf{SS} & 0.079 & 0.314 & 0.924 & 0.064 & 0.371 & 0.944 & 0.047 & 0.350 & 0.941\\[0.2ex]
& SH & 0.230 & 0.438 & 0.889 & 0.244 & 0.521 & 0.921 & 0.278 & 0.510 & 0.915\\
& pffr & 0.340 & 0.590 & 0.927 & 0.306 & 0.623 & 0.949 & 0.325 & 0.596 & 0.943 \\
\hline
\multirow{3}{*}{c.2} & \textbf{SS} & 0.072 & 0.360 & 0.930 & 0.043 & 0.367 & 0.950 & 0.053 & 0.371 & 0.945\\[0.2ex]
& SH & 0.218 & 0.452 & 0.906 & 0.236 & 0.513 & 0.943 & 0.257 & 0.547 & 0.944\\
& pffr & 0.313 & 0.614 & 0.940 & 0.309 & 0.606 & 0.964 & 0.343 & 0.625 & 0.953\\
\hline\\[-3.5ex]
\multirow{3}{*}{c.3} & \textbf{SS} & 0.078 & 0.302 & 0.852 & 0.056 & 0.237 & 0.797 & 0.047 & 0.229 & 0.729\\[0.2ex]
& SH & 0.229 & 0.441 & 0.907 & 0.149 & 0.419 & 0.905 & 0.140 & 0.430 & 0.890\\
& pffr & 0.300 & 0.690 & 0.967 & 0.318 & 0.685 & 0.975 & 0.354 & 0.700 & 0.955\\
\hline
\hline
\end{tabular}
\label{tb:weak_convergence_N6}
\end{table}

% Figures N = 6
\begin{figure}[!ht]
    \centering
    \caption{Rejection rates across 1000 Monte Carlo runs at a 5\% significance level with $N_{\max} = 6$ with $n_1 = 100$ under setting c.1.}
    \includegraphics[scale = 0.67]{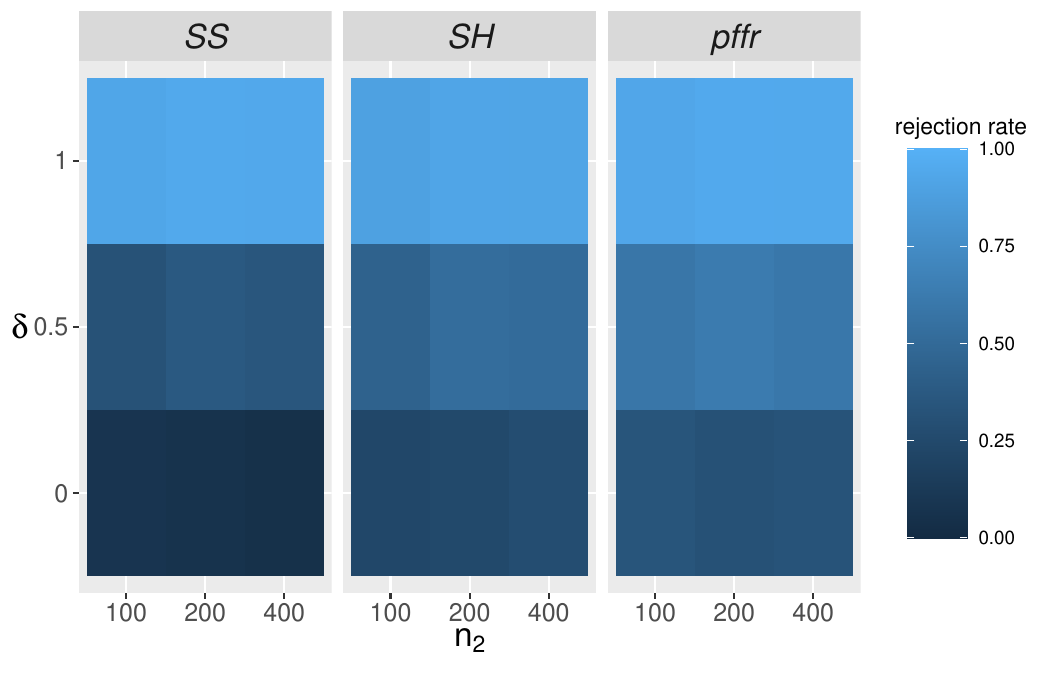}
    \label{fig:weak_convergence_N6c1}
\end{figure}

\begin{figure}[!hb]
    \centering
    \caption{Rejection rates across 1000 Monte Carlo runs at a 5\% significance level with $N_{\max} = 6$ with $n_1 = 100$ under setting c.2.}
    \includegraphics[scale = 0.67]{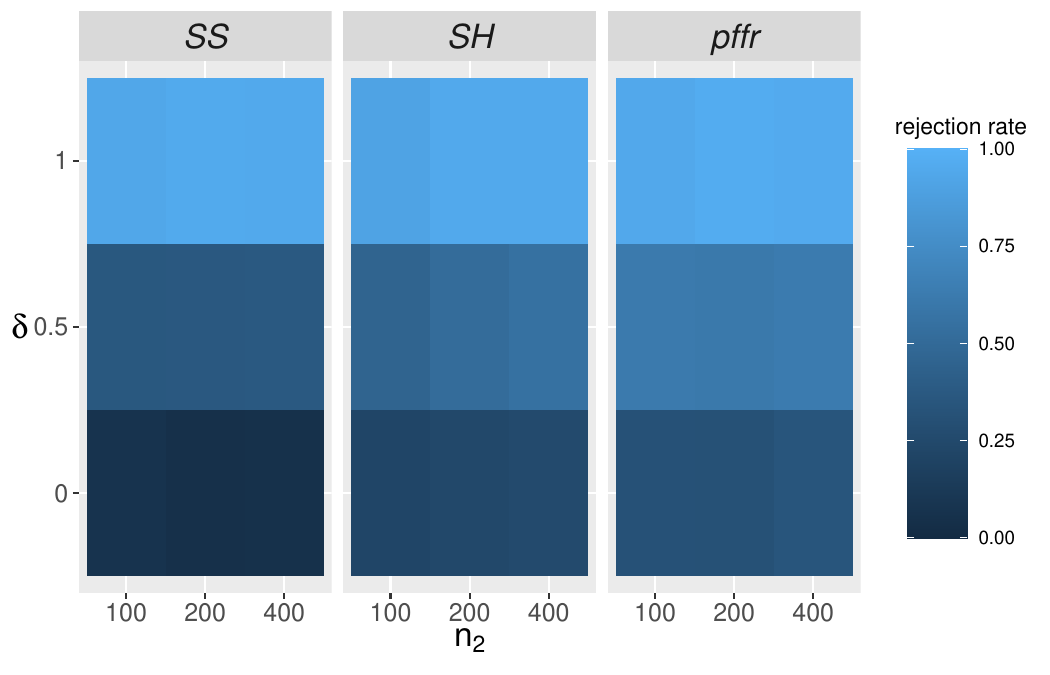}
    \label{fig:weak_convergence_N6c2}
\end{figure}

\clearpage
\begin{figure}[!hb]
    \centering
    \caption{Rejection rates across 1000 Monte Carlo runs at a 5\% significance level with $N_{\max} = 6$ with $n_1 = 100$ under setting c.3.}
    \includegraphics[scale = 0.67]{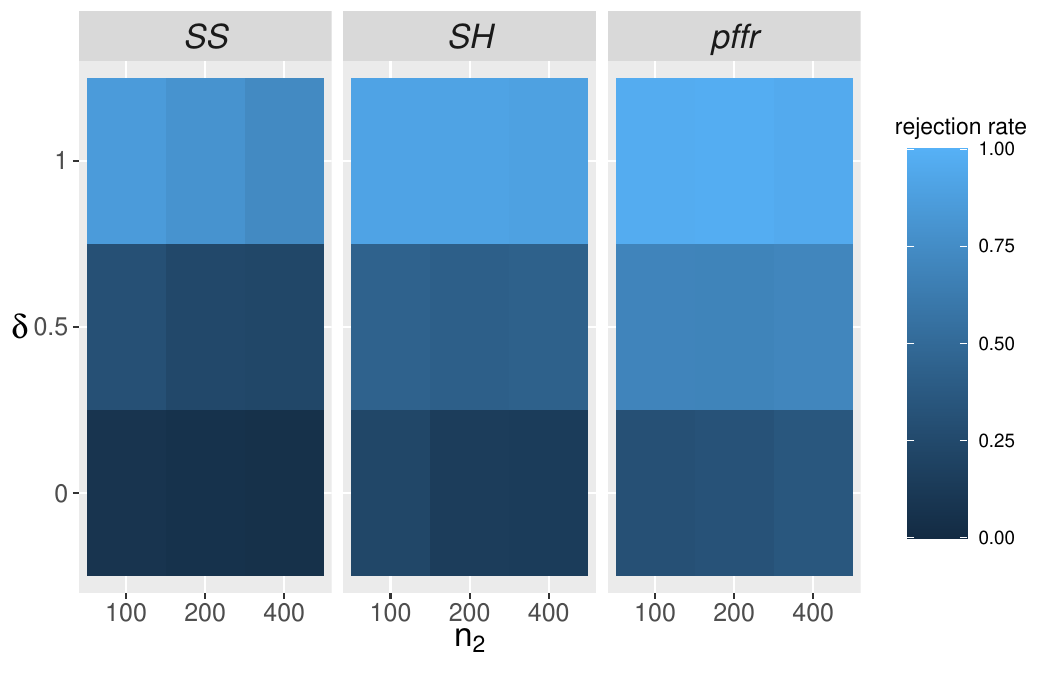}
    \label{fig:weak_convergence_N6c3}
\end{figure}
% N = 14
% Table N = 14
\begin{table}[!hb]
\caption{Rejection rates across 1000 Monte Carlo runs at a 5\% significance level with $N_{\max} = 14$}
\centering
\addtolength{\tabcolsep}{-2.5pt} 
\def\arraystretch{1.3}
\begin{tabular}[t]{c c @{\hspace{1.6em}} c c c @{\hspace{2em}} c c c @{\hspace{2em}} c c c}
\hline
\hline\\[-3ex]
\multirow{2}[2]{*}{Setting} & \multirow{2}[2]{*}{Test} & \multicolumn{3}{c}{\hspace{-2.4em} $n_1=200, n_2 = 100$} & \multicolumn{3}{c}{\hspace{-2.2em}$n_1=200, n_2 = 200$} & \multicolumn{3}{c}{\hspace{-0.5em}$n_1=200, n_2 = 400$}\\[0.3ex]
%\multirow{2}{c}{Designs} & 1 & 2 & 3 & 4 & 5 & 6\\
\cline{3-11}\\[-3ex]
 & & $\delta = 0$ & $\delta = 0.5$ & $\delta = 1$ & $\delta = 0$ & $\delta = 0.5$ & $\delta = 1$ & $\delta = 0$ & $\delta = 0.5$ & $\delta = 1$\\[0.5ex]
\hline\\[-3.5ex]
\multirow{3}{*}{c.1} & \textbf{SS} & 0.057 & 0.405 & 0.974 & 0.053 & 0.453 & 0.989 & 0.047 & 0.446 & 0.991\\[0.2ex]
& SH & 0.206 & 0.527 & 0.955 & 0.232 & 0.577 & 0.974 & 0.251 & 0.563 & 0.972\\
& pffr & 0.533 & 0.778 & 0.984 & 0.534 & 0.814 & 0.988 & 0.524 & 0.791 & 0.988 \\
\hline
\multirow{3}{*}{c.2} & \textbf{SS} & 0.062 & 0.428 & 0.985 & 0.053 & 0.472 & 0.992 & 0.034 & 0.433 & 0.991\\[0.2ex]
& SH & 0.187 & 0.520 & 0.975 & 0.187 & 0.595 & 0.980 & 0.191 & 0.592 & 0.982\\
& pffr & 0.533 & 0.789 & 0.989 & 0.580 & 0.812 & 0.989 & 0.558 & 0.819 & 0.992\\
\hline\\[-3.5ex]
\multirow{3}{*}{c.3} & \textbf{SS} & 0.049 & 0.278 & 0.878 & 0.052 & 0.253 & 0.854 & 0.036 & 0.224 & 0.807\\[0.2ex]
& SH & 0.209 & 0.531 & 0.968 & 0.125 & 0.530 & 0.970 & 0.133 & 0.510 & 0.964\\
& pffr & 0.515 & 0.858 & 0.999 & 0.538 & 0.875 & 0.999 & 0.586 & 0.841 & 0.990\\
\hline
\hline
\end{tabular}
\label{tb:weak_convergence_N14}
\end{table}

% Figures N = 14
\begin{figure}[!ht]
    \centering
    \caption{Rejection rates across 1000 Monte Carlo runs at a 5\% significance level with $N_{\max} = 14$ with $n_1 = 100$ under setting c.1.}
    \includegraphics[scale = 0.67]{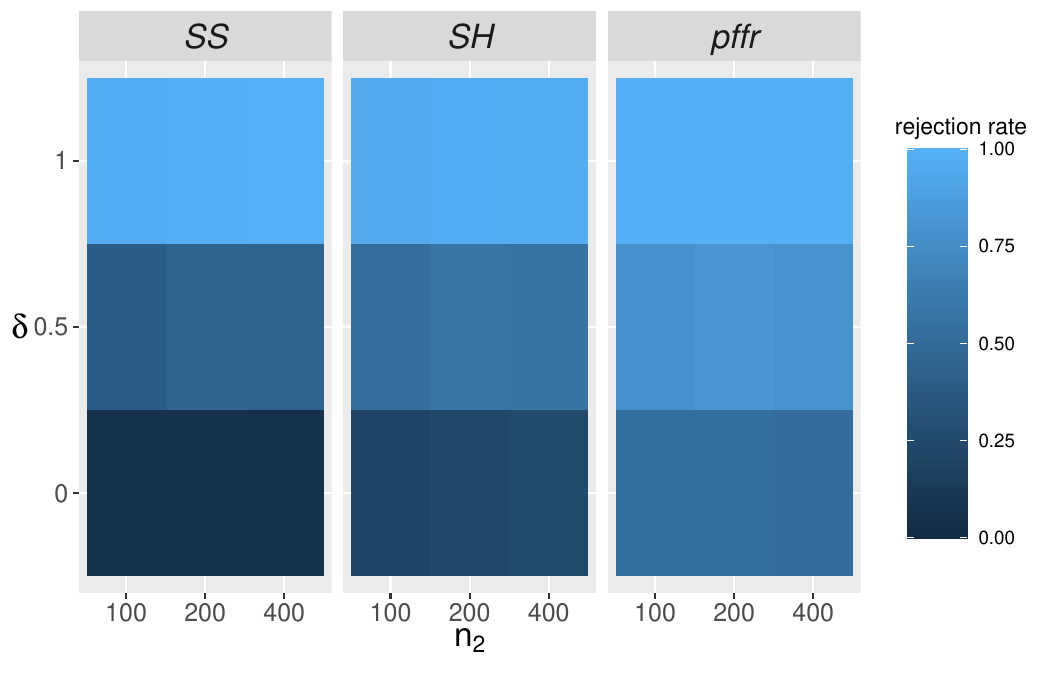}
    \label{fig:weak_convergence_N14c1}
\end{figure}
\begin{figure}[!hb]
    \centering
    \vspace{-2em}
    \caption{Rejection rates across 1000 Monte Carlo runs at a 5\% significance level with $N_{\max} = 14$ with $n_1 = 100$ under setting c.2.}
    \includegraphics[scale = 0.67]{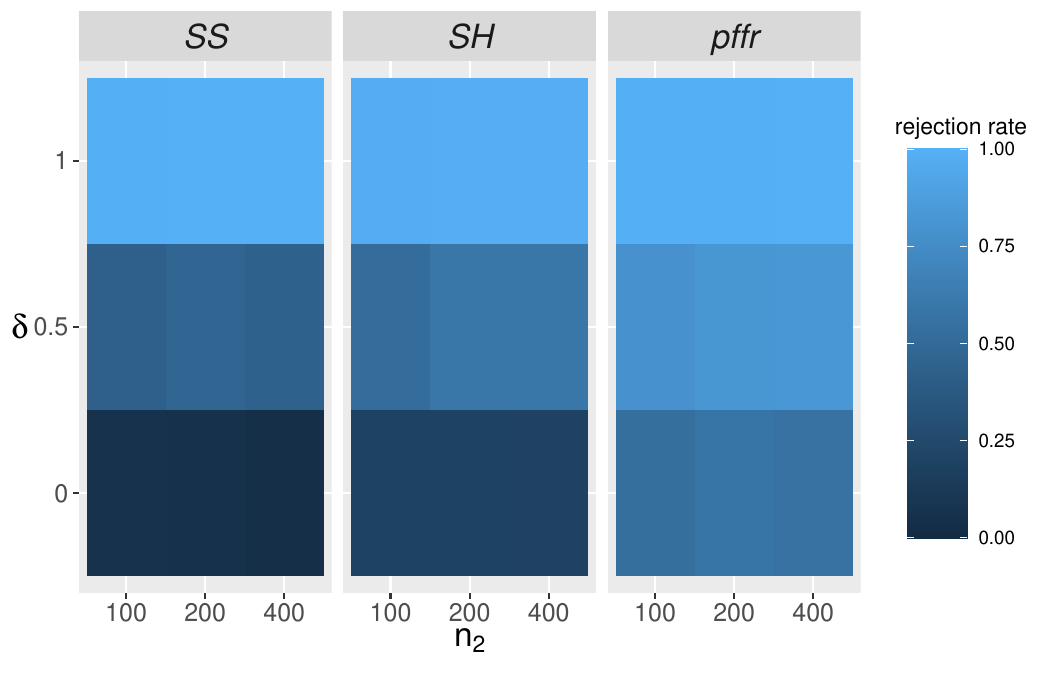}
    \label{fig:weak_convergence_N14c2}
\end{figure}

\clearpage
\begin{figure}[!ht]
    \centering
    \caption{Rejection rates across 1000 Monte Carlo runs at a 5\% significance level with $N_{\max} = 14$ with $n_1 = 100$ under setting c.3.}
    \includegraphics[scale = 0.67]{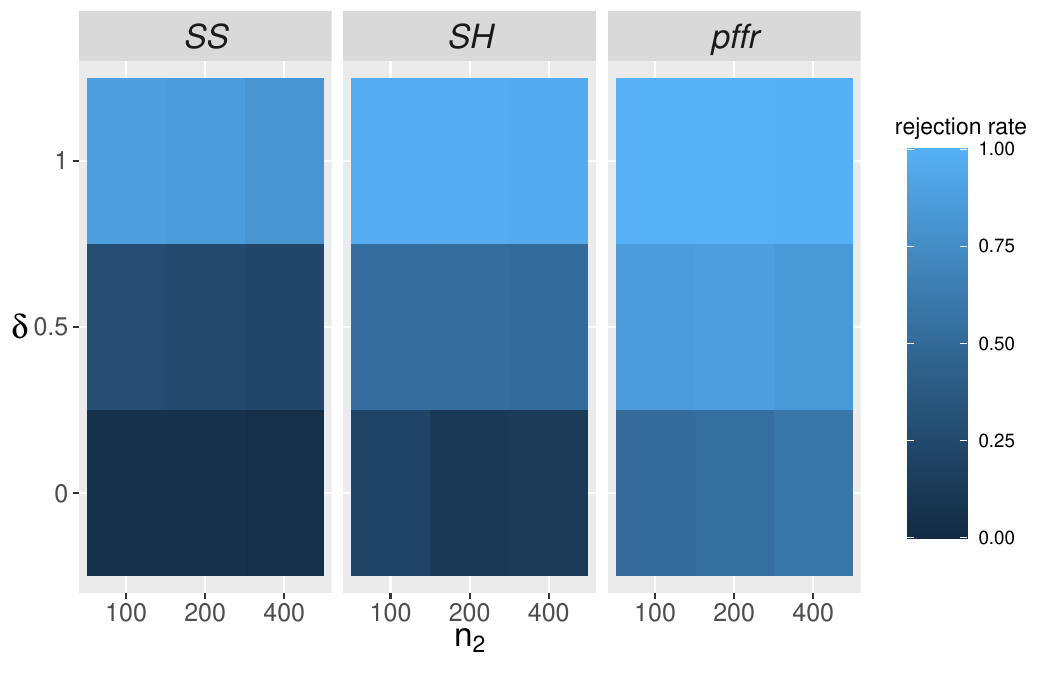}
    \label{fig:weak_convergence_N14c3}
\end{figure}

% N = 18
% Table N = 18
\begin{table}[!hb]
\caption{Rejection rates across 1000 Monte Carlo runs at a 5\% significance level with $N_{\max} = 18$}
\centering
\addtolength{\tabcolsep}{-2.5pt} 
\def\arraystretch{1.3}
\begin{tabular}[t]{c c @{\hspace{1.6em}} c c c @{\hspace{2em}} c c c @{\hspace{2em}} c c c}
\hline
\hline\\[-3ex]
\multirow{2}[2]{*}{Setting} & \multirow{2}[2]{*}{Test} & \multicolumn{3}{c}{\hspace{-2.4em} $n_1=200, n_2 = 100$} & \multicolumn{3}{c}{\hspace{-2.2em}$n_1=200, n_2 = 200$} & \multicolumn{3}{c}{\hspace{-0.5em}$n_1=200, n_2 = 400$}\\[0.3ex]
%\multirow{2}{c}{Designs} & 1 & 2 & 3 & 4 & 5 & 6\\
\cline{3-11}\\[-3ex]
 & & $\delta = 0$ & $\delta = 0.5$ & $\delta = 1$ & $\delta = 0$ & $\delta = 0.5$ & $\delta = 1$ & $\delta = 0$ & $\delta = 0.5$ & $\delta = 1$\\[0.5ex]
\hline\\[-3.5ex]
\multirow{3}{*}{c.1} & \textbf{SS} & 0.048 & 0.425 & 0.989 & 0.051 & 0.472 & 0.998 & 0.047 & 0.438 & 0.991\\[0.2ex]
& SH & 0.184 & 0.566 & 0.975 & 0.216 & 0.586 & 0.974 & 0.214 & 0.608 & 0.973\\
& pffr & 0.572 & 0.837 & 0.992 & 0.603 & 0.853 & 0.996 & 0.624 & 0.851 & 0.993 \\
\hline
\multirow{3}{*}{c.2} & \textbf{SS} & 0.057 & 0.454 & 0.993 & 0.044 & 0.504 & 0.994 & 0.037 & 0.494 & 0.994\\[0.2ex]
& SH & 0.170 & 0.533 & 0.974 & 0.192 & 0.617 & 0.979 & 0.175 & 0.599 & 0.985\\
& pffr & 0.605 & 0.831 & 0.992 & 0.607 & 0.869 & 0.990 & 0.638 & 0.872 & 0.991\\
\hline\\[-3.5ex]
\multirow{3}{*}{c.3} & \textbf{SS} & 0.044 & 0.300 & 0.899 & 0.048 & 0.221 & 0.863 & 0.041 & 0.249 & 0.802\\[0.2ex]
& SH & 0.210 & 0.558 & 0.982 & 0.128 & 0.518 & 0.987 & 0.130 & 0.559 & 0.977\\
& pffr & 0.600 & 0.920 & 1.000 & 0.642 & 0.901 & 0.999 & 0.637 & 0.884 & 0.993\\
\hline
\hline
\end{tabular}
\label{tb:weak_convergence_N18}
\end{table}

% Figures N = 18
\begin{figure}[!ht]
    \centering
    \caption{Rejection rates across 1000 Monte Carlo runs at a 5\% significance level with $N_{\max} = 18$ with $n_1 = 100$ under setting c.1.}
    \includegraphics[scale = 0.67]{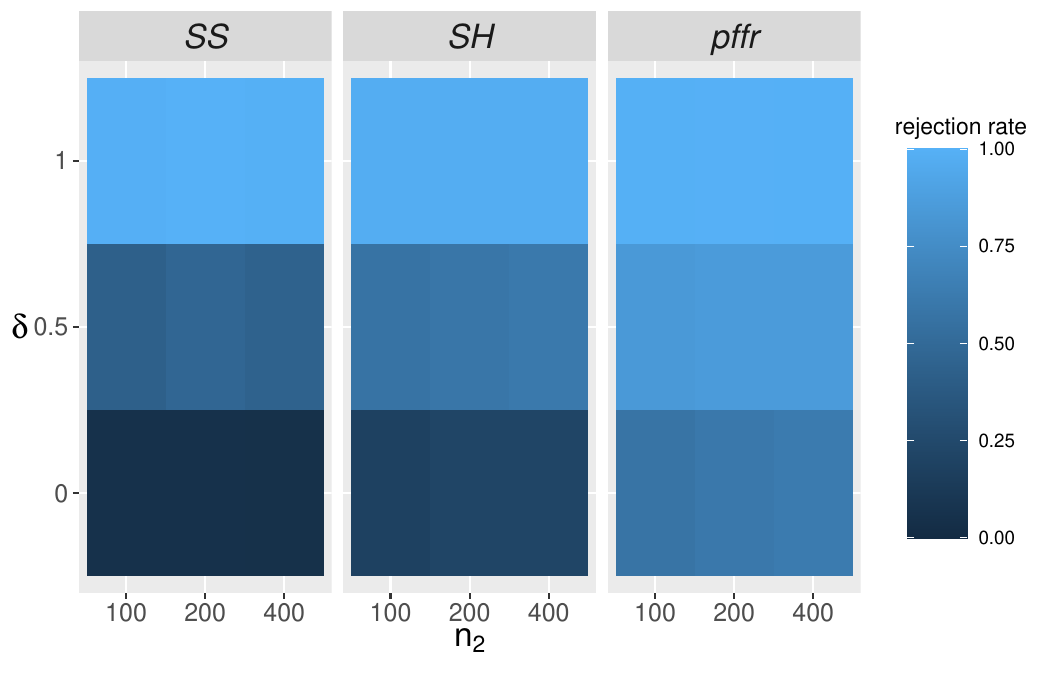}
    \label{fig:weak_convergence_N18c1}
\end{figure}

\begin{figure}[!hb]
    \centering
    \vspace{-2em}
    \caption{Rejection rates across 1000 Monte Carlo runs at a 5\% significance level with $N_{\max} = 18$ with $n_1 = 100$ under setting c.2.}
    \includegraphics[scale = 0.67]{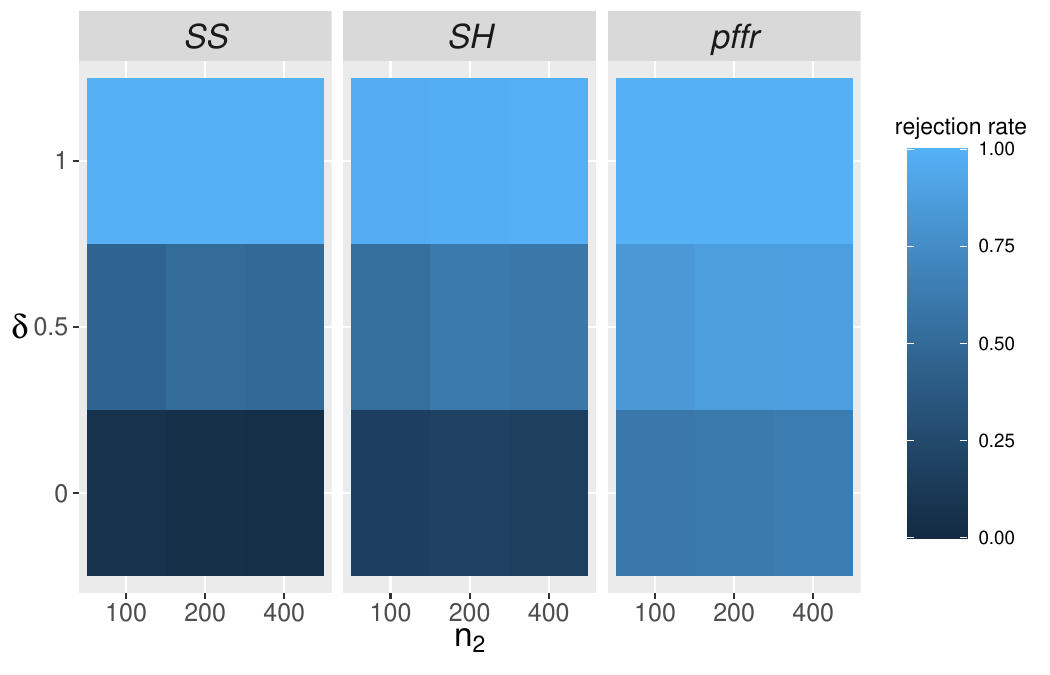}
    \label{fig:weak_convergence_N18c2}
\end{figure}

\begin{figure}[!ht]
    \centering
    \caption{Rejection rates across 1000 Monte Carlo runs at a 5\% significance level with $N_{\max} = 18$ with $n_1 = 100$ under setting c.3.}
    \includegraphics[scale = 0.67]{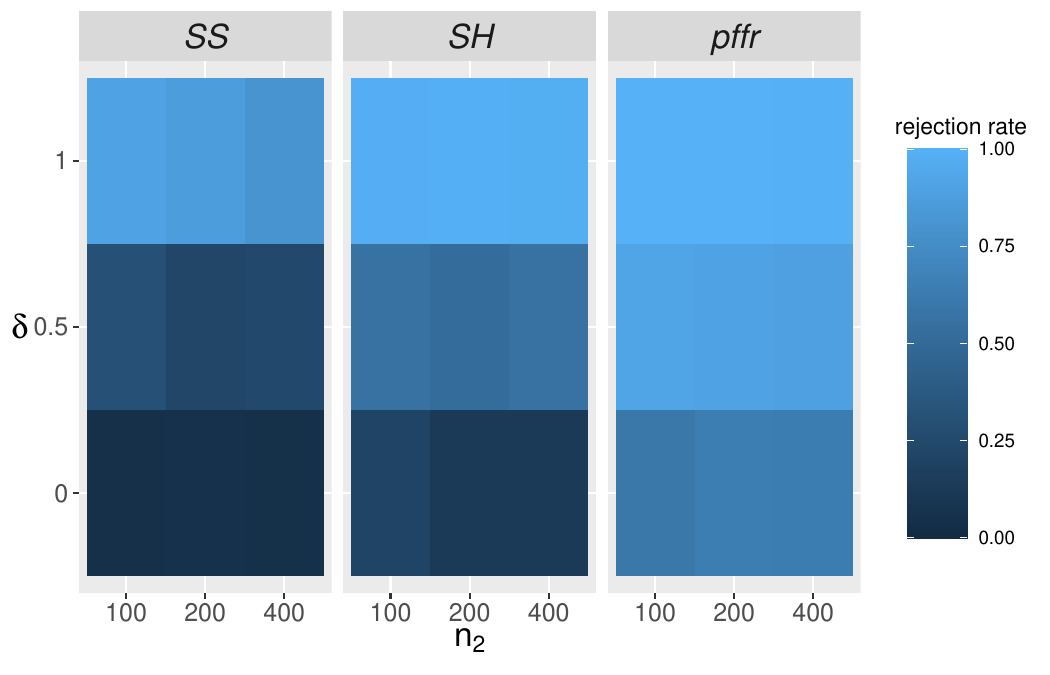}
    \label{fig:weak_convergence_N18c3}
\end{figure}
\end{appendix}
\clearpage
\bibliographystyle{imsart-nameyear.bst} % Style BST file (imsart-number.bst or imsart-nameyear.bst)
\bibliography{bib.bib}       % Bibliography file (usually '*.bib')

\begin{thebibliography}{48}
% BibTex style file: imsart-nameyear.bst, 2017-11-03
% Default style options (sort=1,type=nameyear).
% Used options (sort=1,type=nameyear).

\bibitem[\protect\citeauthoryear{Anderson}{1962}]{AndersonCMTest1962}
\begin{barticle}[author]
\bauthor{\bsnm{Anderson},~\bfnm{T.~W.}\binits{T.~W.}}
(\byear{1962}).
\btitle{{On the distribution of the two-sample Cramer-von Mises criterion}}.
\bjournal{The Annals of Mathematical Statistics}
\bvolume{33}
\bpages{1148-1159}.
\bdoi{10.1214/aoms/1177704477}
\end{barticle}
\endbibitem

\bibitem[\protect\citeauthoryear{Bücher and
  Kojadinovic}{2019}]{bucher_note_2019}
\begin{barticle}[author]
\bauthor{\bsnm{Bücher},~\bfnm{Axel}\binits{A.}} \AND
  \bauthor{\bsnm{Kojadinovic},~\bfnm{Ivan}\binits{I.}}
(\byear{2019}).
\btitle{A Note on Conditional Versus Joint Unconditional Weak Convergence in
  Bootstrap Consistency Results}.
\bjournal{Journal of Theoretical Probability}
\bvolume{32}
\bpages{1145-1165}.
\bdoi{10.1007/s10959-018-0823-3}
\end{barticle}
\endbibitem

\bibitem[\protect\citeauthoryear{Cai and Yuan}{2011}]{TCaiAOS11}
\begin{barticle}[author]
\bauthor{\bsnm{Cai},~\bfnm{T.~Tony}\binits{T.~T.}} \AND
  \bauthor{\bsnm{Yuan},~\bfnm{Ming}\binits{M.}}
(\byear{2011}).
\btitle{{Optimal estimation of the mean function based on discretely sampled
  functional data: Phase transition}}.
\bjournal{The Annals of Statistics}
\bvolume{39}
\bpages{2330-2355}.
\bdoi{10.1214/11-AOS898}
\end{barticle}
\endbibitem

\bibitem[\protect\citeauthoryear{Cardot}{2000}]{Cardot2000JNS}
\begin{barticle}[author]
\bauthor{\bsnm{Cardot},~\bfnm{Hervé}\binits{H.}}
(\byear{2000}).
\btitle{Nonparametric estimation of smoothed principal components analysis of
  sampled noisy functions}.
\bjournal{Journal of Nonparametric Statistics}
\bvolume{12}
\bpages{503-538}.
\bdoi{10.1080/10485250008832820}
\end{barticle}
\endbibitem

\bibitem[\protect\citeauthoryear{Craven and Wahba}{1978}]{GCV1979}
\begin{barticle}[author]
\bauthor{\bsnm{Craven},~\bfnm{Peter}\binits{P.}} \AND
  \bauthor{\bsnm{Wahba},~\bfnm{Grace}\binits{G.}}
(\byear{1978}).
\btitle{Smoothing noisy data with spline functions}.
\bjournal{Numerische Mathematik}
\bvolume{31}
\bpages{377-403}.
\bdoi{10.1007/BF01404567}
\end{barticle}
\endbibitem

\bibitem[\protect\citeauthoryear{Goldsmith et~al.}{2023}]{refundRPkg}
\begin{bmanual}[author]
\bauthor{\bsnm{Goldsmith},~\bfnm{Jeff}\binits{J.}},
  \bauthor{\bsnm{Scheipl},~\bfnm{Fabian}\binits{F.}},
  \bauthor{\bsnm{Huang},~\bfnm{Lei}\binits{L.}},
  \bauthor{\bsnm{Wrobel},~\bfnm{Julia}\binits{J.}},
  \bauthor{\bsnm{Di},~\bfnm{Chongzhi}\binits{C.}},
  \bauthor{\bsnm{Gellar},~\bfnm{Jonathan}\binits{J.}},
  \bauthor{\bsnm{Harezlak},~\bfnm{Jaroslaw}\binits{J.}},
  \bauthor{\bsnm{McLean},~\bfnm{Mathew~W.}\binits{M.~W.}},
  \bauthor{\bsnm{Swihart},~\bfnm{Bruce}\binits{B.}},
  \bauthor{\bsnm{Xiao},~\bfnm{Luo}\binits{L.}},
  \bauthor{\bsnm{Crainiceanu},~\bfnm{Ciprian}\binits{C.}} \AND
  \bauthor{\bsnm{Reiss},~\bfnm{Philip~T.}\binits{P.~T.}}
(\byear{2023}).
\btitle{refund: Regression with Functional Data}
\bnote{R package version 0.1-32}.
\end{bmanual}
\endbibitem

\bibitem[\protect\citeauthoryear{Gu}{2013}]{gu2013smoothing}
\begin{bbook}[author]
\bauthor{\bsnm{Gu},~\bfnm{Chong}\binits{C.}}
(\byear{2013}).
\btitle{Smoothing Spline ANOVA Models},
\bedition{2} ed.
\bpublisher{Springer}, \baddress{New York}.
\end{bbook}
\endbibitem

\bibitem[\protect\citeauthoryear{Hall and Horowitz}{2007}]{Hall2007Annals}
\begin{barticle}[author]
\bauthor{\bsnm{Hall},~\bfnm{Peter}\binits{P.}} \AND
  \bauthor{\bsnm{Horowitz},~\bfnm{Joel~L.}\binits{J.~L.}}
(\byear{2007}).
\btitle{{Methodology and convergence rates for functional linear regression}}.
\bjournal{The Annals of Statistics}
\bvolume{35}
\bpages{70-91}.
\bdoi{10.1214/009053606000000957}
\end{barticle}
\endbibitem

\bibitem[\protect\citeauthoryear{Hao et~al.}{2021}]{semiCoxJASA2021}
\begin{barticle}[author]
\bauthor{\bsnm{Hao},~\bfnm{Meiling}\binits{M.}},
  \bauthor{\bsnm{Liu},~\bfnm{Kin-Yat}\binits{K.-Y.}},
  \bauthor{\bsnm{Xu},~\bfnm{Wei}\binits{W.}} \AND
  \bauthor{\bsnm{Zhao},~\bfnm{Xingqiu}\binits{X.}}
(\byear{2021}).
\btitle{Semiparametric inference for the functional Cox Model}.
\bjournal{Journal of the American Statistical Association}
\bvolume{116}
\bpages{1319-1329}.
\bdoi{10.1080/01621459.2019.1710155}
\end{barticle}
\endbibitem

\bibitem[\protect\citeauthoryear{Hoek et~al.}{2013}]{hoek_long-term_2013}
\begin{barticle}[author]
\bauthor{\bsnm{Hoek},~\bfnm{Gerard}\binits{G.}},
  \bauthor{\bsnm{Krishnan},~\bfnm{Ranjini~M.}\binits{R.~M.}},
  \bauthor{\bsnm{Beelen},~\bfnm{Rob}\binits{R.}},
  \bauthor{\bsnm{Peters},~\bfnm{Annette}\binits{A.}},
  \bauthor{\bsnm{Ostro},~\bfnm{Bart}\binits{B.}},
  \bauthor{\bsnm{Brunekreef},~\bfnm{Bert}\binits{B.}} \AND
  \bauthor{\bsnm{Kaufman},~\bfnm{Joel~D.}\binits{J.~D.}}
(\byear{2013}).
\btitle{Long-term air pollution exposure and cardio- respiratory mortality: a
  review}.
\bjournal{Environmental Health}
\bvolume{12}
\bpages{43}.
\end{barticle}
\endbibitem

\bibitem[\protect\citeauthoryear{Hsing and Eubank}{2015}]{Hsing2015FDA}
\begin{bbook}[author]
\bauthor{\bsnm{Hsing},~\bfnm{Tailen}\binits{T.}} \AND
  \bauthor{\bsnm{Eubank},~\bfnm{Randall}\binits{R.}}
(\byear{2015}).
\btitle{Theoretical Foundations of Functional Data Analysis, with an
  Introduction to Linear Operators}.
\bpublisher{John Wiley \& Sons}, \baddress{Chichester}.
\end{bbook}
\endbibitem

\bibitem[\protect\citeauthoryear{Kley et~al.}{2016}]{Kley2016Bernoulli}
\begin{barticle}[author]
\bauthor{\bsnm{Kley},~\bfnm{Tobias}\binits{T.}},
  \bauthor{\bsnm{Volgushev},~\bfnm{Stanislav}\binits{S.}},
  \bauthor{\bsnm{Dette},~\bfnm{Holger}\binits{H.}} \AND
  \bauthor{\bsnm{Hallin},~\bfnm{Marc}\binits{M.}}
(\byear{2016}).
\btitle{{Quantile spectral processes: Asymptotic analysis and inference}}.
\bjournal{Bernoulli}
\bvolume{22}
\bpages{1770-1807}.
\bdoi{10.3150/15-BEJ711}
\end{barticle}
\endbibitem

\bibitem[\protect\citeauthoryear{Kokoszka and Reimherr}{2017}]{KokoszkaFDA2017}
\begin{bbook}[author]
\bauthor{\bsnm{Kokoszka},~\bfnm{Piotr}\binits{P.}} \AND
  \bauthor{\bsnm{Reimherr},~\bfnm{Matthew}\binits{M.}}
(\byear{2017}).
\btitle{Introduction to Functional Data Analysis}.
\bpublisher{Chapman and Hall/CRC}, \baddress{New York}.
\end{bbook}
\endbibitem

\bibitem[\protect\citeauthoryear{Kokoszka, Rice and
  Shang}{2017}]{sumWeightedChi2Greg2017}
\begin{barticle}[author]
\bauthor{\bsnm{Kokoszka},~\bfnm{Piotr}\binits{P.}},
  \bauthor{\bsnm{Rice},~\bfnm{Gregory}\binits{G.}} \AND
  \bauthor{\bsnm{Shang},~\bfnm{Han~Lin}\binits{H.~L.}}
(\byear{2017}).
\btitle{Inference for the autocovariance of a functional time series under
  conditional heteroscedasticity}.
\bjournal{Journal of Multivariate Analysis}
\bvolume{162}
\bpages{32-50}.
\bdoi{doi.org/10.1016/j.jmva.2017.08.004}
\end{barticle}
\endbibitem

\bibitem[\protect\citeauthoryear{Kong et~al.}{2016}]{kong2016partially}
\begin{barticle}[author]
\bauthor{\bsnm{Kong},~\bfnm{Dehan}\binits{D.}},
  \bauthor{\bsnm{Xue},~\bfnm{Kaijie}\binits{K.}},
  \bauthor{\bsnm{Yao},~\bfnm{Fang}\binits{F.}} \AND
  \bauthor{\bsnm{Zhang},~\bfnm{Hao~H}\binits{H.~H.}}
(\byear{2016}).
\btitle{Partially functional linear regression in high dimensions}.
\bjournal{Biometrika}
\bvolume{103}
\bpages{147--159}.
\end{barticle}
\endbibitem

\bibitem[\protect\citeauthoryear{Lelieveld et~al.}{2015}]{Lelieveld2015Nature}
\begin{barticle}[author]
\bauthor{\bsnm{Lelieveld},~\bfnm{J}\binits{J.}},
  \bauthor{\bsnm{Evans},~\bfnm{John}\binits{J.}},
  \bauthor{\bsnm{Fnais},~\bfnm{M}\binits{M.}},
  \bauthor{\bsnm{Giannadaki},~\bfnm{Despina}\binits{D.}} \AND
  \bauthor{\bsnm{Pozzer},~\bfnm{A}\binits{A.}}
(\byear{2015}).
\btitle{The contribution of outdoor air pollution sources to premature
  mortality on a global scale}.
\bjournal{Nature}
\bvolume{525}
\bpages{367-71}.
\bdoi{10.1038/nature15371}
\end{barticle}
\endbibitem

\bibitem[\protect\citeauthoryear{Li and Hsing}{2010}]{li2010uniform}
\begin{barticle}[author]
\bauthor{\bsnm{Li},~\bfnm{Yehua}\binits{Y.}} \AND
  \bauthor{\bsnm{Hsing},~\bfnm{Tailen}\binits{T.}}
(\byear{2010}).
\btitle{Uniform convergence rates for nonparametric regression and principal
  component analysis in functional/longitudinal data}.
\bjournal{The Annals of Statistics}
\bvolume{38}
\bpages{3321-3351}.
\bdoi{10.1214/10-AOS813}
\end{barticle}
\endbibitem

\bibitem[\protect\citeauthoryear{Liang et~al.}{2015}]{LiangJRSSA2015}
\begin{barticle}[author]
\bauthor{\bsnm{Liang},~\bfnm{Xuan}\binits{X.}},
  \bauthor{\bsnm{Zou},~\bfnm{Tao}\binits{T.}},
  \bauthor{\bsnm{Guo},~\bfnm{Bin}\binits{B.}},
  \bauthor{\bsnm{Li},~\bfnm{Shuo}\binits{S.}},
  \bauthor{\bsnm{Zhang},~\bfnm{Haozhe}\binits{H.}},
  \bauthor{\bsnm{Zhang},~\bfnm{Shuyi}\binits{S.}},
  \bauthor{\bsnm{Huang},~\bfnm{Hui}\binits{H.}} \AND
  \bauthor{\bsnm{Chen},~\bfnm{Song~Xi}\binits{S.~X.}}
(\byear{2015}).
\btitle{Assessing {Beijing}'s {PM$_{2.5}$} pollution: severity, weather impact,
  {APEC} and winter heating}.
\bjournal{Proceedings of the Royal Society A: Mathematical, Physical and
  Engineering Sciences}
\bvolume{471}
\bpages{20150257}.
\bdoi{10.1098/rspa.2015.0257}
\end{barticle}
\endbibitem

\bibitem[\protect\citeauthoryear{Lundborg, Shah and
  Peters}{2022}]{conditional_independence_2022_JRSSB}
\begin{barticle}[author]
\bauthor{\bsnm{Lundborg},~\bfnm{Anton~Rask}\binits{A.~R.}},
  \bauthor{\bsnm{Shah},~\bfnm{Rajen~D.}\binits{R.~D.}} \AND
  \bauthor{\bsnm{Peters},~\bfnm{Jonas}\binits{J.}}
(\byear{2022}).
\btitle{Conditional Independence Testing in {Hilbert} Spaces with Applications
  to Functional Data Analysis}.
\bjournal{Journal of the Royal Statistical Society Series B: Statistical
  Methodology}
\bvolume{84}
\bpages{1821-1850}.
\bdoi{10.1111/rssb.12544}
\end{barticle}
\endbibitem

\bibitem[\protect\citeauthoryear{Ozturk et~al.}{2010}]{DTIBioStudy2010}
\begin{barticle}[author]
\bauthor{\bsnm{Ozturk},~\bfnm{A.}\binits{A.}},
  \bauthor{\bsnm{Smith},~\bfnm{SA}\binits{S.}},
  \bauthor{\bsnm{Gordon-Lipkin},~\bfnm{EM}\binits{E.}},
  \bauthor{\bsnm{Harrison},~\bfnm{DM}\binits{D.}},
  \bauthor{\bsnm{Shiee},~\bfnm{N.}\binits{N.}},
  \bauthor{\bsnm{Pham},~\bfnm{DL}\binits{D.}},
  \bauthor{\bsnm{Caffo},~\bfnm{BS}\binits{B.}},
  \bauthor{\bsnm{Calabresi},~\bfnm{PA}\binits{P.}} \AND
  \bauthor{\bsnm{Reich},~\bfnm{DS}\binits{D.}}
(\byear{2010}).
\btitle{{MRI of the corpus callosum in multiple sclerosis: association with
  disability}}.
\bjournal{Multiple Sclerosis Journal}
\bvolume{16}
\bpages{166-177}.
\bdoi{10.1177/1352458509353649}
\end{barticle}
\endbibitem

\bibitem[\protect\citeauthoryear{Pigoli et~al.}{2014}]{pigoli_distances_2014}
\begin{barticle}[author]
\bauthor{\bsnm{Pigoli},~\bfnm{Davide}\binits{D.}},
  \bauthor{\bsnm{Aston},~\bfnm{John A.~D.}\binits{J.~A.~D.}},
  \bauthor{\bsnm{Dryden},~\bfnm{Ian~L.}\binits{I.~L.}} \AND
  \bauthor{\bsnm{Secchi},~\bfnm{Piercesare}\binits{P.}}
(\byear{2014}).
\btitle{Distances and inference for covariance operators}.
\bjournal{Biometrika}
\bvolume{101}
\bpages{409-422}.
\bdoi{10.1093/biomet/asu008}
\end{barticle}
\endbibitem

\bibitem[\protect\citeauthoryear{Pomann, Staicu and Ghosh}{2016}]{JRSSC2016}
\begin{barticle}[author]
\bauthor{\bsnm{Pomann},~\bfnm{Gina-Maria}\binits{G.-M.}},
  \bauthor{\bsnm{Staicu},~\bfnm{Ana-Maria}\binits{A.-M.}} \AND
  \bauthor{\bsnm{Ghosh},~\bfnm{Sujit}\binits{S.}}
(\byear{2016}).
\btitle{A two-sample distribution-free test for functional data with
  application to a diffusion tensor imaging study of multiple sclerosis}.
\bjournal{Journal of the Royal Statistical Society: Series C (Applied
  Statistics)}
\bvolume{65}
\bpages{395-414}.
\bdoi{https://doi.org/10.1111/rssc.12130}
\end{barticle}
\endbibitem

\bibitem[\protect\citeauthoryear{Pope~III et~al.}{2002}]{PopeJAMA2002}
\begin{barticle}[author]
\bauthor{\bsnm{Pope~III},~\bfnm{C.~Arden}\binits{C.~A.}},
  \bauthor{\bsnm{Burnett},~\bfnm{Richard~T.}\binits{R.~T.}},
  \bauthor{\bsnm{Thun},~\bfnm{Michael~J.}\binits{M.~J.}},
  \bauthor{\bsnm{Calle},~\bfnm{Eugenia~E.}\binits{E.~E.}},
  \bauthor{\bsnm{Krewski},~\bfnm{Daniel}\binits{D.}},
  \bauthor{\bsnm{Ito},~\bfnm{Kazuhiko}\binits{K.}} \AND
  \bauthor{\bsnm{Thurston},~\bfnm{George~D.}\binits{G.~D.}}
(\byear{2002}).
\btitle{{Lung cancer, cardiopulmonary mortality, and long-term exposure to fine
  particulate air pollution}}.
\bjournal{JAMA}
\bvolume{287}
\bpages{1132-1141}.
\bdoi{10.1001/jama.287.9.1132}
\end{barticle}
\endbibitem

\bibitem[\protect\citeauthoryear{Ramsay and Silverman}{2005}]{ramsayFDA2005}
\begin{bbook}[author]
\bauthor{\bsnm{Ramsay},~\bfnm{J.~O.}\binits{J.~O.}} \AND
  \bauthor{\bsnm{Silverman},~\bfnm{B.~W.}\binits{B.~W.}}
(\byear{2005}).
\btitle{Functional Data Analysis},
\bedition{2} ed.
\bpublisher{Springer}, \baddress{New York}.
\end{bbook}
\endbibitem

\bibitem[\protect\citeauthoryear{Rice and
  Silverman}{1991}]{RiceAndSilverman1991JRSSB}
\begin{barticle}[author]
\bauthor{\bsnm{Rice},~\bfnm{John~A.}\binits{J.~A.}} \AND
  \bauthor{\bsnm{Silverman},~\bfnm{B.~W.}\binits{B.~W.}}
(\byear{1991}).
\btitle{Estimating the mean and covariance structure nonparametrically when the
  data are curves}.
\bjournal{Journal of the Royal Statistical Society: Series B (Methodological)}
\bvolume{53}
\bpages{233-243}.
\bdoi{https://doi.org/10.1111/j.2517-6161.1991.tb01821.x}
\end{barticle}
\endbibitem

\bibitem[\protect\citeauthoryear{Romano}{2004}]{uniform_testing_2004}
\begin{barticle}[author]
\bauthor{\bsnm{Romano},~\bfnm{Joseph~P.}\binits{J.~P.}}
(\byear{2004}).
\btitle{On Non-parametric Testing, the Uniform Behaviour of the t-test, and
  Related Problems}.
\bjournal{Scandinavian Journal of Statistics}
\bvolume{31}
\bpages{567-584}.
\bdoi{10.1111/j.1467-9469.2004.00407.x}
\end{barticle}
\endbibitem

\bibitem[\protect\citeauthoryear{Rudin}{1986}]{rudin1976principles}
\begin{bbook}[author]
\bauthor{\bsnm{Rudin},~\bfnm{Walter}\binits{W.}}
(\byear{1986}).
\btitle{Principles of Mathematical Analysis},
\bedition{3rd} ed.
\bpublisher{McGraw - Hill}, \baddress{New York}.
\end{bbook}
\endbibitem

\bibitem[\protect\citeauthoryear{Scheipl, Staicu and
  Greven}{2015}]{scheipl2015functional}
\begin{barticle}[author]
\bauthor{\bsnm{Scheipl},~\bfnm{Fabian}\binits{F.}},
  \bauthor{\bsnm{Staicu},~\bfnm{Ana-Maria}\binits{A.-M.}} \AND
  \bauthor{\bsnm{Greven},~\bfnm{Sonja}\binits{S.}}
(\byear{2015}).
\btitle{Functional additive mixed models}.
\bjournal{Journal of Computational and Graphical Statistics}
\bvolume{24}
\bpages{477--501}.
\end{barticle}
\endbibitem

\bibitem[\protect\citeauthoryear{Shah and
  Peters}{2020}]{conditional_independence_AOS2020}
\begin{barticle}[author]
\bauthor{\bsnm{Shah},~\bfnm{Rajen~D.}\binits{R.~D.}} \AND
  \bauthor{\bsnm{Peters},~\bfnm{Jonas}\binits{J.}}
(\byear{2020}).
\btitle{{The hardness of conditional independence testing and the generalised
  covariance measure}}.
\bjournal{The Annals of Statistics}
\bvolume{48}
\bpages{1514-1538}.
\bdoi{10.1214/19-AOS1857}
\end{barticle}
\endbibitem

\bibitem[\protect\citeauthoryear{Shang and Cheng}{2013}]{zfsannals13}
\begin{barticle}[author]
\bauthor{\bsnm{Shang},~\bfnm{Zuofeng}\binits{Z.}} \AND
  \bauthor{\bsnm{Cheng},~\bfnm{Guang}\binits{G.}}
(\byear{2013}).
\btitle{{Local and global asymptotic inference in smoothing spline models}}.
\bjournal{The Annals of Statistics}
\bvolume{41}
\bpages{2608-2638}.
\bdoi{10.1214/13-AOS1164}
\end{barticle}
\endbibitem

\bibitem[\protect\citeauthoryear{Shang and Cheng}{2015}]{zfsannals15}
\begin{barticle}[author]
\bauthor{\bsnm{Shang},~\bfnm{Zuofeng}\binits{Z.}} \AND
  \bauthor{\bsnm{Cheng},~\bfnm{Guang}\binits{G.}}
(\byear{2015}).
\btitle{{Nonparametric inference in generalized functional linear models}}.
\bjournal{The Annals of Statistics}
\bvolume{43}
\bpages{1742-1773}.
\bdoi{10.1214/15-AOS1322}
\end{barticle}
\endbibitem

\bibitem[\protect\citeauthoryear{Staicu et~al.}{2014}]{LRT2014SJS}
\begin{barticle}[author]
\bauthor{\bsnm{Staicu},~\bfnm{Ana-Maria}\binits{A.-M.}},
  \bauthor{\bsnm{Li},~\bfnm{Yingxing}\binits{Y.}},
  \bauthor{\bsnm{Crainiceanu},~\bfnm{Ciprian~M.}\binits{C.~M.}} \AND
  \bauthor{\bsnm{Ruppert},~\bfnm{David}\binits{D.}}
(\byear{2014}).
\btitle{Likelihood ratio tests for dependent data with applications to
  longitudinal and functional data analysis}.
\bjournal{Scandinavian Journal of Statistics}
\bvolume{41}
\bpages{932-949}.
\bdoi{https://doi.org/10.1111/sjos.12075}
\end{barticle}
\endbibitem

\bibitem[\protect\citeauthoryear{Sun et~al.}{2018}]{sun2018optimal}
\begin{barticle}[author]
\bauthor{\bsnm{Sun},~\bfnm{Xiaoxiao}\binits{X.}},
  \bauthor{\bsnm{Du},~\bfnm{Pang}\binits{P.}},
  \bauthor{\bsnm{Wang},~\bfnm{Xiao}\binits{X.}} \AND
  \bauthor{\bsnm{Ma},~\bfnm{Ping}\binits{P.}}
(\byear{2018}).
\btitle{Optimal penalized function-on-function regression under a reproducing
  kernel {Hilbert} space framework}.
\bjournal{Journal of the American Statistical Association}
\bvolume{113}
\bpages{1601--1611}.
\end{barticle}
\endbibitem

\bibitem[\protect\citeauthoryear{Van~der Vaart}{1998}]{vaart_asymptotic_1998}
\begin{bbook}[author]
\bauthor{\bparticle{Van~der} \bsnm{Vaart},~\bfnm{A.~W.}\binits{A.~W.}}
(\byear{1998}).
\btitle{Asymptotic Statistics}.
\bpublisher{Cambridge University Press}, \baddress{Cambridge}.
\end{bbook}
\endbibitem

\bibitem[\protect\citeauthoryear{Van~der Vaart and Wellner}{1996}]{van1996weak}
\begin{bbook}[author]
\bauthor{\bparticle{Van~der} \bsnm{Vaart},~\bfnm{Aad~W.}\binits{A.~W.}} \AND
  \bauthor{\bsnm{Wellner},~\bfnm{Jon~A.}\binits{J.~A.}}
(\byear{1996}).
\btitle{Weak Convergence and Empirical Processes: With Applications to
  Statistics}.
\bpublisher{Springer}, \baddress{New York}.
\end{bbook}
\endbibitem

\bibitem[\protect\citeauthoryear{Wahba}{1990}]{wahba1990spline}
\begin{bbook}[author]
\bauthor{\bsnm{Wahba},~\bfnm{Grace}\binits{G.}}
(\byear{1990}).
\btitle{Spline Models for Observational Data}.
\bpublisher{Society for Industrial and Applied Mathematics},
  \baddress{Philadelphia}.
\end{bbook}
\endbibitem

\bibitem[\protect\citeauthoryear{Wang}{2021}]{ejs2sample}
\begin{barticle}[author]
\bauthor{\bsnm{Wang},~\bfnm{Qiyao}\binits{Q.}}
(\byear{2021}).
\btitle{Two-sample inference for sparse functional data}.
\bjournal{Electronic Journal of Statistics}
\bvolume{15}
\bpages{1395-1423}.
\bdoi{10.1214/21-EJS1802}
\end{barticle}
\endbibitem

\bibitem[\protect\citeauthoryear{Wang, Chiou and
  M\"{u}ller}{2016}]{Wang2016FDAReview}
\begin{barticle}[author]
\bauthor{\bsnm{Wang},~\bfnm{Jane-Ling}\binits{J.-L.}},
  \bauthor{\bsnm{Chiou},~\bfnm{Jeng-Min}\binits{J.-M.}} \AND
  \bauthor{\bsnm{M\"{u}ller},~\bfnm{Hans-Georg}\binits{H.-G.}}
(\byear{2016}).
\btitle{Functional Data Analysis}.
\bjournal{Annual Review of Statistics and Its Application}
\bvolume{3}
\bpages{257-295}.
\bdoi{10.1146/annurev-statistics-041715-033624}
\end{barticle}
\endbibitem

\bibitem[\protect\citeauthoryear{Wang et~al.}{2023}]{Wang2023Sinica}
\begin{bmisc}[author]
\bauthor{\bsnm{Wang},~\bfnm{Yueying}\binits{Y.}},
  \bauthor{\bsnm{Wang},~\bfnm{Guannan}\binits{G.}},
  \bauthor{\bsnm{Klinedinst},~\bfnm{Brandon}\binits{B.}},
  \bauthor{\bsnm{Willette},~\bfnm{Auriel}\binits{A.}} \AND
  \bauthor{\bsnm{Wang},~\bfnm{Lily}\binits{L.}}
(\byear{2023}).
\btitle{Statistical inference for mean functions of complex {3D} objects}.
\bnote{\mbox{doi}: \url{10.5705/ss.202023.0071}}.
\bdoi{10.5705/ss.202023.0071}
\end{bmisc}
\endbibitem

\bibitem[\protect\citeauthoryear{Waudby-Smith, Kennedy and
  Ramdas}{2024}]{waudbysmith2024arXiv1}
\begin{bmisc}[author]
\bauthor{\bsnm{Waudby-Smith},~\bfnm{Ian}\binits{I.}},
  \bauthor{\bsnm{Kennedy},~\bfnm{Edward~H.}\binits{E.~H.}} \AND
  \bauthor{\bsnm{Ramdas},~\bfnm{Aaditya}\binits{A.}}
(\byear{2024}).
\btitle{Distribution-uniform anytime-valid sequential inference}.
\bnote{\mbox{doi}: \url{https://arxiv.org/abs/2311.03343}}.
\end{bmisc}
\endbibitem

\bibitem[\protect\citeauthoryear{Yao, Müller and Wang}{2005}]{FY2005}
\begin{barticle}[author]
\bauthor{\bsnm{Yao},~\bfnm{Fang}\binits{F.}},
  \bauthor{\bsnm{Müller},~\bfnm{Hans-Georg}\binits{H.-G.}} \AND
  \bauthor{\bsnm{Wang},~\bfnm{Jane-Ling}\binits{J.-L.}}
(\byear{2005}).
\btitle{Functional data analysis for sparse longitudinal data}.
\bjournal{Journal of the American Statistical Association}
\bvolume{100}
\bpages{577-590}.
\bdoi{10.1198/016214504000001745}
\end{barticle}
\endbibitem

\bibitem[\protect\citeauthoryear{Zhang and Chen}{2007}]{JinTingZhangAOS2007}
\begin{barticle}[author]
\bauthor{\bsnm{Zhang},~\bfnm{Jin-Ting}\binits{J.-T.}} \AND
  \bauthor{\bsnm{Chen},~\bfnm{Jianwei}\binits{J.}}
(\byear{2007}).
\btitle{Statistical inferences for functional data}.
\bjournal{The Annals of Statistics}
\bvolume{35}
\bpages{1052-1079}.
\end{barticle}
\endbibitem

\bibitem[\protect\citeauthoryear{Zhang, Chen and Bao}{2023}]{Zhang2023Envir}
\begin{barticle}[author]
\bauthor{\bsnm{Zhang},~\bfnm{Ying}\binits{Y.}},
  \bauthor{\bsnm{Chen},~\bfnm{Song~Xi}\binits{S.~X.}} \AND
  \bauthor{\bsnm{Bao},~\bfnm{Le}\binits{L.}}
(\byear{2023}).
\btitle{Air pollution estimation under air stagnation—A case study of
  {Beijing}}.
\bjournal{Environmetrics}
\bvolume{34}
\bpages{e2819}.
\bdoi{https://doi.org/10.1002/env.2819}
\end{barticle}
\endbibitem

\bibitem[\protect\citeauthoryear{Zhang, Liang and
  Xiao}{2010}]{JinTingZhangJSTP2010}
\begin{barticle}[author]
\bauthor{\bsnm{Zhang},~\bfnm{Jin-Ting}\binits{J.-T.}},
  \bauthor{\bsnm{Liang},~\bfnm{Xuehua}\binits{X.}} \AND
  \bauthor{\bsnm{Xiao},~\bfnm{Shengning}\binits{S.}}
(\byear{2010}).
\btitle{On the two-sample {Behrens-Fisher} problem for functional data}.
\bjournal{Journal of Statistical Theory and Practice}
\bvolume{4}
\bpages{571-587}.
\bdoi{10.1080/15598608.2010.10412005}
\end{barticle}
\endbibitem

\bibitem[\protect\citeauthoryear{Zhang and Liang}{2014}]{JinTingZhangSJS2014}
\begin{barticle}[author]
\bauthor{\bsnm{Zhang},~\bfnm{Jin-Ting}\binits{J.-T.}} \AND
  \bauthor{\bsnm{Liang},~\bfnm{Xuehua}\binits{X.}}
(\byear{2014}).
\btitle{{One-way ANOVA for functional data via globalizing the pointwise
  F-test}}.
\bjournal{Scandinavian Journal of Statistics}
\bvolume{41}
\bpages{51-71}.
\bdoi{https://doi.org/10.1111/sjos.12025}
\end{barticle}
\endbibitem

\bibitem[\protect\citeauthoryear{Zhang and Wang}{2016}]{zhang2016sparse}
\begin{barticle}[author]
\bauthor{\bsnm{Zhang},~\bfnm{Xiaoke}\binits{X.}} \AND
  \bauthor{\bsnm{Wang},~\bfnm{Jane-Ling}\binits{J.-L.}}
(\byear{2016}).
\btitle{From sparse to dense functional data and beyond}.
\bjournal{The Annals of Statistics}
\bvolume{44}
\bpages{2281-2321}.
\bdoi{10.1214/16-AOS1446}
\end{barticle}
\endbibitem

\bibitem[\protect\citeauthoryear{Zhou, Yao and
  Zhang}{2022}]{zhou_functional_2022}
\begin{barticle}[author]
\bauthor{\bsnm{Zhou},~\bfnm{Hang}\binits{H.}},
  \bauthor{\bsnm{Yao},~\bfnm{Fang}\binits{F.}} \AND
  \bauthor{\bsnm{Zhang},~\bfnm{Huiming}\binits{H.}}
(\byear{2022}).
\btitle{{Functional linear regression for discretely observed data: from ideal
  to reality}}.
\bjournal{Biometrika}
\bvolume{110}
\bpages{381-393}.
\bdoi{10.1093/biomet/asac053}
\end{barticle}
\endbibitem

\bibitem[\protect\citeauthoryear{Zhu, Yao and Zhang}{2014}]{zhu2014structured}
\begin{barticle}[author]
\bauthor{\bsnm{Zhu},~\bfnm{Hongxiao}\binits{H.}},
  \bauthor{\bsnm{Yao},~\bfnm{Fang}\binits{F.}} \AND
  \bauthor{\bsnm{Zhang},~\bfnm{Hao~Helen}\binits{H.~H.}}
(\byear{2014}).
\btitle{Structured functional additive regression in reproducing kernel
  {Hilbert} spaces}.
\bjournal{Journal of the Royal Statistical Society Series B: Statistical
  Methodology}
\bvolume{76}
\bpages{581--603}.
\end{barticle}
\endbibitem

\end{thebibliography}

\end{document}